%% file: p390-xin.tex
\documentclass{vldb}

\usepackage{etex}
\usepackage{booktabs} 
\usepackage[usenames, dvipsnames]{color}
\usepackage[table]{xcolor}
\usepackage{hyperref}
\usepackage{listings}
\usepackage{multirow}
\usepackage{enumitem}
\usepackage{textcomp}
\usepackage[T1]{fontenc}
\usepackage{upquote}
\usepackage{times}
\usepackage{balance}
\usepackage{tcolorbox}

\usepackage{amsthm}

\usepackage{subcaption}
\usepackage{float}
\usepackage[font=small]{caption}
\usepackage{amsmath,amssymb,amsfonts,graphicx,nicefrac,mathtools}
\usepackage{thmtools,thm-restate}
\usepackage{xspace}
\usepackage{soul}
\usepackage[normalem]{ulem}
\usepackage{array}
\usepackage{cleveref}
\usepackage{tablefootnote}
\usepackage{makecell}
\usepackage[long]{optidef}

\usepackage{tikz}
\usepackage{pgf}
\usetikzlibrary{positioning,shapes,shadows,arrows,snakes,automata,backgrounds,petri}

\usepackage{syntax}

\usepackage[linesnumbered,ruled,resetcount,noline, noend]{algorithm2e}

\newcolumntype{C}[1]{>{\centering\arraybackslash}m{#1}}
\newcolumntype{M}[1]{>{\arraybackslash}m{#1}}

\input{vldbreffix.tex}

\input{macros.tex}

\input{notation.tex}

\xspaceaddexceptions{[]\}(}

\newenvironment{denselist}{
    \begin{list}{\tiny{$\bullet$}}%
    {\setlength{\itemsep}{0ex} \setlength{\topsep}{0ex}
    \setlength{\parsep}{0pt} \setlength{\itemindent}{0pt}
    \setlength{\leftmargin}{1.5em}
    \setlength{\partopsep}{0pt}}}%
    {\end{list}}

\newcommand{\topic}[1]{\vspace{-5pt}\smallskip \smallskip \noindent{\bf #1.}}
\newcommand{\stitle}[1]{\vspace{0.25em}\noindent\textbf{#1}}
\newcommand{\emtitle}[1]{\vspace{0.25em}\noindent{\em #1}}
\newcommand{\frameme}[1]{
\vspace{1pt}
\noindent\fbox{
  \parbox{0.95\linewidth}{
    \noindent #1
    }
  }
\vspace{2pt}
}
\newcommand{\hidden}[1]{}

\newcommand{\eat}[1]{}
\newcommand{\papertext}[1]{\ignorespaces}
\newcommand{\techreport}[1]{#1}

\makeatletter
\newcommand{\thickhline}{%
    \noalign {\ifnum 0=`}\fi \hrule height 1pt
    \futurelet \reserved@a \@xhline
}
\newcolumntype{"}{@{\hskip\tabcolsep\vrule width 1pt\hskip\tabcolsep}}
\makeatother

\newcommand{\code}[1]{\texttt{#1}}
\newcommand{\gestalt}{{\sc Helix}\xspace}
\newcommand{\helix}{{\sc Helix}\xspace}
\newcommand{\gs}{\gestalt}
\newcommand{\name}{\helix}
\newcommand{\lang}{HML\xspace}
\newcommand{\wf}{\texttt{Workflow}\xspace}
\newcommand{\opt}{\name \textsc{Opt}\xspace}
\newcommand{\am}{\name \textsc{AM}\xspace}
\newcommand{\nm}{\name \textsc{NM}\xspace}

\newcommand{\dcsu}{DC$_{SU}$\xspace}

\newcommand{\dce}{DC$_{E}$\xspace}

\newtheorem{example}{Example}
\crefname{example}{Example}{Examples}

\makeatletter
\def\thm@space@setup{\thm@preskip=2pt
\thm@postskip=2pt}
\makeatother
\newtheoremstyle{newstyle}      
{} 
{} 
{\itshape} 
{} 
{\bfseries} 
{.} 
{ } 
{} 
\theoremstyle{newstyle}

\newtheorem{theorem}{Theorem}
\crefname{theorem}{Theorem}{Theorems}
\newtheorem{lemma}{Lemma}
\crefname{lemma}{Lemma}{Lemmas}

\crefname{corollary}{Corollary}{Corollaries}
\newtheorem{constraint}{Constraint}
\crefname{constraint}{Constraint}{Constraints}

\crefname{invariant}{Invariant}{Invariants}

\newtheorem{definition}{Definition}
\crefname{definition}{Definition}{Definitions}
\newtheorem{problem}{Problem}
\crefname{problem}{Problem}{Problems}

\newcommand{\dataprep}{data preprocessing\xspace}

\newcommand{\latency}{run time\xspace}

\newcommand{\mf}{{\sc Max-Flow}\xspace}

\vldbTitle{{\name}: Holistic Optimization for Accelerating Iterative Machine Learning}
\vldbAuthors{Doris Xin, Stephen Macke, Litian Ma, Jialin Liu, Shuchen Song, Aditya Parameswaran}
\vldbDOI{https://doi.org/10.14778/3297753.3297763}
\vldbVolume{12}
\vldbNumber{4}
\vldbYear{2018}

\begin{document}

\title{{\huge \name}: Holistic Optimization for \\ Accelerating Iterative Machine Learning}

\author{Doris Xin, Stephen Macke, Litian Ma, Jialin Liu, Shuchen Song, Aditya Parameswaran \\
\affaddr{University of Illinois (UIUC)} \\
\affaddr{\{dorx0,smacke,litianm2,jialin2,ssong18,adityagp\}}@illinois.edu\;}

\maketitle

\begin{abstract}
\input{abstract}
\end{abstract}

\section{Introduction}
\label{sec:intro}
\input{intro}

\section{Background and Overview}
\label{sec:overview}
\input{overview}

\section{Programming Interface} 
\label{sec:interface}
\input{interface}

\section{Compilation and Representation}
\label{sec:representation}
\input{representation}

\section{Optimization}
\label{sec:optimization}
\input{optimization}

\section{Empirical Evaluation}
\label{sec:experiments}
\input{experiments}

\section{Related Work}
\label{sec:related}
\input{related2}

\balance
\section{Conclusions and Future Work}
\label{sec:conclusion}
\input{conclusion}

\newpage
\bibliographystyle{abbrv} 
\bibliography{dml.bib}

\techreport{\pagebreak 
\appendix
\input{appendix.tex}}
\end{document}

%% file: macros.tex
\newcommand{\BIT}{\begin{itemize}}
\newcommand{\EIT}{\end{itemize}}
\newcommand{\BNUM}{\begin{enumerate}}
\newcommand{\ENUM}{\end{enumerate}}
\newcommand\mbb[1]{\mathbb{#1}}

\def\reals{\mathbb{R}} 
\def\bigo#1{\mathcal{O}\left(#1\right)} 
\def\indic#1{\mbb{I}\left\{{#1}\right\}} 
\providecommand{\argmax}{\mathop\mathrm{arg max}} 
\providecommand{\argmin}{\mathop\mathrm{arg min}}

\def\P{\mathbb{P}} 

\renewcommand{\emptyset}{\varnothing}

\newcommand{\revision}[1]{\textcolor{black}{#1}}

%% file: notation.tex
\newcommand{\jobAi}{X_{a_i}}
\newcommand{\jobAj}{X_{a_j}}
\newcommand{\jobBi}{X_{b_i}}

%% file: abstract.tex

Machine learning workflow development is a process of
trial-and-error: developers   {\em iterate} on 
workflows by testing out small modifications until the
desired accuracy is achieved. 
Unfortunately, existing machine learning 
systems 
focus narrowly on model training---a small fraction of the overall development time---and
neglect to address iterative development.
We propose \name, a machine learning system
that  {\em optimizes 
the execution across iterations}---intell\-igently
caching and reusing, or recomputing intermediates as appropriate.
\name captures a wide variety of application needs within its Scala DSL,
with succinct syntax
defining unified processes for \dataprep, model specification, and learning.
We demonstrate that the reuse problem 
can be cast as a {\sc Max-Flow} problem, 
while the caching problem is {\sc NP-Hard}. 
We develop effective light\-weight heuristics for the latter. 
Empirical evaluation shows
that \name is not 
only able to handle a wide variety of use
cases in one unified workflow 
but also much faster, 
providing \latency reductions 
of up to $\mathbf{19\times}$
over state-of-the-art systems, 
such as DeepDive or KeystoneML, 
on four real-world applications 
in natural language processing, computer vision, social and natural sciences.

%% file: intro.tex

From emergent applications like 
precision medicine, voice-cont\-rolled devices,
and driverless cars,
to well-established ones like product recommendations
and credit card fraud detection,
machine learning continues to be the key
driver of innovations that are transforming our everyday lives.
At the same time, 
developing machine learning applications
is time-consuming and cumbersome.
To this end, a number of efforts
attempt to make machine learning more
declarative and to speed up the model training process~\cite{boehm2016declarative}.

However, the majority of the development time is in fact
spent {\em iterating on the machine learning workflow} 
by incrementally modifying steps within, 
including (i) {\em preprocessing}: altering data cleaning or extraction, or engineering features;
(ii) {\em model training}: tweaking hyperparameters,
or changing the objective or learning algorithm;
and (iii) {\em postprocessing}: evaluating with new data,  
or generating additional statistics or visualizations.
These iterations are necessitated by the difficulties in predicting
the performance of a workflow {\em a priori}, due
to both the variability of data and the complexity and unpredictability
of machine learning.
{\bf \em Thus, developers must resort to iterative modifications
of the workflow via ``trial-and-error'' to improve performance.}
A recent survey reports that less than 15\% of development time
is actually spent on model training~\cite{munson2012study},
with the bulk of the time spent iterating on  
the machine learning workflow.

\begin{example}[Gene Function Prediction]
Consider the following example from our bioinformatics
collaborators who form part of a genomics center
at the University of Illinois~\cite{ren2017life}.
Their goal is to discover novel relationships between genes and diseases
 by mining scientific literature.
To do so, they process
published papers to extract entity---gene and disease---mentions,
compute embeddings using an approach like word2vec~\mbox{\cite{mikolov2013distributed}},
and finally cluster the embeddings to find related entities.
They repeatedly iterate on this workflow 
to improve the quality of the relationships discovered as assessed by collaborating clinicians.
For example, they may (i) expand or shrink the literature corpus,
(ii) add in external sources such as gene databases to refine how entities are identified, and
(iii) try different NLP libraries for tokenization and entity recognition. 
They may
also (iv) change the algorithm used for computing word embedding vectors, 
e.g., from word2vec to LINE~\cite{tang2015line}, or
(v) tweak the number of clusters to control the granularity of the clustering.
Every single change that they make necessitates waiting for
the entire workflow to rerun from scratch---often
multiple hours on a large server for each single change,
even though the change may be quite small.
\label{ex:gene}
\end{example}

\noindent As this example illustrates, 
the key bottleneck in applying
machine learning is {\bf \em iteration}---{\em every
change to the workflow
requires hours of recomputation from scratch,
even though the change may only impact a small 
portion of the workflow}. 
For instance, 
normalizing a feature, or changing
the regularization would not impact
the portions of the workflow that do not depend on it---and 
yet the current approach is to simply rerun from scratch.

One approach to address the expensive recomputation issue
is for developers to
explicitly materialize all intermediates
that do not change across iterations, but this requires
writing code to handle materialization 
and to reuse materialized results by identifying changes between iterations.
Even if this were a viable option, 
materialization of all intermediates
is extremely wasteful, and figuring out the optimal reuse of materialized
results is not straightforward.
Due to the cumbersome and inefficient nature of this approach,
developers often opt to rerun the entire workflow from scratch.

Unfortunately, existing machine learning systems
do not optimize for rapid iteration. 
For example, KeystoneML~\cite{sparks2016end}, 
which allows developers to specify workflows
at a high-level abstraction, 
only optimizes the one-shot execution of workflows
by applying techniques such as
common subexpression elimination
and intermediate result caching.
On the other extreme, DeepDive~\cite{zhang2015deepdive},
targeted at knowledge-base construction, 
materializes the results of all of the feature extraction
and engineering steps, while also applying 
approximate inference to speed up model training.
Although this na{\"i}ve
materialization approach 
does lead to reuse in iterative executions,
it is wasteful and time-consuming.

We present \name, a {\em declarative, general-purpose
machine learning system that optimizes across iterations}.
\name is able to match or exceed the performance of
KeystoneML and DeepDive on one-shot execution,
while providing {\bf \em gains of up to 19$\times$}
on iterative execution across four real-world applications. 
By optimizing across iterations, \name allows data scientists
to avoid wasting time running the workflow from scratch
every time they make a change and instead
run their workflows   
in time proportional to the complexity
of the change made. 
\name is able to thereby substantially 
increase developer productivity 
\revision{while simultaneously lowering resource consumption.}

Developing \name involves two types of challenges---challenges in 
{\em iterative execution optimization} and challenges in {\em specification and generalization}.

\stitle{Challenges in Iterative Execution Optimization.} 
A machine learning workflow can be represented as a directed acyclic graph,
where each node corresponds to a collection
of data---the original data items, such as documents or images,
the transformed data items, such as sentences or words,
the extracted features, or the final outcomes. 
This graph, for practical workflows, can be quite large and complex.
One simple approach 
to enable iterative execution optimization (adopted by DeepDive) is to
materialize every single node, such that the next
time the workflow is run, we can simply check if the result can be reused
from the previous iteration, and if so, reuse it.
Unfortunately, this approach is not only wasteful in storage
but also potentially very time-consuming due to materialization overhead.
Moreover, in a subsequent iteration, it may be cheaper to 
recompute an intermediate result, as opposed to reading it from disk.

A better approach is to determine whether
a node is worth materializing
by considering both the time taken for computing a node 
and the time taken for computing its ancestors.
Then, during subsequent iterations, we can determine whether
to read the result for a node from persistent storage 
(if materialized), which could lead to large portions of the graph being pruned,
or to compute it from scratch.
In this paper, we prove that 
the reuse plan problem is in {\sc PTIME} via
a non-trivial {\bf \em reduction to {\sc Max-Flow} using
the {\sc Project Selection Problem}~\cite{kleinberg2006algorithm},
while the materialization problem is, in fact, {\sc NP-Hard}}.

\stitle{Challenges in Specification and Generalization.}
To enable iterative execution optimization, we need to support 
the specification of the end-to-end machine learning workflow in a high-level
language.
\revision{This is challenging because \dataprep can vary greatly across applications, 
often requiring ad hoc code involving
complex composition of declarative statements 
and UDFs~\cite{armbrust2015sparksql},}
making it hard to automatically
analyze the workflow 
to apply holistic iterative execution optimization.

We adopt a hybrid approach within \name:
developers specify their workflow in an {\bf \em intuitive, high-level domain-specific language (DSL)
in Scala} (similar to existing systems like KeystoneML),
using {\bf \em imperative code as needed for UDFs},
say for feature engineering. 
This interoperability allows developers
to seamlessly integrate existing JVM machine learning libraries~\cite{team2016deeplearning4j, rasband2012imagej}.
Moreover, \name is built on top of Spark,
allowing data scientists to 
leverage Spark's parallel processing capabilities. 
We have developed a GUI on top of the \name DSL to further facilitate development~\mbox{\cite{demo}}.

\name's DSL 
not only enables automatic identification of data dependencies and data flow, 
but also encapsulates
all typical machine learning workflow designs.
\techreport{Unlike DeepDive~\cite{zhang2015deepdive}, 
\name is not restricted to regression or factor graphs, 
allowing data scientists to use the most suitable model for their tasks.}
\revision{All of the functions in Scikit-learn's (a popular ML toolkit)
can be mapped to functions in the DSL~\cite{dorx2017},}
allowing \name to easily capture applications ranging from natural language processing,
to knowledge extraction, to computer vision. 
Moreover, by studying the variation in the dataflow graph across iterations, \name
is able to identify reuse opportunities across iterations.
Our work is a first step in a broader agenda to improve human-in-the-loop
ML~\cite{deem}.

\stitle{Contributions and Outline.} 
The rest of the paper is organized as follows: 
\Cref{sec:overview} presents 
\techreport{a quick recap of ML workflows, 
statistics on how users iteration on ML workflows collected from applied ML literature,}
an architectural overview of the system,
and a concrete workflow to illustrate concepts discussed in the subsequent sections;
\Cref{sec:interface} describes the programming interface\techreport{ for effortless end-to-end workflow specification};
\Cref{sec:representation} discusses \name system internals, 
including the workflow DAG generation
and change tracking between iterations;
\Cref{sec:optimization} formally presents the two major optimization problems 
in accelerating iterative ML and \name's solution to both problems.
We evaluate our framework on four workflows from different applications domains 
and against two state-of-the-art systems in
\Cref{sec:experiments}. 
We discuss related work in \Cref{sec:related}.

%% file: overview.tex

In this section, we provide a brief overview of machine learning workflows,
describe the \name system architecture
and present a sample workflow in \name
that will serve as a running example.

\revision{A machine learning (ML)
workflow accomplishes a specific ML task,
 ranging from simple ones like classification or clustering,
 to complex ones like entity resolution or image captioning.
Within \name, we decompose ML workflows into three components:
\dataprep~(DPR), where raw data is transformed into ML-compatible representations, 
learning/inference (L/I), where ML models are trained and used to perform inference on new data, 
and postprocessing (PPR), where learned models and inference results are processed to obtain summary metrics, create dashboards, and power applications.
We discuss specific operations in each of these components in Section~\ref{sec:interface}.
As we will demonstrate, these three components are generic and 
sufficient for describing a wide variety of supervised, semi-supervised, and unsupervised settings.}

\subsection{System Architecture}
\input{overview-arch}

\subsection{Example Workflow}
\input{overview-example}

%% file: overview-arch.tex

\label{sec:overview-arch-example}
\begin{figure}[t]
\centering
\includegraphics[width=0.3\textwidth]{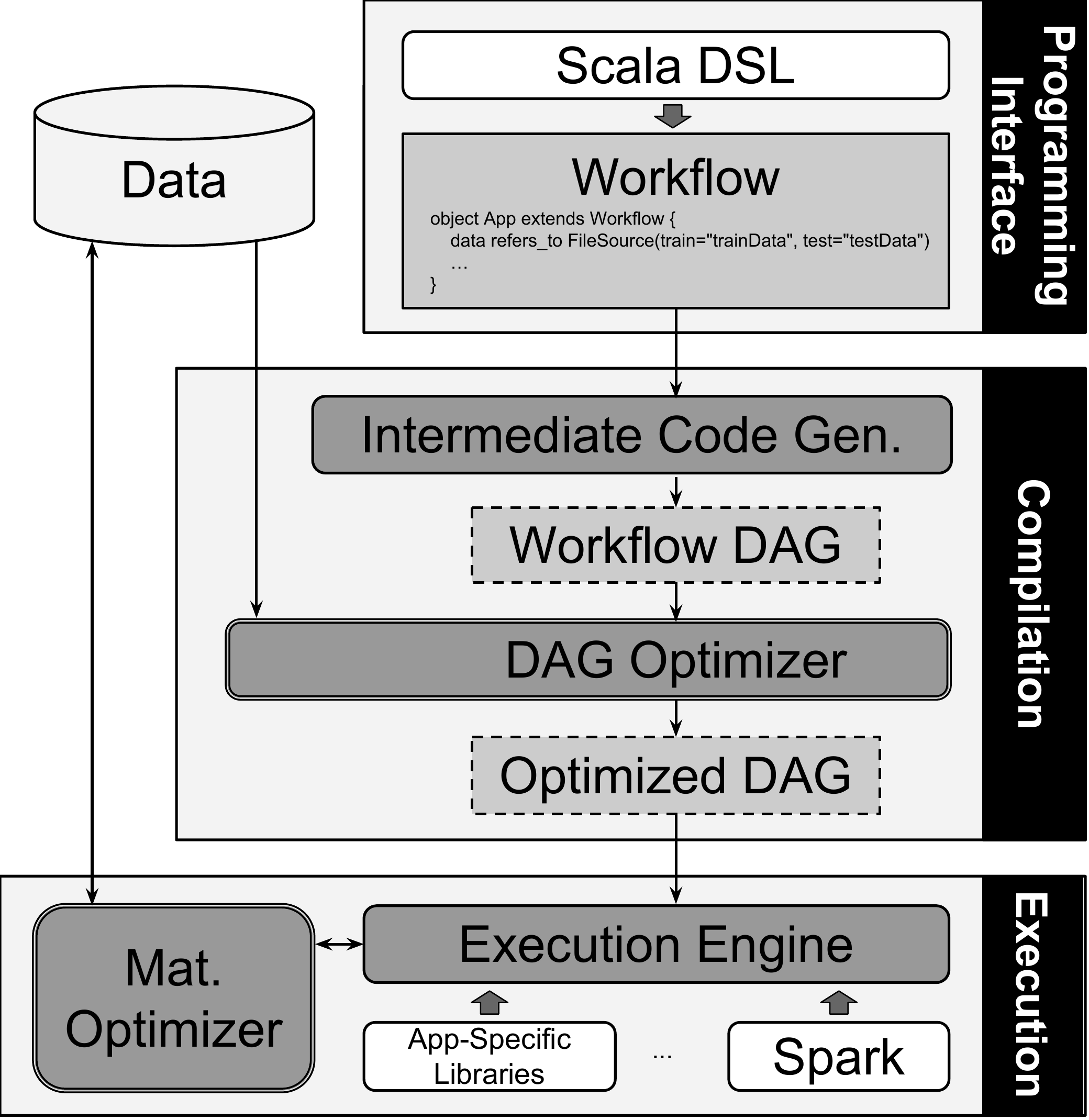}
\caption{\name System architecture. 
\techreport{A program written by the user in the \name DSL, known as a \wf, 
is first compiled into an intermediate DAG representation,
which is optimized to produce a physical plan 
to be run by the execution engine. At runtime, the execution engine 
selectively materializes
intermediate results to disk.}}
\label{fig:overview}
\end{figure}

The \name system consists of 
a domain specific language (DSL) 
in Scala as the programming interface, 
a compiler for the DSL, 
and an execution engine, as shown in Figure~\ref{fig:overview}. 
The three components work collectively to 
\textit{minimize the execution time for both the 
current iteration and subsequent iterations}:

\stitle{1. Programming Interface (Section~\ref{sec:interface}).} 
\name provides a single Scala interface named \wf
for programming the entire workflow; the \name DSL also enables
embedding of imperative code in declarative statements. 
Through just a handful of extensible operator types, 
the DSL supports a wide range of use cases for 
both \dataprep and machine learning. 

\stitle{2. Compilation (Sections~\ref{sec:representation}, ~\ref{sec:optPrelim}--\ref{sec:oep}).}
A \wf is internally represented as a 
directed acyclic graph (DAG) 
of operator outputs.
The DAG is compared to the one in previous
iterations to determine reusability (Section~\ref{sec:representation}).
The {\em DAG Optimizer} uses this information
to produce an optimal {\em physical execution plan}
that {\em minimizes the one-shot runtime of the workflow},
by selectively loading previous results
via a \mf-based algorithm (Section~\ref{sec:optPrelim}--\ref{sec:oep}).

\stitle{3. Execution Engine (Section~\ref{sec:omp}).} 
The execution engine carries 
out the physical plan produced during the compilation phase,
while communicating with the {\em materialization operator}
to materialize intermediate results, 
to {\em minimize runtime of future executions}.
The execution engine uses Spark~\cite{zaharia2012resilient} 
for data processing 
and domain-specific libraries 
\techreport{ such as CoreNLP~\cite{manning2014stanford} and Deeplearning4j~\cite{dd2deeplearning4j} }
for custom needs.
\revision{
\name defers operator pipelining and scheduling for asynchronous execution to Spark.
\techreport{Operators that can run concurrently are invoked in an arbitrary order, 
executed by Spark via Fair Scheduling.
While by default we use Spark in the batch processing mode,
it can be configured to perform stream processing using the same APIs as batch.
We discuss optimizations for streaming in Section~\ref{sec:optimization}.}
}

\subsection{The Workflow Lifecycle}
\label{sec:lifecycle}
\begin{figure}[h]
\centering
\includegraphics[width=.4\textwidth]{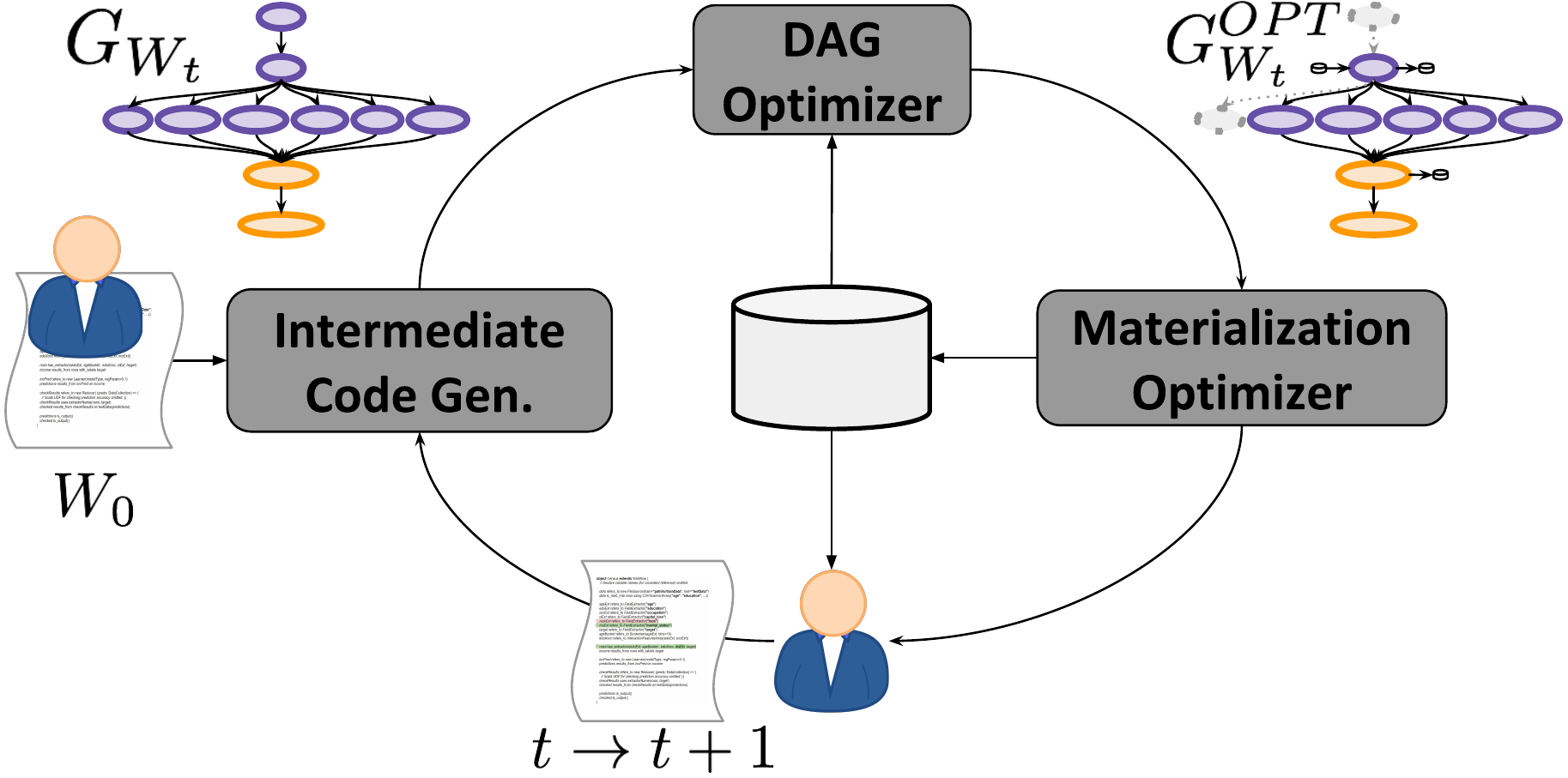}
\vspace{4pt}
\caption{Roles of system components in the \name workflow lifecycle.
}
\label{fig:lifecycle}
\end{figure}

\revision{Given the system components described in the previous section, 
\Cref{fig:lifecycle} illustrates how they fit into the lifecycle of ML workflows.}
Starting with $W_0$, an initial version of the workflow,
the lifecycle includes the following stages:
\begin{denselist}
\item \textbf{DAG Compilation.} 
The \wf $W_t$ is compiled into a DAG $G_{W_t}$ of operator outputs.
\item \textbf{DAG Optimization.}
The DAG optimizer creates a physical plan $G_{W_t}^{OPT}$ 
to be executed
by pruning and ordering the nodes in $G_{W_t}$
and deciding whether any computation can be replaced 
with loading previous results from disk.
\item \textbf{Materialization Optimization.} 
During execution, the materialization optimizer determines
which nodes in $G_{W_t}^{OPT}$ should be persisted to disk
for future use. 
\item \textbf{User Interaction.} 
Upon execution completion, the user may modify the workflow
from $W_t$ to $W_{t+1}$ based on the results. 
The updated workflow $W_{t+1}$ fed back to \name marks 
the beginning of a new iteration,
and the cycle repeats.
\end{denselist}

\techreport{Without loss of generality, 
we assume that a workflow $W_t$ is only executed once in each iteration.
We model a repeated execution of $W_t$ as a new iteration 
where $W_{t+1} = W_t$.
Distinguishing two executions of the same workflow is important
because they may have different run times---the second execution 
can reuse results materialized in the first execution for a potential run time reduction.}

%% file: overview-example.tex

\label{sec:exWf}
\begin{figure*}
\centering
\includegraphics[width=0.9\textwidth]{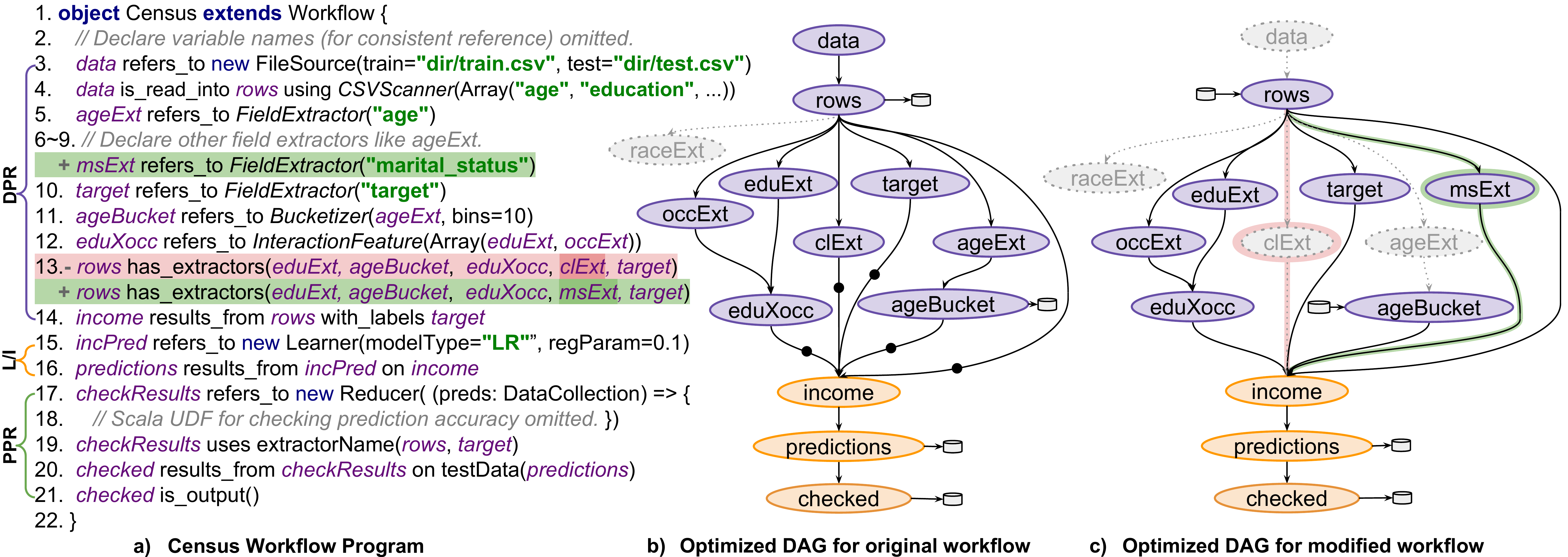}
\vspace{4pt}
\caption{Example workflow for predicting income from census data.}
\label{fig:wfEx}
\end{figure*}

We demonstrate the usage of \name with 
a simple example ML workflow for predicting 
income using census data from Kohavi~\cite{kohavi1996scaling}, 
shown in Figure~\ref{fig:wfEx}a)\techreport{;
this workflow will serve as a running example throughout the paper.
Details about the individual operators will be 
provided in subsequent sections}.
We overlay the original workflow with an iterative update,
with additions annotated with + and deletions annotated with $-$,
while the rest of the lines are retained as is. 
We begin by describing the original workflow
consisting of all the unannotated lines plus
the line annotated with $-$ (deletions).
\papertext{Additional details can be found in our 
technical report~\cite{dorx2017}.}

\papertext{
\stitle{Original Workflow.}
\revision{
A data collection \code{rows} is created
from a data source \code{data} containing both training
and test CSV files. 
The user declares simple features such as \code{age} and \code{education}
corresponding to each column (lines 5--10), 
as well as discretized features like \code{ageBucket} (line 11),
and interaction features for higher-order patterns (line 12).
These features are then
associated with \code{rows} (line 13),
and used to create collection of examples \code{income} 
using \code{target} as labels (line 14).
\techreport{Overall, the user does not need to worry about
the physical representation of features as concrete Scala objects.}
For L/I, the user declares an ML model named \code{incPred} (line 15),
to be evaluated on all of the data in \code{income} (line 16).
Subsequently, the user creates a
new UDF named \code{checkResults} to check the prediction accuracy (line 17-18). 
Finally, the user declares that the output scalar named \code{checked}
is to be computed from the test data in \code{income} (line 20) and is part of the output (line 21).
}
}
\techreport{
\stitle{Original Workflow: DPR Steps.}
First, after some variable name declarations,  
the user defines in line 3-4 a data collection \code{rows}
read from a data source \code{data} consisting of two CSV files,
one for training and one for test data,
and names the columns of the CSV files \code{age, education}, etc.
In lines 5-10, the user declares simple features
that are values from specific named columns. 
Note that the user is not required to specify the feature type,
which is automatically inferred by \name from data.
In line 11 \code{ageBucket} is declared as a derived feature
formed by discretizing age into ten buckets (whose boundaries are computed by \name),
while line 12 declares an interaction feature, 
commonly used to capture higher-order patterns, 
formed out of the concatenation of \code{eduExt} and \code{occExt}.

Once the features are declared, 
the next step, line 13, declares the features to be extracted from 
and associated with each element of \code{rows}.
Users do not need to worry about 
how these features are attached and propagated; 
users are also free to perform manual feature selection here, 
studying the impact of various feature combinations, 
by excluding some of the feature extractors.
Finally, as last step of data preprocessing, 
line 14 declares that an example collection named \code{income} 
is to be made from \code{rows}
using \code{target} as labels.
Importantly, this step converts the features 
from human-readable formats (e.g., color=red)
into an indexed vector representation required for learning.
} 

\techreport{
\stitle{Original Workflow: L/I \& PPR Steps.} 
Line 15 declares an ML model named \code{incPred} 
with type ``Logistic Regression'' and regularization parameter 0.1,
while line 16 specifies that \code{incPred} is to be 
learned on the training data in \code{income}
and applied on all data in \code{income} 
to produce a new example collection called \code{predictions}.
Line 17-18 declare a \textit{Reducer} named \code{checkResults}, 
which outputs a scalar using a UDF for computing prediction accuracy.
Line 19 explicitly specifies \code{checkResults}'s dependency on \code{target}
since the content of the UDF is opaque to the optimizer.
Line 20 declares that the output scalar named \code{checked} 
is only to be computed from the test data in \code{income}.
Lines 21 declares that \code{checked} 
must be part of the final output.
}

\stitle{Original Workflow: Optimized DAG.}
\papertext{
\revision{
The \name compiler translates the program in
Figure~\ref{fig:wfEx}a) into a DAG in Figure~\ref{fig:wfEx}b), which is then transformed 
by the optimizer by \techreport{adding in inferred
edges labeled with dots to link relevant features to the model,
and }eliminating \code{raceExt} since it does not contribute
to the model. 
DPR steps are in purple, while others are in orange.
Some of the intermediate results are materialized to disk (displayed with a drum).}}
\techreport{
The \name compiler first translates verbatim the program in Figure~\ref{fig:wfEx}a) into a DAG,
which contains all nodes including \code{raceExt}
and all edges (including the dashed edge) except the ones marked with dots in Figure~\ref{fig:wfEx}b).
This DAG is then transformed by the optimizer, 
which prunes away \code{raceExt} (grayed out) because it does not contribute to the output,
and adds the edges marked by dots to link relevant features to the model.
DPR involves nodes in purple,
and L/I and PPR involve nodes in orange.
Nodes with a drum to the right are materialized to disk,
either as mandatory output 
or for aiding in future iterations.
}

\stitle{Updated Workflow: Optimized DAG.}
In the updated version of the workflow, 
a new feature named \code{msExt} is added (below line 9), and 
\code{clExt} is removed (line 13);
correspondingly, in the updated DAG, a new node is added for 
\code{msExt} (green edges), while \code{clExt} gets pruned (pink edges). 
In addition, \name chooses to load
materialized results for \code{rows} from the previous iteration
allowing \code{data} to be pruned, avoiding a costly parsing step.
\name also loads \code{ageBucket} instead
of recomputing the bucket boundaries requiring a full scan. 
\name materializes \code{predictions} in both iterations
since it has changed.
\papertext{
In this instance, \name is able to prune 
the computation of \code{ageBucket} 
and that of \code{rows} from \code{data}.
}
\techreport{Although \code{predictions} is not reused in the updated workflow,
its materialization has high expected payoff over iterations
because PPR iterations (changes to \code{checked} in this case) are the most common
as per our survey results shown in Figure~\ref{fig:survey}(c). 
This example illustrates that 
\begin{denselist}
\item Nodes selected for materialization lead to significant speedup in subsequent iterations.
\item \name reuses results safely, deprecating old results when changes are detected (e.g., \code{predictions} is not reused because of the model change).
\item \name correctly prunes away extraneous operations via dataflow analysis.
\end{denselist}}

%% file: interface.tex

To program ML workflows 
with high-level abstractions, \name users program in a language called \lang, 
an \textit{embedded DSL} in Scala. 
An embedded DSL exists as a library in the host language (Scala in our case), leading to seamless integration. 
LINQ~\cite{linq}, a data query framework 
integrated in .NET languages, is another example of an embedded DSL.
In \name, users can freely incorporate Scala code for user-defined functions (UDFs) directly into \lang.
JVM-based libraries 
can be imported directly into \lang to support application-specific needs.
Development in other languages can be supported 
with wrappers in the same style as PySpark~\cite{rosenpyspark}.

\revision{\subsection{Operations in ML Workflows}
\label{sec:opBasis}
}

\techreport{
\revision{In this section, we argue that common operations in ML workflows 
can be decomposed into a small set of \textit{basis functions} $\mathcal{F}$.
We first introduce $\mathcal{F}$ and then enumerate its mapping onto operations in Scikit-learn~\cite{pedregosa2011scikit},
one of the most comprehensive ML libraries,
thereby demonstrating coverage.
In Section~\ref{sec:hml}, we introduce \lang, which implements the capabilities offered by $\mathcal{F}$.}

\revision{As mentioned in Section~\ref{sec:overview}, an ML workflow consists of three components:
data preprocessing (DPR), learning/inference (L/I), and postprocessing (PPR). 
They are captured by the \textit{Transformer}, 
\textit{Estimator}, and \textit{Predictor} interfaces in Scikit-learn, respectively.
Similar interfaces can be found in many ML libraries, 
such as MLLib~\cite{meng2016mllib}, TFX~\cite{baylor2017tfx}, and KeystoneML.}

\topic{\revision{Data Representation}}
\revision{Conventionally, the input space to ML, $\mathcal{X}$, 
is a $d$-dimensional vector space, $\mathbb{R}^d, d\geq 1$,
where each dimension corresponds to a feature.
Each datapoint is represented by a feature vector (FV), $\mathbf{x} \in \mathbb{R}^d$.
For notational convenience, we denote a $d$-dimensional FV,  
$\mathbf{x} \in \mathbb{R}^d$,  as $\mathbf{x}^d$.
While inputs in some applications can be easily loaded into FVs, 
e.g., images are 2D matrices that can be flattened into a vector,
many others require more complex transformations, 
e.g., vectorization of text requires tokenization and word indexing.
We denote the input dataset of FVs to an ML algorithm as $\mathcal{D}$.}

\topic{\revision{DPR}}
\revision{
The goal of DPR is to transform raw input data into $\mathcal{D}$.
We use the term \textit{record}, denoted by $r$, 
to refer to a data object in formats incompatible with ML,
such as text and JSON, 
requiring preprocessing.
Let $\mathcal{S} = \{r\}$ be a data source, e.g., a csv file, or a collection of text documents.
DPR includes transforming records
from one or more data sources
from one format to another or into FVs $\reals^{d'}$;
as well as feature transformations (from $\reals^d$ to $\reals^{d'}$).
DPR operations can thus be decomposed into the following categories:
\begin{denselist}
\item \textit{Parsing} $r \mapsto (r_1, r_2, \ldots)$: transforming a record into a set of records, e.g., parsing an article into words via \textit{tokenization}.
\item \textit{Join} $(r_1, r_2, \ldots) \mapsto r$: combining multiple records into a single record,
where $r_i$ can come from different data sources. 
\item \textit{Feature Extraction} $r \mapsto \mathbf{x}^d$: extracting features from a record.
\item \textit{Feature Transformation} $T: \mathbf{x}^d \mapsto \mathbf{x}^{d'}$: deriving a new set of features from the input features.
\item \textit{Feature Concatenation} $(\mathbf{x}^{d_1}, \mathbf{x}^{d_2}, \ldots) \mapsto \mathbf{x}^{\sum_i d_i}$: concatenating features extracted in separate operations to form an FV.
\end{denselist}
Note that sometimes these functions need to be \textit{learned} from the input data.
For example, discretizing a continuous feature $x_i$ into four even-sized bins 
requires the distribution of $x_i$, 
which is usually estimated empirically by collecting all values of $x_i$ in $\mathcal{D}$.
We address this use case along with L/I next.
}

\vspace{2pt}
\topic{\revision{L/I}}
\revision{
At a high-level, L/I is about learning a function $f$ from the input $\mathcal{D}$,
where $f: \mathcal{X} \rightarrow \mathbb{R}^{d'}, d' \geq 1$.
This is more general than learning ML models,
and also includes feature transformation functions mentioned above.
The two main operations in L/I are 
1) \textit{learning}, which produces functions using data from  $\mathcal{D}$,
and 2) \textit{inference}, which uses the function obtained from learning 
to draw conclusions about new data. 
Complex ML tasks can be broken down into simple learning steps captured by these two operations,
e.g., image captioning can be broken down into object identification via classification, 
followed by sentence generation using a language model~\cite{karpathy2015deep}.
Thus, L/I can be decomposed into:
\begin{denselist}
\item \textit{Learning} $\mathcal{D} \mapsto f$: learning a function $f$ from the dataset $\mathcal{D}$.
\item \textit{Inference} $(\mathcal{D}, f) \mapsto \mathcal{Y}$:
using the ML model $f$ to infer feature values, i.e., \textit{labels}, $\mathcal{Y}$ from the input FVs in $\mathcal{D}$.
\end{denselist}
Note that labels 
can be represented as FVs like other features,
hence the usage of a single $\mathcal{D}$ in learning to represent both the training data and labels to unify the abstraction for both supervised and unsupervised learning
and to enable easy model composition. }

\topic{\revision{PPR}}
\revision{Finally, a wide variety of operations can take place in PPR, 
using the learned models and inference results from L/I as input, 
including model evaluation, data visualization, and other application-specific activities. 
The most commonly supported PPR operations in general purpose ML libraries are 
model evaluation and model selection, which can be represented by a computation
whose output does not depend on the size of the data $\mathcal{D}$. 
We refer to a computation with output sizes independent
of input sizes as a {\em reduce:}
\begin{denselist}
\item \textit{Reduce} $(\mathcal{D}, s') \mapsto s$: applying an operation on the input dataset $\mathcal{D}$ and $s'$, 
where $s'$ can be any non-dataset object.
For example, $s'$ can store a set of hyperparameters over which {\em reduce}
optimizes, learning various models and outputting $s$, which can represent
a function corresponding to the model with the best cross-validated hyperparameters.
\end{denselist}}

\revision{\subsubsection{Comparison with Scikit-learn}
\label{sec:scikitComp}
}
\input{sklearn-tables.tex}

\revision{A dataset in Scikit-learn is represented as a matrix of FVs, denoted by \code{X}.
This is conceptually equivalent to $\mathcal{D} = \{\mathbf{x}^d\}$ introduced earlier,
as the order of rows in \code{X} is not relevant. 
Operations in Scikit-learn are categorized into
dataset loading and transformations, learning, and model selection and evaluation~\cite{sklearnUG}.
\techreport{Operations like loading and transformations that do not tailor their behavior to
particular characteristics present in the dataset $\mathcal{D}$
map trivially onto the DPR basis functions $\in \mathcal{F}$ introduced at the start of
Section~\ref{sec:opBasis}, so we focus on comparing data-dependent DPR and L/I,
and model selection and evaluation.}}\papertext{\revision{The Scikit-learn
operations for DPR, L/I, and PPR, and
their corresponding equivalents in terms of functions from $\mathcal{F}$,
are summarized in Table~\ref{tab:sklearn-dpr-li-ppr}.
For a full discussion of Scikit-learn primitives and their coverage
in terms of $\mathcal{F}$, please see our technical report~\cite{dorx2017}.}}}

\techreport{\topic{\revision{Scikit-learn Operations for DPR and L/I}}
\revision{Scikit-learn objects for DPR and L/I implement one or more of the following
interfaces~\cite{sklearnApi}:
\begin{denselist}
	\item {\bf Estimator}, used to indicate that an operation has data-dependent behavior
	via a \code{fit(X[, y])} method, where \code{X}
	contains FVs or raw records, and \code{y} contains labels if the operation represents
	a supervised model.
	\item {\bf Predictor}, used to indicate that the operation may be used for inference via a
	\code{predict(X)} method, taking a matrix of FVs and producing predicted labels.
	Additionally, if the operation implementing Predictor is a classifier for which
	inference may produce raw floats (interpreted as probabilities),
	it may optionally implement \code{predict_proba}.
	\item {\bf Transformer}, used to indicate that the operation may be used for feature
	transformations via a \code{transform(X)} method, taking a matrix of FVs and producing
	a new matrix \code{X}$_{new}$.
\end{denselist}
An operation implementing
both Estimator and Predictor has a \code{fit_predict} method, and an operation
implementing both Estimator and Transformer has a \code{fit_transform} method,
for when inference or feature transformation, respectively, is applied immediately
after fitting to the data.
The rationale for providing a separate Estimator interface is likely due to the fact that
it is useful for both feature transformation and inference to have data-dependent behavior
determined via the result of a call to \code{fit}. For example, a useful data-dependent
feature transformation for a Naive Bayes classifier maps word tokens to positions in a
sparse vector and tracks word counts.
The position mapping will depend on the vocabulary represented in the raw training data.
Other examples of data-dependent transformations include feature scaling,
descretization, imputation, dimensionality reduction, and kernel transformations.}

\emtitle{\revision{Coverage in terms of basis functions $\mathcal{F}$.}}
\revision{The first part of Table~\ref{tab:sklearn-dpr-li-ppr} summarizes the mapping from Scikit-learn's
interfaces for DPR and L/I to (compositions of) basis functions from $\mathcal{F}$.
In particular, note that there is nothing special about Scikit-learn's use
of separate interfaces for inference (via Predictor) and data-dependent
transformations (via Transformer); the separation exists mainly to draw
attention to the semantic separation between DPR and L/I.}

\topic{\revision{Scikit-learn Operations for PPR}}
\revision{Scikit-learn interfaces for operations implementing
model selection and evaluation are not as standardized as
those for DPR and L/I. For evaluation, the typical strategy
is to define a simple function that compares model outputs with
labels, computing metrics like accuracy or $F_1$ score.
For model selection, the typical
strategy is to define a class that implements methods \code{fit} and \code{score}.
The \code{fit} method takes a set of hyperparameters over which to search,
with different models scored according to the \code{score} method (with
identical interface as for evaluation in Scikit-learn).
The actual model over which hyperparameter search is
performed is implemented by an Estimator that is passed into the
model selection operation's constructor.}

\emtitle{\revision{Coverage in terms of basis functions $\mathcal{F}$.}}
\revision{As summarized in the second part of Table~\ref{tab:sklearn-dpr-li-ppr}, Scikit-learn's operations
for evaluation may be implemented via compositions
of (optionally) {\em inference}, {\em joining}, and {\em reduce} $\in \mathcal{F}$.
Model selection may be implemented via a reduce that
internally uses learning basis functions to learn models
for the set of hyperparameters specified by $s'$,
followed by composition with inference and another reduce $\in \mathcal{F}$
for scoring, eventually returning
the final selected model.}}

\revision{
\subsection{\lang}
\label{sec:hml}
}

\revision{
\lang is a declarative language for specifying an ML workflow DAG.
The basic building blocks of \lang are {\em \gs objects}, 
which correspond to the nodes in the DAG.
Each \name object is either
a {\em data collection} (DC) or an {\em operator}.
Statements in \lang either declare new instances of objects 
or relationships between declared objects. 
Users program the entire workflow in a single \wf interface, 
as shown in Figure~\ref{fig:wfEx}a).
The complete grammar for \lang in Backus-Naur Form
as well as the semantics of all of the expressions 
can be found 
in the technical report~\cite{dorx2017}.
Here, we describe high-level concepts including DCs and operators 
and discuss the strengths and limitations of \lang in~\Cref{sec:hmlLim}.}

\subsubsection{Data Collections}
\label{sec:dc}
A \textit{data collection} (DC) is analogous to a relation in a RDBMS;
each \textit{element} in a DC is analogous to a tuple. 
The content of a DC either derives from disk, 
e.g., \code{data} in Line 3 in Figure~\ref{fig:wfEx}a), 
or from operations on other DCs, 
e.g., \code{rows} in Line 4 in Figure~\ref{fig:wfEx}a).
An element in a DC can either be a \textit{semantic unit},
\revision{the data structure for DPR},
or an \textit{example},
\revision{the data structure for L/I}.

A DC can only contain a single type of element.
\dcsu and \dce denote a DC of semantic units and a DC of examples, respectively.
The type of elements in a DC is determined by the operator that produced the DC 
and not explicitly specified by the user. 
We elaborate on the relationship between operators and element types 
in Section~\ref{sec:operator}, after introducing the operators.

\topic{Semantic units} 
\revision{
\techreport{Recall that many}
\papertext{Many} DPR operations require 
going through the entire dataset to learn the exact transformation or extraction function.
For a workflow with many such operations, 
processing $\mathcal{D}$ to learn each operator separately can be highly inefficient.
We introduce the notion of semantic units (SU) to compartmentalize the logical and physical representations of features,
so that the learning of DPR functions can be delayed and batched. 
}

\revision{
Formally, each SU contains an input $i$, which can be a set of records or FVs,
a pointer to a DPR function $f$, which can be of type parsing, join, feature extraction, feature transformation, or feature concatenation,
and an output $o$, which can be a set of records or FVs and is the output of $f$ on $i$.
The variables $i$ and $f$ together serve as the \textit{semantic}, or logical, representation of the features,
whereas $o$ is the lazily evaluated physical representation that can
only be obtained after $f$ is fully instantiated.
}

\topic{Examples}
\revision{Examples gather all the FVs contained in the output of various SUs into a single FV for learning.
Formally, an example contains a set of SUs $S$, 
and an optional pointer to one of the SUs 
whose output will be used as the label in supervised settings,
and an output FV, which is formed by concatenating the outputs of $S$.
\techreport{In the implementation,  
the order of SUs in the concatenation is determined globally across $\mathcal{D}$, and SUs whose outputs are not FVs are filtered out.}}

\topic{\revision{Sparse vs. Dense Features}}
\revision{
The combination of SUs and examples affords \name 
a great deal of flexibility in the physical representation of features.
Users can explicitly program their DPR functions to output dense vectors,
in applications such as computer vision. 
For sparse categorical features, they are kept in the raw key-value format 
until the final FV assembly, where they are transformed into sparse or dense vectors
depending on whether the ML algorithm supports sparse representations.
\techreport{Note that users do not have to commit to a single representation for the entire application,
since different SUs can contain different types of features. }
When assembling a mixture of dense and spare FVs, \name currently opts for a dense representation but can be extended to support optimizations considering space and time tradeoffs. 
}

\topic{Unified learning support}\label{para:unified}
\lang provides unified support for training and test data 
by treating them as a single DC, as done in Line 4 in Figure~\ref{fig:wfEx}a).
This design ensures that both training and test data undergo the exact same \dataprep steps,
eliminating bugs caused by inconsistent 
\mbox{\dataprep} procedures handling
training and test data separately.
\name automatically selects the appropriate 
data for training and evaluation.
\revision{However, if desired, users can handle training and test data differently by 
specifying separate DAGs for training and testing.
Common operators can be shared across the two DAGs without code duplication.}

\subsubsection{Operators}
\label{sec:operator}

\revision{Operators in \name are designed to cover the functions enumerated in Section~\ref{sec:opBasis}, using the data structures introduced above.}
A \name \textit{operator} takes one or more DCs and 
outputs DCs, ML models, or scalars.
Each operator encapsulates a function $f$, written in Scala,
to be applied to individual elements in the input DCs.
\revision{ \techreport{As noted above, $f$ can be learned from the input data
or user defined.
Like in Scikit-learn, \lang provides off-the-shelf implementations for common operations for ease of use.}
We describe the relationships between operator interfaces in \lang 
and $\mathcal{F}$ enumerated in Section~\ref{sec:opBasis} below.
}

\topic{Scanner}
\revision{\textit{Scanner} is the interface for parsing $\in \mathcal{F}$
and acts like a flatMap,
 i.e., for each input element, it adds zero or more elements to the output DC.
\techreport{Thus, it can also be used to perform filtering.}
The input and output of Scanner are \dcsu{}s.
\code{CSVScanner} 
in Line 4 of~\Cref{fig:wfEx}a) 
is an example of a Scanner that parses lines in a CSV file into key-value pairs for columns.}

\topic{Synthesizer}
\revision{\textit{Synthesizer} supports join $\in \mathcal{F}$,
for elements both across multiple DCs and within the same DC.
Thus, it can also support aggregation operations such as sliding windows in time series.
Synthesizers also serve the important purpose of specifying the set of SUs that make up an example \techreport{(where output FVs from the SUs are automatically assembled into a single FV)}.
In the simple case where each SU in a \dcsu corresponds to an example, 
a pass-through synthesizer is implicitly declared
by naming the output \dce, 
such as in Line 14 of Figure~\ref{fig:wfEx}a).}

\topic{Learner}
\revision{\textit{Learner} is the interface for learning and inference $\in \mathcal{F}$}, 
in a single operator.
A learner operator $L$ contains \revision{a learned function} $f$, 
\techreport{which can be populated by learning from the input data 
or loading from disk.
\revision{$f$ can be an ML model, 
but it can also be a feature transformation function 
that needs to be learned from the input dataset.}}
\papertext{
\revision{which can be an ML model or 
a feature transformation function learned from the input dataset.}}
When $f$ is empty, $L$ learns a model 
using input data designated for model training;
when $f$ is populated, $L$ performs inference on the input data using $f$ 
and outputs the inference results into a \dce.
For example, the learner \code{incPred} in Line 15 of Figure~\ref{fig:wfEx}a) 
is a learner trained on the ``train'' portion of the \dce \code{income}
and outputs inference results as the \dce \code{predictions}.

\topic{Extractor}
\revision{\textit{Extractor} is the interface 
for feature extraction and feature transformation $\in \mathcal{F}$.
Extractor contains the function $f$ applied on the input of SUs,
thus the input and output to an extractor are \dcsu{}s.
For functions that need to be learned from data, 
Extractor contains a pointer to the learner operator for learning $f$.}

\topic{Reducer}
\revision{Reducer is the interface for reduce $\in \mathcal{F}$
and thus the main operator interface for PPR.
The inputs to a reducer are \dce and an optional scalar and the output is a scalar,
where scalars refer to non-dataset objects.}
For example, \code{checkResults} in Figure~\ref{fig:wfEx}a) Line 17
\revision{is a reducer that} computes the prediction accuracy of the inference results in \code{predictions}.

\revision{
\subsection{Scope and Limitations}
\label{sec:hmlLim}
}

\input{scopeLimit}

%% file: sklearn-tables.tex
\begin{table}[t]
\begin{center}
\revision{\begin{tabular}{|c|m{15.5em}|}\hline
\multicolumn{1}{|c|}{\cellcolor{black!30}\textbf{Scikit-learn\techreport{ DPR, L/I}}} & 
  \multicolumn{1}{|c|}{\cellcolor{black!30}\textbf{Composed Members of $\mathcal{F}$}}\\\hline \hline
\code{fit(X[, y])} & {\em learning} ($\mathcal{D}\mapsto f$) \\ \hline
\code{predict_proba(X)} & {\em inference} ($(\mathcal{D}, f)\mapsto\mathcal{Y}$)\\ \hline
\code{predict(X)} & {\em inference}, optionally followed by {\em transformation} \\ \hline
\code{fit_predict(X[, y])} & {\em learning}, then {\em inference} \\ \hline
\code{transform(X)} & {\em transformation} or {\em inference},
                   depending on whether operation is learned via prior
                   call to \code{fit} \\ \hline
\code{fit_transform(X)} & {\em learning}, then {\em inference} \\ \hline \hline
\techreport{\multicolumn{1}{|c|}{\cellcolor{black!30}\textbf{Scikit-learn PPR}} & 
  \multicolumn{1}{|c|}{\cellcolor{black!30}\textbf{Composed Members of $\mathcal{F}$}}\\\hline \hline}
eval: \code{score(y$_{true}$, y$_{pred}$)} & {\em join} \code{y}$_{true}$ and \code{y}$_{pred}$ into
                                       a single dataset $\mathcal{D}$, then {\em reduce} \\ \hline
eval: \code{score(op, X, y)} & {\em inference}, then {\em join}, then {\em reduce} \\ \hline
selection: \code{fit($p_1, \ldots, p_n$)} & {\em reduce}, implemented in terms of
{\em learning}, {\em inference}, and {\em reduce} (for scoring) \\ \hline
\end{tabular}}
\end{center}
\caption{\revision{Scikit-learn DPR, L/I, and PPR coverage in terms of $\mathcal{F}$.}}
\label{tab:sklearn-dpr-li-ppr}
\end{table}

%% file: scopeLimit.tex

\topic{\revision{Coverage}}
\revision{ 
\techreport{In Section~\ref{sec:opBasis}, we described how
the set of basis operations $\mathcal{F}$ we propose covers
all major operations in Scikit-learn, one of the most comprehensive ML libraries.
We then showed in Section~\ref{sec:hml} that \lang captures all functions in $\mathcal{F}$.}
\papertext{We showed in Section~\ref{sec:hml} that \lang captures all functions in $\mathcal{F}$,
which is shown to cover all major operations in ML workflows~\cite{dorx2017}.}
}
\revision{
While \lang's interfaces are general enough to support all the common use cases,
users can additionally manually plug into our interfaces external implementations, 
such as from MLLib~\cite{meng2016mllib} and Weka~\cite{hall2009weka}, of missing operations.
\textit{Note that we provide utility functions that allow functions  
to work directly with raw records and FVs instead of \lang data structures
to enable direct application of external libraries.}
\techreport{For example, since all MLLib models implement 
the train (equivalent to learning) and predict (equivalent to inference) methods, 
they can easily be plugged into Learner in \name.}
We demonstrate in Section~\ref{sec:experiments} 
that the current set of implemented operations is sufficient for supporting 
applications across different domains.
}

\vspace{2pt}
\topic{\revision{Limitations}}
\revision{
Since \name currently relies on its Scala DSL for workflow specification, 
popular non-JVM libraries, such as TensorFlow~\cite{abadi2016tensorflow} and Pytorch~\cite{paszke2017pytorch}, 
cannot be imported easily without significantly degrading performance 
compared to their native runtime environment. 
\techreport{Developers with workflows implemented in other languages 
will need to translate them into \lang,
which should be straightforward due to the natural correspondence between \name operators and those in standard ML libraries, as established in Section~\ref{sec:hml}.}
That said, our contributions in materialization and reuse 
apply across all languages.
In the future, we plan on abstracting the DAG representation
in \name into a language-agnostic system 
that can sit below the language layer 
for all DAG based systems, including TensorFlow, Scikit-learn, and Spark.}

\revision {The other downside of \lang is that ML models are treated largely as black boxes.
Thus, work on optimizing learning, e.g., \cite{recht2011hogwild, zinkevich2010parallelized},
orthogonal to (and can therefore be combined with) our work, which operates at a coarser granularity.
}

%% file: representation.tex

In this section, we describe the \wf DAG, 
the abstract model used internally by \name
to represent a \wf program.
The \wf DAG model enables operator-level change tracking between iterations
and end-to-end optimizations. 

\subsection{The Workflow DAG}
\label{sec:dag}
\input{dag}

\subsection{Tracking Changes}
\label{sec:change}
\input{wfChanges}

%% file: dag.tex

At compile time, \name's intermediate code generator 
constructs a {\em \wf DAG} from 
\lang declarations\papertext{.}\techreport{, with nodes corresponding to operator outputs, 
(DCs, scalars, or ML models),
and edges corresponding to input-output relationships between operators.}

\begin{definition}
For a \wf $W$ containing \name operators $F = \{f_i\}$, 
the \wf DAG is a directed acyclic graph $G_W = (N, E)$,
where node $n_i \in N$ represents the output of $f_i \in F$
and $(n_i, n_j) \in E$ if the output of $f_i$ is an input to $f_j$.
\end{definition}

\noindent \Cref{fig:wfEx}b) shows the \wf DAG 
for the program in~\Cref{fig:wfEx}a). 
\techreport{Nodes for operators involved in DPR are colored purple
whereas those involved in L/I and PPR are colored orange.}
This transformation is straightforward, 
creating a node for each declared operator 
and adding edges between nodes based on the linking expressions,
e.g., \code{A results\_from B} creates an edge $(B, A)$.
Additionally, the intermediate code generator introduces edges not specified in the \wf 
between the extractor and the synthesizer nodes,
such as the edges marked by dots ($\bullet$) in Figure~\ref{fig:wfEx}b).
These edges  
connect extractors to downstream DCs
in order to automatically aggregate all features for learning.
One concern is that this may lead to redundant computation of unused features;
we describe pruning mechanisms to address this issue in Section~\ref{sec:prune}.

%% file: wfChanges.tex

As described in Section~\ref{sec:lifecycle}, 
a user starts with an initial workflow $W_0$ 
and iterates on this workflow. 
Let $W_t$ be the version of the workflow at iteration $t \geq 0$
with the corresponding DAG $G_W^t = (N_t, E_t)$;
$W_{t+1}$ denotes the workflow obtained in the next iteration.
To describe the changes between $W_t$ and $W_{t+1}$, 
we introduce the notion of \textit{equivalence}.

\begin{definition}\label{def:opEq}
A node $n_i^t \in N_t$ is {\em equivalent} to $n_i^{t+1} \in N_{t+1}$,
denoted as $n_i^t \equiv n_i^{t+1}$,
if \textbf{a}) the operators corresponding to $n_i^t$ and $n_i^{t+1}$ 
compute identical results on the same inputs
and \textbf{b}) $n_j^t \equiv n_j^{t+1} \ \forall \ n_j^t \in parents(n_i^t), n_j^{t+1} \in parents(n_i^{t+1})$.
We say $n_i^{t+1}\in N_{t+1}$ is {\em original} 
if it has no equivalent node in $N_t$.
\end{definition}
Equivalence is symmetric, i.e.,
$n_i^{t'} \equiv n_i^t \Leftrightarrow n_i^t \equiv n_i^{t'} $,
and transitive, 
i.e., $(n_i^{t} \equiv n_i^{t'} \wedge  n_i^{t'} \equiv n_i^{t''}) \Rightarrow n_i^t \equiv n_i^{t''} $.
Newly added operators in $W_{t+1}$ do not have equivalent nodes in $W_{t}$;
neither do nodes in $W_t$ that are removed in $W_{t+1}$.
For a node that persists across iterations, 
we need both the operator and the ancestor nodes to stay the same
for equivalence.
Using this definition of equivalence, 
we determine if intermediate results on disk can be safely reused 
through the notion of equivalent materialization:
\begin{definition}\label{def:eqMat}
A node $n_i^t \in N_t$ has an {\em equivalent materialization}
if $n_i^{t'}$ is stored on disk, where $t' \leq t$ and $n_i^{t'} \equiv n_i^t$.
\end{definition}
One challenge in determining 
equivalence is 
deciding whether two versions of an operator 
compute the same results on the same input.
For arbitrary functions,
this is undecidable as proven by 
Rice's Theorem~\mbox{\cite{rice1953classes}}.
The programming language community 
has a large body of work on
verifying operational equivalence 
for specific classes of programs~\cite{woodcock2009formal, pitts1997operationally, gordon1995tutorial}.
\name currently employs a simple representational 
equivalence verification---an operator 
remains equivalent across iterations 
if its declaration in the DSL is not modified 
and all of its ancestors are unchanged.
Incorporating more advanced 
techniques for verifying equivalence is future work.

To guarantee correctness\papertext{(the proof for which is in our technical report~\cite{dorx2017})}, i.e., results obtained at iteration $t$ 
reflect the specification for $W_t$ and are computed from the appropriate input,
we impose the constraint: 
\begin{constraint}\label{constraint:rerun}
At iteration $t+1$, if an operator $n_i^{t+1}$ is original, it must be recomputed.
\end{constraint}
\techreport{
With Constraint~\ref{constraint:rerun}, our current approach to tracking changes 
yields the following guarantee on result correctness:
\begin{theorem}\label{thm:correct}
\name returns the correct results
if the changes between iterations are made only within the programming interface,
i.e., all other factors, 
such as library versions and files on disk, 
stay invariant, i.e., unchanged, between executions at iteration $t$ and $t+1$.
\end{theorem}

\begin{proof}
First, note that the results for $W_0$ are correct since there is no reuse at iteration 0.
Suppose for contradiction 
that given the results at $t$ are correct,
the results at iteration $t+1$ are incorrect,
i.e., $\exists\ n_i^{t+1}$ s.t. 
the results for $n_i^t$ are reused when $n_i^{t+1}$ is original.
Under the invariant conditions in Theorem~\ref{thm:correct},
we can only have $n_i^{t+1} \not\equiv n_i^t$ 
if the code for $n_i$ changed or the code changed for an ancestor of $n_i$.
Since \name detects all code changes, 
it identifies all original operators.
Thus, for the results to be incorrect in \name,
we must have reused $n_i^t$ for some original $n_i^{t+1}$.
However, this violates Constraint~\ref{constraint:rerun}.
Therefore, the results for $W_t$ are correct $\forall \ t \geq 0$.
\end{proof}
}

%% file: optimization.tex

\def\oep{\textsc{Opt-Exec-Plan}\xspace}
\def\ojs{\textsc{Proj-Selection-Problem}\xspace}
\def\psp{\ojs}
\def\maxflow{\textsc{Max-Flow}\xspace}

\def\omp{\textsc{Opt-Mat-Plan}\xspace}
\def\matF{materialization run time\xspace}

In this section, we describe \gs's workflow-level optimizations,
motivated by the observation that {\em workflows
often share a large amount of intermediate computation
between iterations}; thus,
if certain intermediate results are materialized
at iteration $t$,
these can be used at iteration $t+1$.
We identify two distinct sub-problems:
\oep, which selects the operators to reuse given
previous materializations (Section~\mbox{\ref{sec:oep}}),
and \omp, which decides what to materialize to accelerate future iterations
(Section~\mbox{\ref{sec:omp}}).
We finally discuss pruning optimizations
to eliminate redundant computations (Section~\ref{sec:prune}). 
We begin by introducing common notation and definitions. 

\subsection{Preliminaries}\label{sec:optPrelim}
\input{optPrelim}

\subsection{Optimal Execution Plan}
\label{sec:oep}
\input{optExec}

\subsection{Optimal Materialization Plan}
\label{sec:omp}
\input{optMat}

\subsection{Workflow DAG Pruning}
\label{sec:prune}
\input{pruning}

%% file: optPrelim.tex

When introducing variables below, 
we drop the iteration number $t$ from $W_t$ and $G_W^t$ 
when we are considering a static workflow.

\topic{Operator Metrics}
In a \wf DAG $G_W = (N, E)$, 
each node $n_i \in N$ corresponding to the output of the operator $f_i$ 
is associated with a compute time $c_i$, 
the time it takes to compute $n_i$ from inputs in memory.
Once computed, 
$n_i$ can be materialized on disk 
and loaded back in subsequent iterations in time $l_i$,
referred to as its {\em load time}.
If $n_i$ does not have an equivalent materialization 
as defined in Definition~\ref{def:eqMat}, 
we set $l_i = \infty$.
Root nodes in the \wf DAG, which correspond to data sources,
have $l_i = c_i$.

\topic{Operator State}
During the execution of workflow $W$,
each node $n_i$ assumes one of the following states:
\begin{denselist}
\item {\em Load}, or $S_l$, if $n_i$ is loaded from disk;
\item {\em Compute}, or $S_c$, $n_i$ is computed from inputs;
\item {\em Prune}, or $S_p$, if $n_i$ is skipped (neither loaded nor computed).
\end{denselist}
Let $s(n_i) \in \{S_l, S_c, S_p\}$ denote the state of each $n_i\in N$.
To ensure that nodes in the Compute state have their inputs available, i.e., not \textit{pruned},
the states in a \wf DAG $G_W=(N,E)$ must satisfy the following 
{\em execution state constraint}:
\begin{constraint}
\label{constraint:exstate}
For a node $n_i \in N$, if $s(n_i) = S_c$,
then $s(n_j)\ne S_p$ for every $n_j \in parents(n_i)$.
\end{constraint}

\topic{Workflow Run Time}
A node $n_i$ in state $S_c$, $S_l$, or $S_p$ has run time $c_i$,
$l_i$, or $0$, respectively. 
The total run time of $W$ w.r.t. $s$ is thus 
\begin{equation} \label{eqn:wfcost}
T(W,s) = \sum_{n_i \in N} \indic{s(n_i) =S_c}c_i + \indic{s(n_i)=S_l}l_i
\end{equation}
where $\indic{}$ is the indicator function.

Clearly, setting all nodes to $S_p$ trivially minimizes Equation~\ref{eqn:wfcost}. 
However, recall that Constraint~\ref{constraint:rerun} requires 
all original operators to be rerun.
Thus, if an original operator $n_i$ is introduced, 
we must have $s(n_i) = S_c$, 
which by Constraint~\ref{constraint:exstate}
requires that $S(n_j) \neq S_p \ \forall n_j \in parents(n_i)$.
\techreport{Deciding whether to load or compute the parents 
can have a cascading effect on the states of their ancestors.}
We explore how to determine 
the states for each nodes 
to minimize Equation~\ref{eqn:wfcost} next.

%% file: optExec.tex

The {\em Optimal Execution Plan} (OEP) problem is
the core problem solved by \gs's DAG optimizer,
which determines at compile time the optimal execution plan 
given results and statistics from previous iterations.

\begin{problem} (\oep)
\label{prob:oep}
Given a \wf $W$ with DAG $G_W = (N, E)$, 
the compute time and the load time $c_i, l_i$ for each $n_i \in N$,
and a set of previously materialized operators $M$,
find a state assignment 
$s: N \rightarrow \{S_c, S_l, S_p\}$ 
that minimizes $T(W,s)$ while satisfying 
\mbox{\Cref{constraint:rerun}} and
\Cref{constraint:exstate}.
\end{problem}

Let $T^*(W)$ be the minimum execution time achieved by the solution to OEP, i.e.,
\begin{equation} \label{eqn:wfcost2}
T^*(W) = \min\limits_{s}\ T(W,s)
\end{equation}
Since this optimization takes place \textit{prior} to execution,
we must resort to operator statistics from past iterations.
{\em This does not compromise accuracy because
if a node $n_i$ has an equivalent materialization as defined in Definition~\ref{def:opEq},
we would have run the exact same operator before
and recorded accurate $c_i$ and $l_i$.}
A node $n_i$ without an equivalent materialization, such as a model with changed hyperparameters,
needs to be recomputed (Constraint~\ref{constraint:rerun}).

Deciding to load certain nodes can have cascading effects 
since  ancestors of a loaded node can potentially be pruned,
leading to large reductions in run time. 
On the other hand, Constraint~\ref{constraint:exstate} disallows
the parents of computed nodes to be pruned.
Thus, the decisions to load a node $n_i$ can be affected 
by nodes outside of the set of ancestors to $n_i$.
For example, in the DAG on the left in \Cref{fig:reduction}, 
loading $n_7$ allows $n_{1-6}$ to be potentially pruned.
However, the decision to compute $n_8$, 
possibly arising from the fact that $l_8 \gg c_8$,
requires that $n_5$ must not be pruned.

Despite such complex dependencies between the decisions for individual nodes, 
Problem~\ref{prob:oep} can be solved optimally in polynomial time
through a linear time reduction to the \textit{project-selection problem} (PSP), 
which is an application of \maxflow~\cite{kleinberg2006algorithm}.

\begin{problem} \ojs{} (PSP)
\label{prob:ojs}
Let $P$ be a set of projects. 
Each project $i \in P$ has a real-valued profit $p_i$ and a set of prerequisites $Q \subseteq P$.
Select a subset $A \subseteq P$ such that all prerequisites of a project $i \in A$ are also in $A$
and the total profit of the selected projects, $\sum_{i \in A} p_i$, is maximized.
\end{problem}

\input{reduction.tex}

\topic{Reduction to the Project Selection Problem}
We can reduce an instance of \Cref{prob:oep} $x$
to an equivalent instance of PSP $y$
such that the optimal solution to 
$y$ maps to an optimal solution of $x$. 
Let $G = (N, E)$ be the \wf DAG in $x$,
and $P$ be the set of projects in $y$.
We can visualize the prerequisite requirements in $y$ as a DAG 
with the projects as the nodes 
and an edge $(j, i)$ indicating 
that project $i$ is a prerequisite of project $j$.
The reduction, $\varphi$, 
depicted in \Cref{fig:reduction} for an example instance of $x$, 
is shown in Algorithm~\ref{algo:oep}.
For each node $n_i \in N$, we create two projects in PSP:
$a_i$ with profit $-l_i$ and $b_i$ with profit $l_i - c_i$.
We set $a_i$ as the prerequisite for $b_i$.
For an edge $(n_i, n_j) \in E$, 
we set the project $a_i$ corresponding to node $n_i$ 
as the prerequisite for the project $b_j$ corresponding to node $n_j$.
Selecting both projects $a_i$ and $b_i$ corresponding to $n_i$ 
is equivalent to computing $n_i$, i.e., $s(n_i) = S_c$,
while selecting only $a_i$ is equivalent to loading $n_i$, i.e., $s(n_i)=S_l$.
Nodes with neither projects selected are pruned.
An example solution mapping from PSP to OEP is shown in Figure~\ref{fig:reduction}.
Projects $a_4, a_5, a_6, b_6, a_7, b_7, a_8$ are selected, 
which cause nodes $n_4, n_5, n_8$ to be loaded, 
$n_6$ and $n_7$ to be computed,
and $n_1, n_2, n_3$ to be pruned.

\begin{algorithm}[t]\papertext{\scriptsize}
 \caption{OEP via Reduction to PSP}\label{algo:oep}
 \KwIn{$G_W = (N, E), \{l_i\}, \{c_i\}$}
  $P \leftarrow \emptyset$\;
 \For{$n_i \in N$}{
 	$P \leftarrow P \cup \{a_i\}$ \tcp*[r]{Create a project $a_i$}
 	$profit[a_i] \leftarrow -l_i$ \tcp*[r]{Set profit of $a_i$ to $-l_i$}
 	$P \leftarrow P \cup \{b_i\}$ \tcp*[r]{Create a project $b_i$}
 	$profit[b_i] \leftarrow l_i - c_i$ \tcp*[r]{Set profit of $b_i$ to $l_i - c_i$}
 	\tcp*[h]{Add $a_i$ as prerequisite for $b_i$.}\;
 	$prerequisite[b_i] \leftarrow prerequisite[b_i] \cup a_i$\; \label{alg:cost}
 	\For{$(n_i, n_j) \in \{\textrm{edges leaving from }n_i\} \subseteq E$}{
 		\tcp*[h]{Add $a_i$ as prerequisite for $b_j$.}\; 
 		$prerequisite[b_j] \leftarrow prerequisite[b_j] \cup a_i$\; \label{algo:depend}
 	}
 }
 \tcp*[h]{$A$ is the set of projects selected by PSP}\;
 $A \leftarrow $ PSP($P, profit[], prerequisite[]$)\;
 \For(\tcp*[f]{Map PSP solution to node states}){$n_i \in N$}{
 	\If{$a_i \in A$ and $b_i \in A$}{
 		$s[n_i] \leftarrow S_c$\; \label{algo:compute}
 		}  
 	\ElseIf{$a_i \in A$ and $b_i \not\in A$}{
 		$s[n_i] \leftarrow S_l$\;  \label{algo:load}
 		}
 	\Else{
 		$s[n_i] \leftarrow S_p$\;
 		}
 }
 \KwRet{$s[]$} \tcp*[r]{State assignments for nodes in $G_W$.}
\end{algorithm}

Overall, the optimization objective in PSP models the ``savings'' in OEP 
incurred by loading nodes instead of computing them from inputs.
We create an equivalence between cost minimization in OEP 
and profit maximization in PSP by mapping the costs in OEP to negative profits in PSP.
For a node $n_i$, 
picking only project $a_i$ (equivalent to loading $n_i$) 
has a profit of $-l_i$,
whereas picking both $a_i$ and $b_i$ (equivalent to computing $n_i$) 
has a profit of $-l_i + (l_i - c_i) = -c_i$.
The prerequisites established in Line~\ref{alg:cost}
that require $a_i$ to also be picked if $b_i$ is picked
are to ensure correct cost to profit mapping.
The prerequisites established in Line~\ref{algo:depend} corresponds to \Cref{constraint:exstate}.
For a project $b_i$ to be picked, 
we must pick every $a_j$ corresponding to each parent $n_j$ of $n_i$.
If it is impossible ($l_j = \infty$) or costly to load $n_j$,
we can offset the load cost by picking $b_j$ for computing $n_j$.
However, computing $n_j$ also requires its parents to be loaded or computed,
as modeled by the outgoing edges from $b_j$.
The fact that $a_i$ projects have no outgoing edges 
corresponds to the fact loading a node removes its dependency on all ancestor nodes.

\begin{theorem}
\label{thm:reduction}
Given an instance of \oep $x$, 
the reduction in Algorithm~\ref{algo:oep} produces
a feasible and optimal solution to $x$.
\end{theorem}
\papertext{\noindent See the technical report~\cite{dorx2017} for a proof.}
\techreport{\noindent See Appendix~\ref{sec:t1proof} for a proof.}

\techreport{
\topic{Computational Complexity}
For a \wf DAG $G_W = (N_W,$ $E_W)$ in OEP,
the reduction above results in $\bigo{|N_W|}$ projects 
and $\bigo{|E_W|}$ prerequisite edges in PSP.
PSP has a straightforward linear reduction to \maxflow~\cite{kleinberg2006algorithm}.
We use the Edmonds-Karp algorithm~\cite{edmonds1972theoretical} 
for \maxflow, 
which runs in time $\bigo{|N_W|\cdot|E_W|^2}$.}

\techreport{
\topic{Impact of change detection precision and recall}
The optimality of our algorithm for OEP assumes
that the changes between iteration $t$ and $t+1$ 
have been identified perfectly. 
In reality, this maybe not be the case 
due to the intractability of change detection,
as discussed in Section~\ref{sec:change}.
An undetected change is a false negative in this case,
while falsely identifying an unchanged operator as deprecated is a false positive.
A detection mechanism with high precision 
lowers the chance of unnecessary recomputation,
whereas anything lower than perfect recall leads to incorrect results.
In our current approach,
we opt for a detection mechanism 
that guarantee correctness under mild assumptions, 
at the cost of some false positives 
such as  $a + b \not\equiv b + a$.
}

%% file: reduction.tex

\begin{figure}[t]
\begin{subfigure}{.48\linewidth}
\begin{tikzpicture}[
            > = stealth, 
            shorten > = 1pt, 
            auto,
            node distance = 1.5cm, 
            semithick, 
            every node/.style={
                scale=0.7,
                shape=circle,
                draw = black,
                thick,
            }
        ]

        \tikzstyle{p}=[
            fill = white,
        ]
        
        \tikzstyle{l}=[
            fill = red!20,
        ]
        
        \tikzstyle{c}=[
            fill = blue!20,
        ]

        \node[p] (n6) [label=left:$S_l$]{$n_4$};
        \node[p] (n1) [above left of=n6, label=$S_p$] {$n_1$};
        \node[p] (n2) [above of=n6, label=$S_p$] {$n_2$};
        \node[p] (n3) [above right of=n6, label=$S_p$] {$n_3$};
        \node[p] (n7) [below right of=n6, label=below:$S_c$]{$n_6$};
        \node[p] (n8) [below left of=n7, label=below:$S_c$]{$n_7$};
        \node[p] (n9) [below right of=n7, label=below:$S_l$]{$n_8$};
        \node[p] (n5) [above right of=n7, label=left:$S_l$]{$n_5$};

        \path[->] (n5) edge (n7);
        \path[->] (n6) edge (n8);
        \path[->] (n6) edge (n7);
        \path[->] (n7) edge (n8);
        \path[->] (n5) edge (n9);
        \path[->] (n1) edge (n6);
        \path[->] (n2) edge (n6);
        \path[->] (n3) edge (n6);

        \node[draw,shape=rectangle,scale=1.5] at (2.7, -0.125) {$\varphi$};
        \draw[black,->,thick] (2, -0.5) -- (3.5, -0.5);
    \end{tikzpicture}
\label{fig:reductionA}
\end{subfigure}
\begin{subfigure}{.5\linewidth}
\begin{tikzpicture}[
            > = stealth, 
            shorten > = 1pt, 
            auto,
            semithick, 
            every node/.style={
                scale=0.7,
                shape=circle,
                draw = black,
                thick,
            }
        ]

        \tikzstyle{a}=[
            node distance = 1.0cm, 
            fill = red!20,
        ]
        
        \tikzstyle{b}=[
            node distance = 2.0cm, 
            fill = blue!20,
        ]

        \node[b] (b6) {$b_4$};
        \node[a] (a6) [below of=b6, label=below:$\checkmark$] {$a_4$};

        \node[b] (b1) [above left of=b6] {$b_1$};
        \node[a] (a1) [below of=b1] {$a_1$};

        \node[b] (b2) [above of=b6] {$b_2$};
        \node[a] (a2) [below of=b2] {$a_2$};

        \node[b] (b3) [above right of=b6] {$b_3$};
        \node[a] (a3) [below of=b3] {$a_3$};

        \node[b] (b7) [below right of=b6, label=above:$\checkmark$]{$b_6$};
        \node[a] (a7) [below of=b7, label=below:$\checkmark$]{$a_6$};

        \node[b] (b8) [below left of=b7, label=left:$\checkmark$]{$b_7$};
        \node[a] (a8) [below of=b8, label=left:$\checkmark$]{$a_7$};

        \node[b] (b9) [below right of=b7]{$b_8$};
        \node[a] (a9) [below of=b9, label=right:$\checkmark$]{$a_8$};

        \node[b] (b5) [above right of=b7]{$b_5$};
        \node[a] (a5) [below of=b5, label=right:$\checkmark$]{$a_5$};
        
        \tikzset{curved/.style={<-,relative=false,in=60,out=0}}
        \tikzset{left/.style={->,bend left,in=120,out=60}}

        \path[<-]     (a1) edge (b1);
        \path[<-]     (a2) edge (b2);
        \path[<-]     (a3) edge (b3);
        \path[<-]     (a5) edge (b5);
        \path[<-]     (a6) edge (b6);
        \path[<-]     (a7) edge (b7);
        \path[<-]     (a8) edge (b8);
        \path[<-]     (a9) edge (b9);

        \path[<-]     (a5) edge (b7);
        \path[left]   (b8) edge (a6);
        \path[<-]     (a6) edge (b7);
        \path[<-]     (a7) edge (b8);
        \path[<-] (a5) edge (b9);
        \path[<-]     (a1) edge (b6);
        \path[<-] (a2) edge (b6);
        \path[<-]     (a3) edge (b6);

    \end{tikzpicture}
\label{fig:reductionB}
\end{subfigure}
\caption{Transforming a \wf DAG to a set of projects
and dependencies.
Checkmarks ($\checkmark$) in the RHS DAG indicate a feasible solution to PSP,
which maps onto the node states ($S_p, S_c, S_l$) in the LHS DAG. }
\label{fig:reduction}
\end{figure}
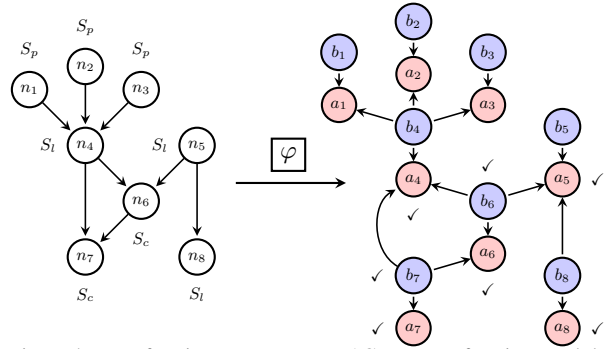

%% file: optMat.tex

The \omp{} (OMP) problem is tackled by \name's materialization optimizer:
while running workflow $W_t$ at iteration $t$, 
intermediate results are selectively materialized 
for the purpose of accelerating execution in iterations $> t$.
We now formally introduce OMP
and show that it is {\sc NP-hard} even under strong assumptions.
We propose an online heuristic for OMP that runs in linear time
and achieves good reuse rate in practice 
(as we will show in Section~\ref{sec:experiments}),
in addition to minimizing memory footprint 
by avoiding unnecessary caching of intermediate results.

\topic{Materialization cost}
We let $s_i$ denote the {\em storage
cost} for materializing $n_i$, 
representing the size of $n_i$ on disk. 
When loading $n_i$ back from disk to memory,
we have the following relationship between
load time and storage cost: 
$l_i = s_i /$(disk read speed).
For simplicity, we also assume the time to write $n_i$ to disk 
is the same as the time for loading it from disk, i.e., $l_i$.
We can easily generalize to 
the setting where load and write latencies are different.

To quantify the benefit of materializing intermediate results 
at iteration $t$ on subsequent iterations,
we formulate the \textit{\matF} $T_M(W_t)$
to capture the tradeoff between 
the additional time to materialize intermediate results
and the run time reduction in iteration $t+1$.
Although materialized results can be reused
in multiple future iterations,
we only consider the $(t+1)$th iteration 
since the total number of future iterations\techreport{, $\mathcal{T}$,} is unknown.
\techreport{Since modeling $\mathcal{T}$ is a complex open problem, 
we defer the amortization model to future work.}

\begin{definition} Given a workflow $W_t$, 
operator metrics $c_i, l_i, s_i$ for every $n_i \in N_t$,
and a subset of nodes $M \subseteq N_t$,
the {\em \matF} is defined as
\label{def:matE}
\begin{equation}\label{eqn:matTime}
T_M(W_t) = \sum\limits_{n_i \in M} l_i +  T^*(W_{t+1})
\end{equation}
where $\sum_{n_i \in M} l_i$ is the time to materialize all nodes selected for materialization,
and
$T^*(W_t)$ is the optimal workflow run time 
obtained using the algorithm in Section~\ref{sec:oep}, 
with $M$ materialized.
\end{definition}
\noindent Equation~\ref{eqn:matTime} defines the optimization objective for OMP.

\begin{problem} (\omp) 
\label{prob:omp}
Given a \wf $W_t$ with DAG $G_W^t=(N_t, E_t)$ at iteration $t$ 
and a storage budget $S$,
find a subset of nodes $M \subseteq N_t$ to materialize at $t$ 
in order to minimize $T_M(W_t)$,
while satisfying the storage constraint $\sum_{n_i \in M} s_i \leq S$.
\end{problem}

\techreport{
Let $M^*$ be the optimal solution to OMP, i.e.,
\begin{equation}\label{eqn:omp}
\argmin\limits_{M \subseteq N_t} \sum\limits_{n_i \in M} l_i +  T^*(W_{t+1})
\end{equation}}
\revision{As discussed in~\cite{xin2018developers},}
there are many possibilities for $W_{t+1}$,
and they vary by application domain.
User modeling and predictive analysis of $W_{t+1}$ itself 
is a substantial research topic
that we will address in future work.
This user model can be incorporated into OMP
by using the predicted changes 
to better estimate the likelihood of reuse for each operator.
However,
even under very restrictive assumptions about $W_{t+1}$,
we can show that \omp{} is {\sc NP-Hard},
via a simple reduction from the {\sc Knapsack} problem.
\begin{theorem}\label{thm:nphard}
\omp{} is NP-hard. 
\end{theorem}
\techreport{\noindent See Appendix~\ref{sec:npProof} for a proof.}
\papertext{\noindent See the technical report~\cite{dorx2017} for a proof.}

\topic{Streaming constraint}
Even when $W_{t+1}$ is known, 
solving \omp{} optimally requires knowing the run time statistics for all operators,
which can be fully obtained only at the end of the workflow.
Deferring materialization decisions
until the end requires all intermediate results to be cached or recomputed,
which imposes undue pressure on memory and cripples performance.
Unfortunately, reusing statistics from past iterations as in Section~\ref{sec:oep}
is not viable here because of the cold-start problem---materialization decisions
need to be made for new operators based on realistic statistics.
Thus, to avoid slowing down execution with high memory usage, 
we impose the following constraint. 

\begin{definition} \label{def:ofs}
Given a \wf DAG $G_w = (N, E)$, $n_i \in N$ is {\em out-of-scope} at runtime 
if all children of $n_i$ have 
been computed or reloaded from disk,
thus removing all dependencies on $n_i$.
\end{definition}

\begin{constraint} \label{cons:stream}
Once $n_i$ becomes out-of-scope, it is either materialized immediately or removed from cache.
\end{constraint}

\topic{OMP Heuristics}
We now describe the heuristic employed by \name to approximate OMP 
while satisfying Constraint~\ref{cons:stream}. 
\begin{definition}
Given \wf DAG $G_w = (N, E)$, the {\em cumulative run time} for a node $n_i$ is defined as 
\begin{equation}
C(n_i) = t(n_i) + \sum\limits_{n_j \in ancestors(n_i)} t(n_j)
\end{equation}
where $t(n_i) = \indic{s(n_i) =S_c}c_i + \indic{s(n_i)=S_l}l_i$.
\end{definition}

\begin{algorithm}[t]
\papertext{\scriptsize}
 \caption{Streaming OMP}\label{algo:somp}
 \KwData{$G_w = (N, E), \{l_i\}, \{c_i\}, \{s_i\}, \text{storage budget } S$}
 $M \leftarrow \emptyset$\;
 \While{Workflow is running}{
     $O \leftarrow$ FindOutOfScope($N$)\;
     \For{$n_i \in O$} {
     		\If{$C(n_i) > 2l_i$ and $S - s_i \geq 0$}{
     		Materialize $n_i$\;
     		$M \leftarrow M \cup \{n_i\}$\;
     		$S \leftarrow S - s_i$
     		}
     }
 }
\end{algorithm}

\noindent Algorithm~\ref{algo:somp} shows the heuristics 
employed by \name's materialization optimizer
to decide what intermediate results to materialize.
In essence, Algorithm~\ref{algo:somp} decides to materialize 
if twice the load cost is less than the cumulative run time for a node.
The intuition behind this algorithm is that 
assuming loading a node allows all of its ancestors to be pruned,
the materialization time in iteration $t$ and the load time in iteration $t+1$ combined
should be less than the total pruned compute time, 
for the materialization to be cost effective.

\techreport{
Note that the decision to materialize does not depend on 
which ancestor nodes have been previously materialized. 
The advantage of this approach is that regardless 
of where in the workflow the changes are made,
the reusable portions leading up to the changes 
are likely to have an efficient execution plan. 
That is to say, if it is cheaper to load a reusable node $n_i$ than to recompute,
Algorithm~\mbox{\ref{algo:somp}} would have materialized $n_i$ previously,
allowing us to make the right choice for $n_i$.
Otherwise, Algorithm~\mbox{\ref{algo:somp}} would have materialized some ancestor $n_j$ of $n_i$
such that loading $n_j$ and computing everything leading to $n_i$ 
is still cheaper than loading $n_i$.}

Due to the streaming Constraint~\ref{cons:stream}, 
complex dependencies between descendants of ancestors 
such as the one between $n_5$ and $n_8$ in Figure~\ref{fig:reduction} 
previously described in Section~\ref{sec:oep},
are ignored by Algorithm~\ref{algo:somp}---we 
cannot retroactively update our decision for $n_5$
after $n_8$ has been run.
We show in Section~\ref{sec:experiments} 
that this simple algorithm is effective in multiple application domains.

\revision{\techreport{
\topic{Limitations of Streaming OMP}
The streaming OMP heuristic given in Algorithm~\ref{algo:somp}
can behave poorly in pathological cases. For one simple example,
consider a workflow given by a chain DAG of $m$ nodes, where node $n_i$ (starting from $i=1$)
is a prerequisite for node $n_{i+1}$. If node $n_i$ has $l_i = i$ and 
$c_i = 3$, for all $i$, then Algorithm~\ref{algo:somp} will choose to
materialize every node, which has storage costs of $\bigo{m^2}$, whereas
a smarter approach would only materialize later nodes and perhaps have
storage cost $\bigo{m}$. If storage is exhausted because Algorithm~\ref{algo:somp}
persists too much early on, this could easily lead to poor execution times
in later iterations. We did not observe this sort of pathological behavior
in our experiments.}
}

\revision{\techreport{
\topic{Mini-Batches}
In the stream processing (to be distinguished from the streaming constraint in Constraint~\ref{cons:stream}) where the input is divided into mini batches processed end-to-end independently, Algorithm~\ref{algo:somp} can be adapted as follows: 
1) make materialization decisions using the load and compute time for the first mini batch processed end-to-end;
2) reuse the same decisions for all subsequent mini batches for each operator.
This approach avoids dataset fragmentation that complicates reuse for different workflow versions.
We plan on investigating other approaches for adapting
\name for stream processing in future work.}
}

%% file: pruning.tex

\label{sec:prune}
\papertext{
In addition to optimizations involving intermediate result reuse,
\name further reduces overall workflow execution 
time by pruning extraneous operators
from the \wf DAG. 
In a nutshell, \name traverses the DAG backwards from the output nodes 
and prunes away any nodes not visited in this traversal---a simple form
of {\em program slicing}~\cite{weiser1981program}.
\name provides two additional mechanisms 
for pruning operators other than using the lack of output dependency:
1) user can explicitly specify nodes to be excluded in \lang for manual feature selection;
and 2) if the user desires it, \name can inspect relevant data to determine 
low-impact operators that can be removed without compromising the model performance.
We plan on investigating the latter extensively in future work.
}

\techreport{
In addition to optimizations involving intermediate result reuse,
\name further reduces overall workflow execution time by
time by pruning extraneous operators
from the \wf DAG. 

\name performs pruning
by applying program slicing on the \wf DAG.
In a nutshell, \name traverses the DAG backwards from the output nodes 
and prunes away any nodes not visited in this traversal.
Users can explicitly guide this process in the programming interface
through the \code{has\_extractors} and \code{uses} keywords, 
described in Table~\ref{tab:semantics}.
An example of an Extractor pruned in this fashion is 
\code{raceExt}(grayed out) in Figure~\ref{fig:wfEx}b), 
as it is excluded from the \code{rows has\_extractors} statement.
This allows users to conveniently perform manual feature selection using domain knowledge.

\name provides two additional mechanisms 
for pruning operators other than using the lack of output dependency, described next.
}

\techreport{
\topic{Data-Driven Pruning}
Furthermore, \name inspects relevant data 
to automatically identify operators to prune.
The key challenge in \textit{data-driven pruning} is 
data lineage tracking across the entire workflow.
For many existing systems, 
it is difficult to trace features in the learned model back to the operators that produced them.
To overcome this limitation, 
\name performs additional provenance bookkeeping 
to track the operators that led to each feature in the model
when converting DPR output to ML-compatible formats.
An example of data-driven workflow optimization enabled by this bookkeeping
is pruning features by model weights.
Operators resulting in features with zero weights can be pruned 
without changing the prediction outcome,
thus lowering the overall run time without compromising model performance.

Data-driven pruning is a powerful technique that can be extended to 
unlock the possibilities for many more impactful automatic workflow optimizations. 
Possible future work includes using this technique 
to minimize online inference time in large scale, high query-per-second settings
and to adapt the workflow online in stream processing.

\topic{Cache Pruning}
While Spark, the underlying data processing engine for \name, 
provides automatic data uncaching 
via a least-recently-used (LRU) scheme,
\name improves upon the performance 
by actively managing the set of data to evict from cache.
From the DAG, \name can detect 
when a node becomes out-of-scope.
Once an operator has finished running,
\name analyzes the DAG to uncache newly out-of-scope nodes.
Combined with the lazy evaluation order,
the intermediate results for an operator reside in cache
only when it is immediately needed for a dependent operator.

One limitation of this eager eviction scheme is 
that any dependencies undetected by \name,
such as the ones created in a UDF,
can lead to premature uncaching of DCs before they are truly out-of-scope.
The \code{uses} keyword in \lang, described in Table~\mbox{\ref{tab:semantics}}, 
provides a mechanism
for users to manually prevent this 
by explicitly declaring a UDF's dependencies on other operators.
In the future, we plan on providing automatic UDF dependency detection
via introspection.
}

%% file: experiments.tex

\begin{table*}[t!]
\centering
\scriptsize
\begin{tabular}{r|c|c|c|c}

     & \textbf{Census} (Source: \cite{ddRepo}) & \textbf{Genomics} (Source: \cite{ren2017life}) &\textbf{IE} (Source: \cite{deepdive2016}) & \textbf{MNIST} (Source: \cite{sparks2016end}) \\
\hline
\hline
\textbf{Num. Data Source} & Single & Multiple & Multiple  & Single   \\
\hline
\textbf{Input to Example Mapping} &  
One-to-One & One-to-Many & One-to-Many & One-to-One  \\
\hline
\textbf{Feature Granularity} & Fine Grained & N/A & Fine Grained & Coarse Grained \\
\hline 
\textbf{Learning Task Type} & Supervised; Classification & Unsupervised & Structured Prediction & Supervised; Classification \\
\hline
\textbf{Application Domain} &Social Sciences & Natural Sciences & NLP & Computer Vision \\
\hline
\hline
\textbf{Supported by \name} & \checkmark  & \checkmark & \checkmark & \checkmark \\
\hline
\textbf{Supported by KeystoneML}  & \checkmark &  \checkmark  &  \cellcolor{black!25}& \checkmark * \\
\hline
\textbf{Supported by DeepDive}  & \checkmark *  & \cellcolor{black!25} & \checkmark * &\cellcolor{black!25}  \\
\hline
\end{tabular}
\vspace{4pt}
\caption{Summary of workflow characteristics and support by the systems compared.
Grayed out cells indicate that 
the system in the row does not support the workflow in the column.
$\checkmark*$ indicates that the implementation is by the original developers of DeepDive/KeystoneML.}
\label{tab:workflows}
\end{table*}

The goal of our evaluation is to test if
\name 
1) {\em supports} ML workflows in a variety of application domains;
2) {\em accelerates} iterative execution through intermediate result reuse, compared to other ML systems 
that don't optimize iteration;
3) is {\em efficient},
enabling optimal reuse without incurring a large storage overhead.

\subsection{Systems and Baselines for Comparison}

We compare the optimized version of \name, \opt, against two state-of-the-art ML 
workflow systems: KeystoneML~\cite{sparks2016end}, and DeepDive~\cite{zhang2015deepdive}.
In addition, we compare \opt 
with two simpler versions, \am and \nm.
While we compare against DeepDive, and KeystoneML 
to verify 1) and 2) above,
\am and \nm are used to verify 3).
We describe each of these variants below:

\vspace{1pt}
\noindent {\bf KeystoneML.}
KeystoneML~\cite{sparks2016end} is a system, written in Scala and built on top of Spark,
for the construction of large scale, end-to-end, ML pipelines.
KeystoneML specializes in classification tasks on structured input data.
No intermediate results are materialized in KeystoneML, 
as it does not optimize execution across iterations.

\vspace{1pt}
\noindent {\bf DeepDive.}
DeepDive~\cite{zhang2015deepdive,deepdive2016} is a system, 
written using Bash scripts and Scala for the main engine,
with a database backend, for the construction of 
end-to-end information extraction pipelines.
Additionally, DeepDive provides limited support for classification tasks.
All intermediate results are materialized in DeepDive.

\vspace{1pt}
\noindent {\bf \opt}. A version of \name that uses 
Algorithm~\ref{algo:oep} for the optimal reuse strategy
and Algorithm~\ref{algo:somp} to decide what to materialize.

\vspace{1pt}
\noindent {\bf \am}. A version of \name that uses 
the same reuse strategy as \opt
and {\em always materializes} all intermediate results.

\vspace{1pt}
\noindent {\bf \nm}. A version of \name that uses 
the same reuse strategy as \opt
and {\em never materializes} any intermediate results.

\subsection{Workflows}
We conduct our experiments using four real-world ML workflows 
spanning a range of application domains.
Table~\ref{tab:workflows} summarizes the characteristics of the four workflows, 
described next.
We are interested in four properties when characterizing each workflow:

\begin{denselist}
\item {\em Number of data sources:}
whether the input data comes from a single source 
(e.g., a CSV file) or multiple sources (e.g., documents and a knowledge base),
\revision{necessitating joins}.

\item {\em Input to example mapping:}
the mapping from each input data unit (e.g., a line in a file) 
to each learning example for ML.
One-to-many mappings require more complex \dataprep than one-to-one mappings.

\item {\em Feature granularity:}
fine-grained features involve 
applying extraction logic on a specific piece of the data (e.g., 2nd column)
and are often application-specific,
whereas coarse-grained features 
are obtained by applying an operation, 
usually a standard DPR technique such as normalization,
on the entire dataset.

\item {\em Learning task type:}
while classification and structured prediction tasks 
both fall under supervised learning for having observed labels,
structured prediction workflows involve more complex \dataprep and models;
unsupervised learning tasks do not have known labels,
so they often require more qualitative and fine-grained analyses of outputs.

\end{denselist}

\topic{Census Workflow}
\label{sec:census}
This workflow corresponds to a classification task 
with simple features from structured inputs
 from the DeepDive Github repository~\mbox{\cite{ddRepo}}.
It uses the Census Income dataset~\cite{Dua:2017},
with 14 attributes 
representing demographic information,
with the goal 
to predict whether a person's annual income is >50K,
using fine-grained features derived from input attributes.
The complexity of this workflow is representative of 
use cases in the social and natural sciences, 
where covariate analysis is conducted on well-defined variables.
\name code for the initial version of this workflow is shown in Figure~\ref{fig:wfEx}a).
\techreport{This workflow evaluates a system's efficiency 
in handling simple ML tasks with fine-grained feature engineering.}

\topic{Genomics Workflow}
This workflow is described in Example~\ref{ex:gene}, involving 
two major steps: 
1) split the input articles into words 
and learn vector representations for entities of interest,
identified by joining with a genomic knowledge base,
using word2vec~\cite{mikolov2013distributed};
2) cluster the vector representation of genes using K-Means
to identify functional similarity.
\techreport{Each input record is an article,
and it maps onto many gene names, which are training examples.}
This workflow has minimal \dataprep with no specific features
but involves multiple learning steps.
Both learning steps are unsupervised,
which leads to more qualitative and exploratory evaluations of the model 
outputs\techreport{ than the standard metrics used for supervised learning}.
We include a workflow with unsupervised learning and multiple learning steps
to verify that the system is able to accommodate variability in the learning task.

\topic{Information Extraction (IE) Workflow}
This workflow involves identifying mentions of spouse pairs
from news articles, using a knowledge-base of known spouse pairs, 
from DeepDive~\cite{deepdive2016}.
The objective is to extract structured information 
from unstructured input text, 
using complex fine-grained features such as part-of-speech tagging.
\techreport{Each input article contains $\geq 0$ spouse pairs, 
hence creating a one-to-many relationship between input records and learning examples.}
This workflow exemplifies use cases in information extraction, 
and tests a system's ability to handle joins and complex \dataprep.

\topic{MNIST Workflow}
The MNIST dataset~\cite{lecun1998gradient}
contains images of handwritten digits to be classified,
which is a well-studied task in the computer vision community,
\revision{from the KeystoneML~\mbox{\cite{sparks2016end}} evaluation}.
The workflow 
involves nondeterministic (and hence not reusable) \dataprep,
with a substantial fraction of the overall run time
spent on L/I in a typical iteration.
We include this application to ensure that 
in the extreme case where there is little reuse across iterations,
\name does not \revision{incur a large overhead}.

\smallskip
\noindent Each workflow was implemented in \name, 
and if supported, in DeepDive and KeystoneML,
with \checkmark * in Table~\ref{tab:workflows} indicating 
that we used an existing implementation by the developers of DeepDive or KeystoneML,
\papertext{
\revision{with scripts enumerated in~\cite{dorx2017}.}
}
\techreport{
which can be found at: 
\begin{denselist}
\item Census DeepDive: \url{https://github.com/HazyResearch/deepdive/blob/master/examples/census/app.ddlog}
\item IE DeepDive: \url{https://github.com/HazyResearch/deepdive/blob/master/examples/spouse/app.ddlog}
\item MNIST KeystoneML: \url{https://github.com/amplab/keystone/blob/master/src/main/scala/keystoneml/pipelines/images/mnist/MnistRandomFFT.scala}
\end{denselist}
}
DeepDive has its own DSL,
while KeystoneML's programming interface is an embedded DSL in Scala, similar to \lang.
We explain limitations that prevent DeepDive and KeystoneML from supporting certain workflows (grey cells) 
in Section~\mbox{\ref{sec:useCaseSupport}}.

\subsection{Running Experiments}

\topic{Simulating iterative development}
In our experiments, we modify the workflows to simulate typical iterative development
by a ML application developer or data scientist.
Instead of arbitrarily choosing operators to modify in each iteration,
we use the iteration frequency in \revision{Figure 3 from our literature study~\cite{xin2018developers}}
to determine the type of modifications to make in each iteration,
for the specific domain of each workflow.
We convert the iteration counts into fractions
that represent the likelihood of a certain type of change.
At each iteration, we draw an iteration type from \{DPR, L/I, PPR\} according to these likelihoods.
Then, we randomly choose an operator of the drawn type and modify its source code.
For example, if an ``L/I'' iteration were drawn, 
we might change the regularization parameter for the ML model.
We run 10 iterations per workflow (except NLP, which has only DPR iterations)\techreport{,
double the average iteration count found in our survey in Section~\ref{sec:overview-survey}}.

\techreport{Note that in real world use, 
the modifications in each iteration are entirely up to the user.  
\name is not designed to suggest modifications,
and the modifications chosen in our experiments 
are for evaluating only system run time and storage use.
We use statistics aggregated over $>100$ papers 
to determine the iterative modifications
in order to simulate behaviors of the {\em average domain expert}
more realistically than arbitrary choice.}

\topic{Environment}
All single-node experiments are run on a server with 
125 GiB of RAM, 
16 cores on 8 CPUs (Intel Xeon @ 2.40GHz),
and 2TB HDD with 170MB/s as both the read and write speeds.
\revision{Distributed experiments are run on nodes 
each with 64GB of RAM, 
16 cores on 8 CPUs (Intel Xeon @ 2.40GHz), 
and 500GB of HDD with 180MB/s as both the read and write speeds.}
We set the storage budget in \name to 10GB.
\revision{That is, 10GB is the maximum accumulated disk storage 
for \opt at all times during the experiments.}
After running the initial version to obtain the run time for iteration 0,
a workflow is modified according to the type of change determined as above.
In all \revision{four} systems the modified workflow is recompiled.
In DeepDive, we rerun the workflow using the command \code{deepdive run}.
In \name and KeystoneML, we resubmit a job to Spark in local mode.
\techreport{We use Postgres as the database backend for DeepDive.}
\revision{Although \name and KeystoneML support distributed execution
via Spark,
DeepDive needs to run on a single server.
Thus, we compare against all systems on a single node 
and additionally compare against KeystoneML on clusters.}

\subsection{Metrics}
\label{sec:metrics}
We evaluate each system's ability to support diverse ML tasks
by qualitative characterization of the workflows and use-cases 
supported by each system.
Our primary metric for workflow execution is \textit{cumulative run time}
over multiple iterations.
\techreport{The cumulative run time considers only the run time of the workflows,
not any human development time.
We measure with wall-clock time because it is the latency experienced by the user.}
When computing cumulative run times, we average the per-iteration 
run times over five complete runs for stability. 
Note that the per-iteration time measures both the time to execute the workflow
and any time spent to materialize intermediate results.
\revision{We also measure \textit{memory usage} \papertext{(results shown in the technical report)}
\techreport{to analyze the effect of batch processing, }
and} measure {\em storage size} 
to compare the run time reduction to storage ratio of time-efficient approaches.
\revision{
Storage is compared only for variants of \name 
since other systems do not support automatic reuse.}

\subsection{Evaluation vs. State-of-the-art Systems}

\subsubsection{Use Case Support}
\label{sec:useCaseSupport}
\frameme{
\name supports ML workflows in
multiple distinct application domains, 
spanning tasks with varying complexity in both supervised and unsupervised learning.
}

Recall that the four workflows used in our experiments are in 
social sciences, NLP, computer vision, and natural sciences, respectively.
Table~\ref{tab:workflows} lists the characteristics of each workflow
and the three systems' ability to support it.
Both KeystoneML and DeepDive have limitations 
that prevent them from supporting certain types of tasks.
The pipeline programming model in KeystoneML
is effective for large scale classification
and can be adapted to support unsupervised learning.
However, it makes fine-grained features cumbersome to program
and is not conducive to structured prediction tasks due to complex \dataprep.
On the other hand, DeepDive is highly specialized for information extraction
and focuses on supporting \dataprep. 
Unfortunately, its learning and evaluation components are not configurable by the user,
limiting the type of ML tasks supported.
DeepDive is therefore unable to support the MNIST and genomics workflows, 
both of which required custom ML models.
Additionally, we are only able to show DeepDive performance for DPR iterations 
for the supported workflows in our experiments.

\subsubsection{Cumulative Run Time}
\label{sec:cTime}

\frameme{
\name achieves up to $\mathbf{19\times}$ cumulative run time reduction in ten iterations 
over state-of-the-art ML systems. 
}

\begin{figure}[t]
\centering

	\begin{subfigure}[t]{0.22\textwidth}
	    \centering
    \caption{Census}
        \includegraphics[width=\textwidth]{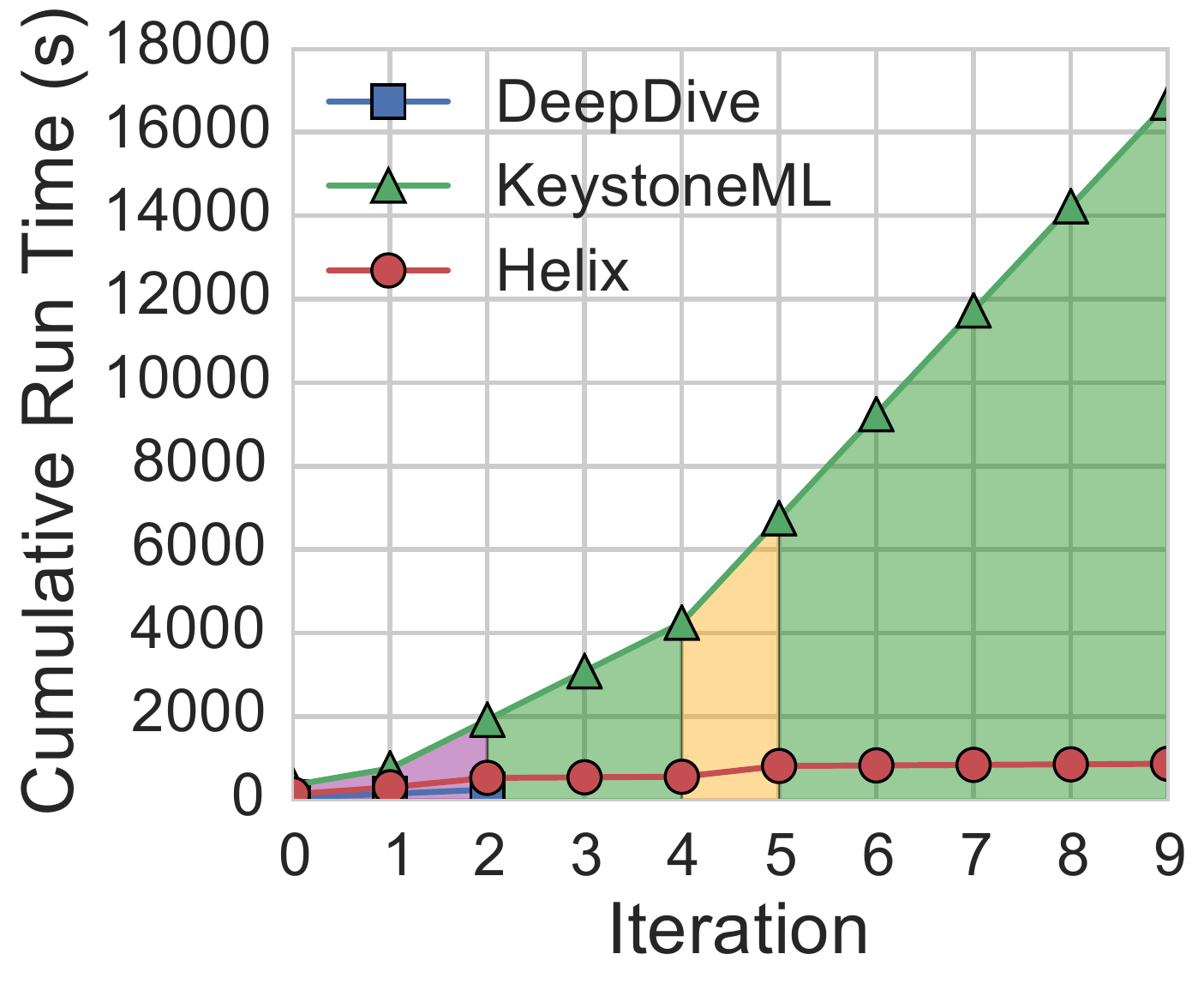}
        \label{fig:censusSys}
    \end{subfigure}
    \hfill
    \begin{subfigure}[t]{0.22\textwidth}
        \centering
    \caption{Genomics}
        \includegraphics[width=\textwidth]{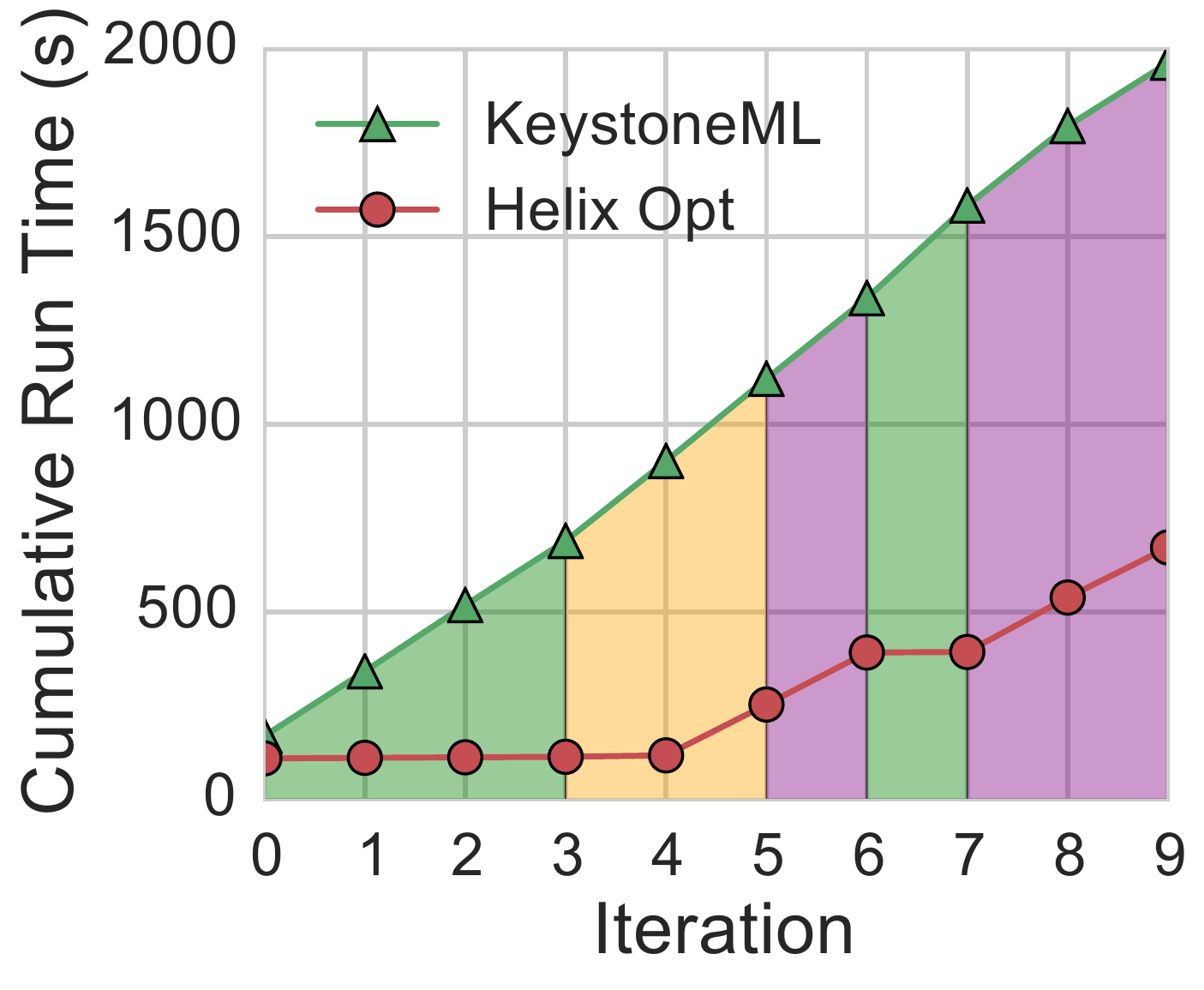}
        \label{fig:knowengSys}
    \end{subfigure}
    
    \begin{subfigure}[t]{0.22\textwidth}
        \centering
        \caption{NLP}
        \includegraphics[width=\textwidth]{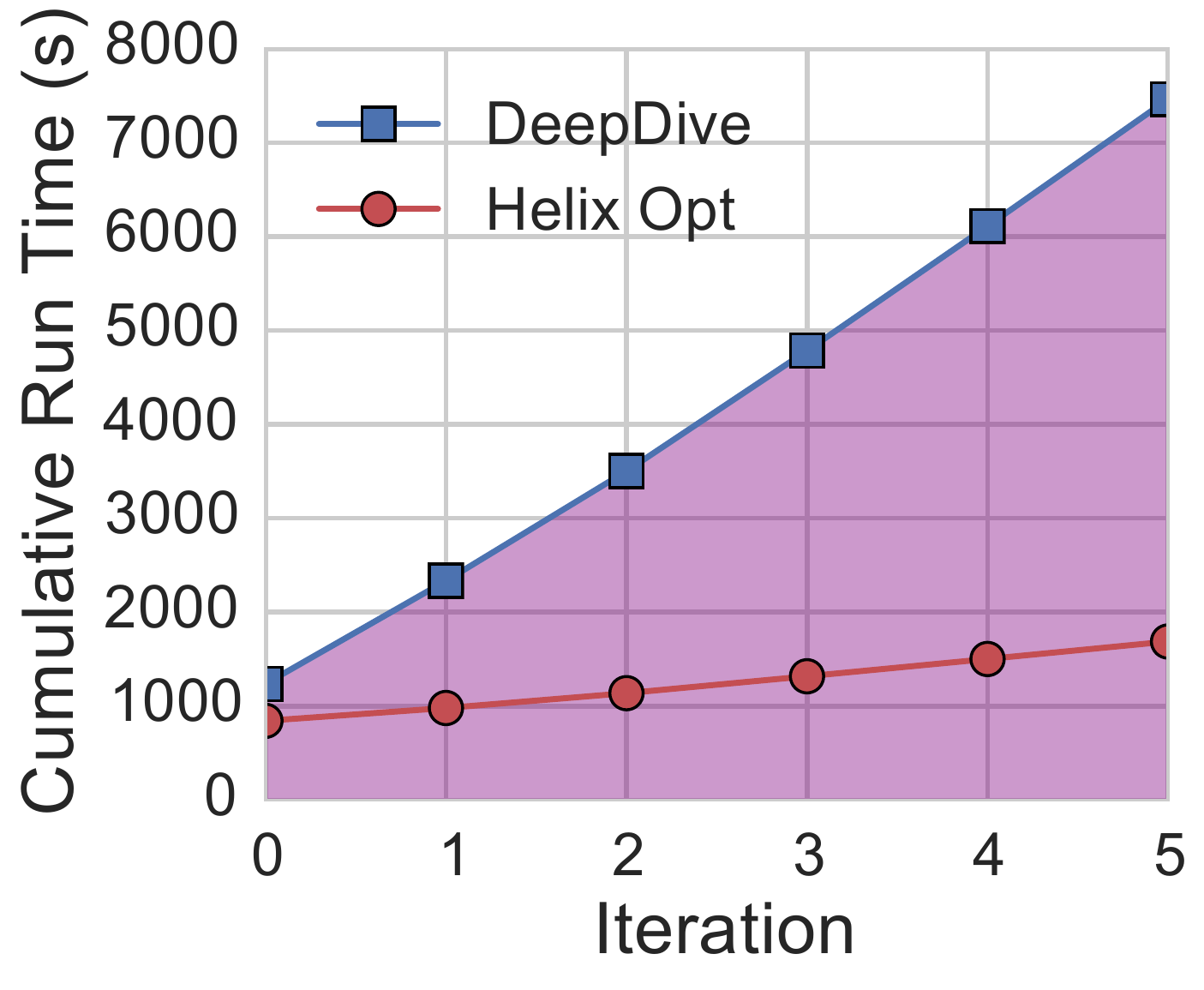}
        \label{fig:nlpSys}
    \end{subfigure}
    \hfill
    \begin{subfigure}[t]{0.22\textwidth}
        \centering
        \caption{MNIST}
        \includegraphics[width=\textwidth]{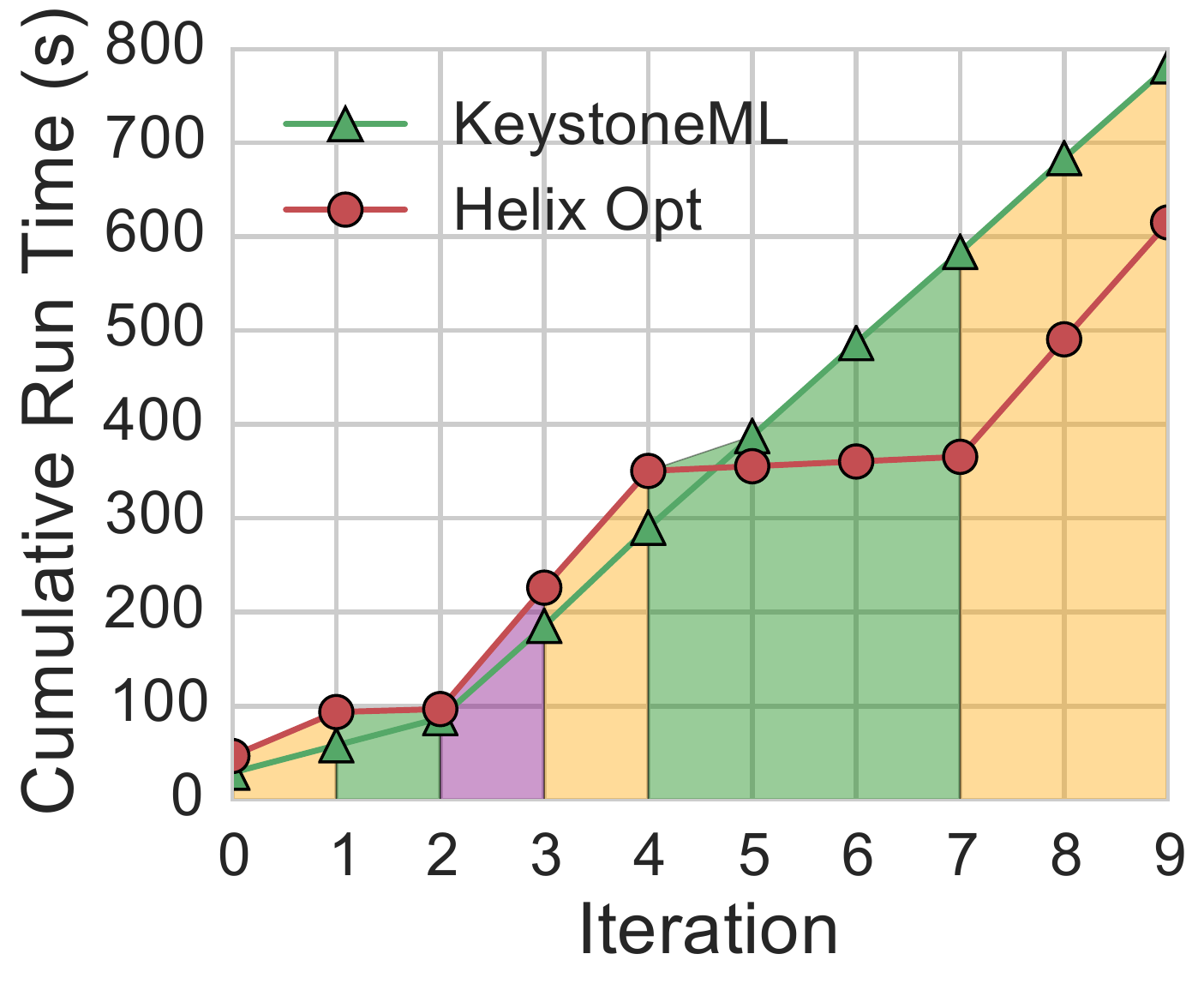}
        \label{fig:mnistSys}
    \end{subfigure}
\caption{Cumulative run time for the four workflows.
The color under the curve indicates the type of change in each iteration:
purple for DPR, orange for L/I, and green for PPR.
}
\label{fig:sysBaselines}
\end{figure}

\begin{figure}[t]
\centering

	\begin{subfigure}[t]{0.23\textwidth}
    \caption{Census}
        \includegraphics[width=\textwidth]{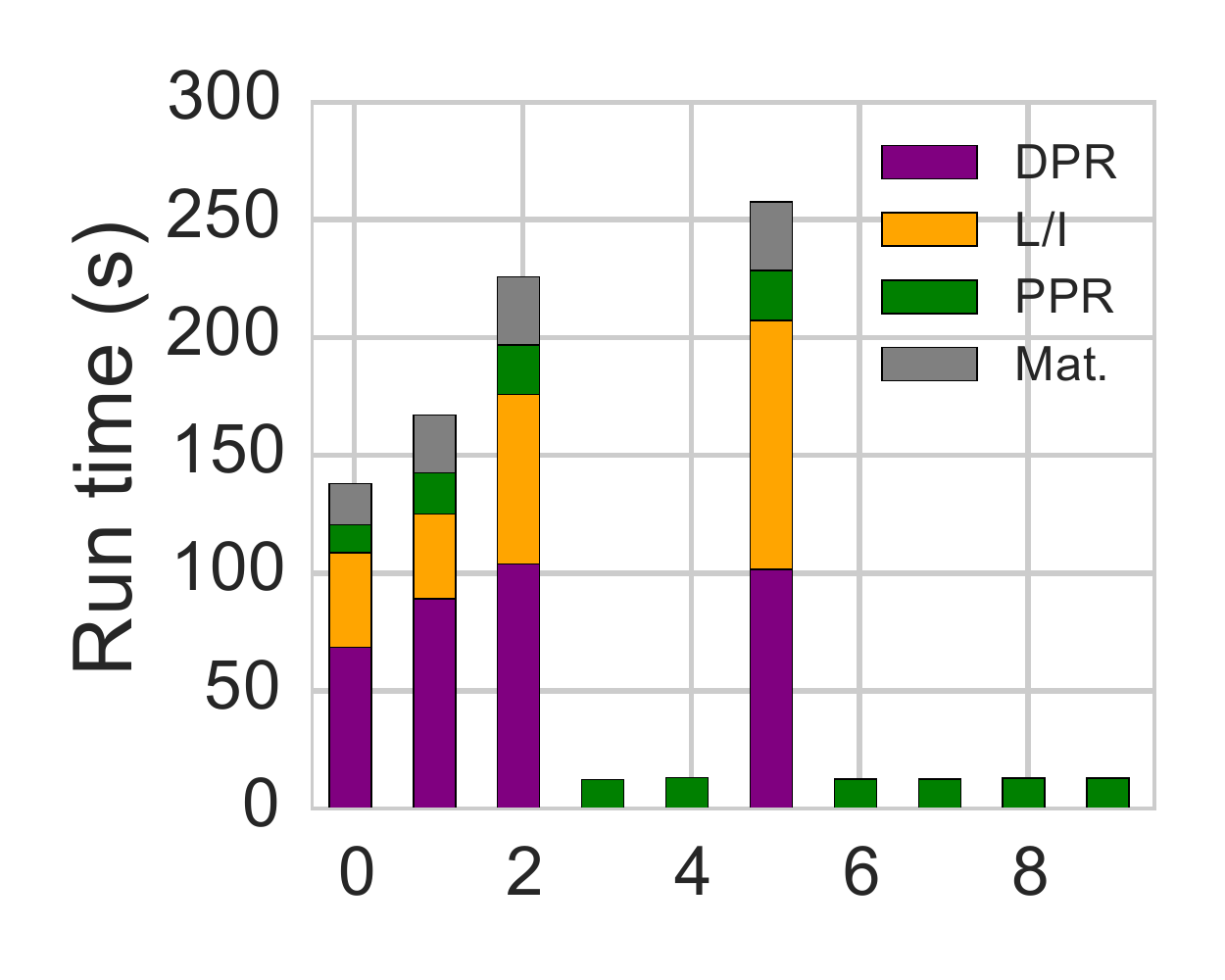}
        \label{fig:censusBreak}
    \end{subfigure}
    \hfill
    \begin{subfigure}[t]{0.23\textwidth}
    \caption{Genomics}
        \includegraphics[width=\textwidth]{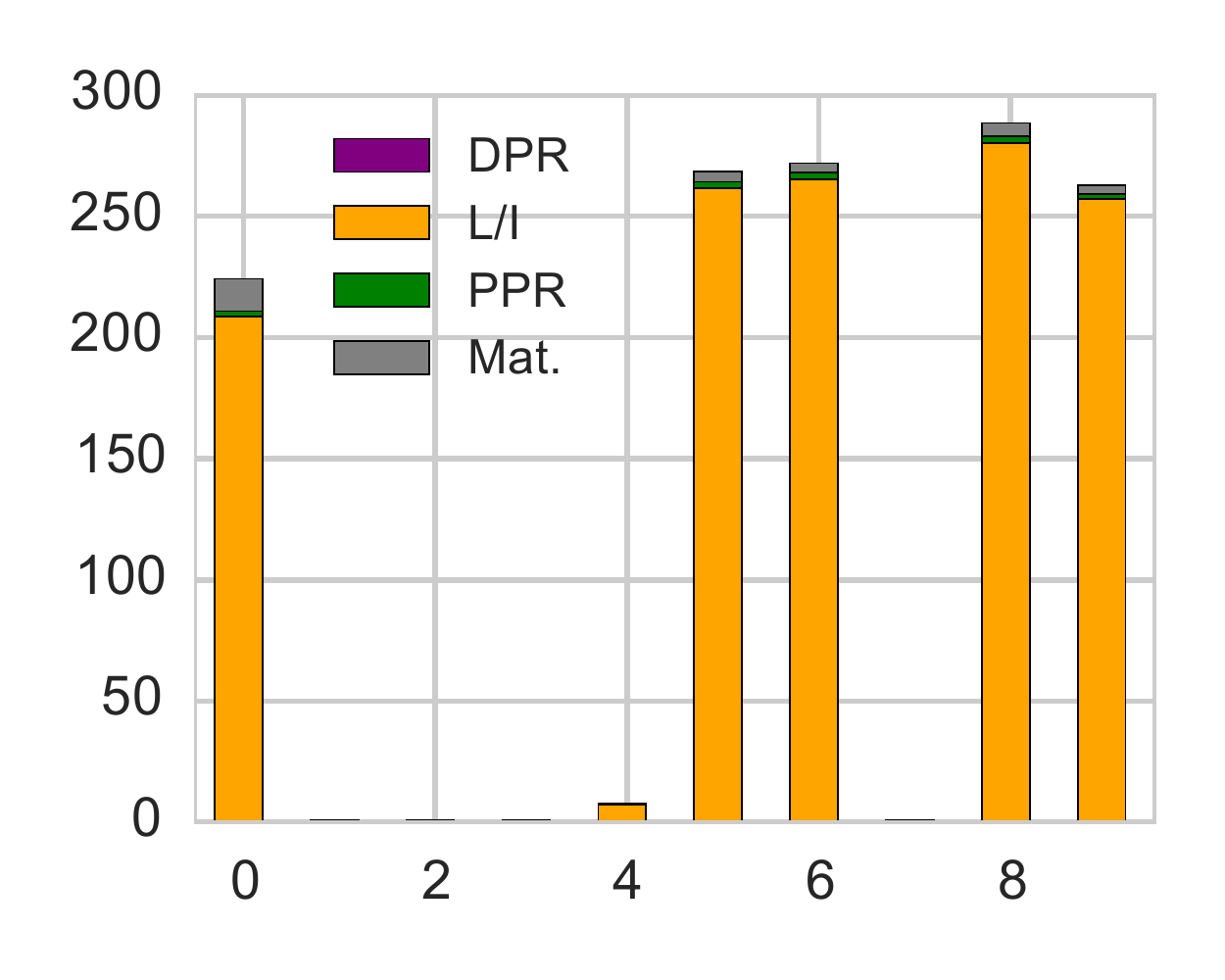}
        \label{fig:knowengBreak}
    \end{subfigure}
    
    \begin{subfigure}[t]{0.23\textwidth}
        \caption{NLP}
        \includegraphics[width=\textwidth]{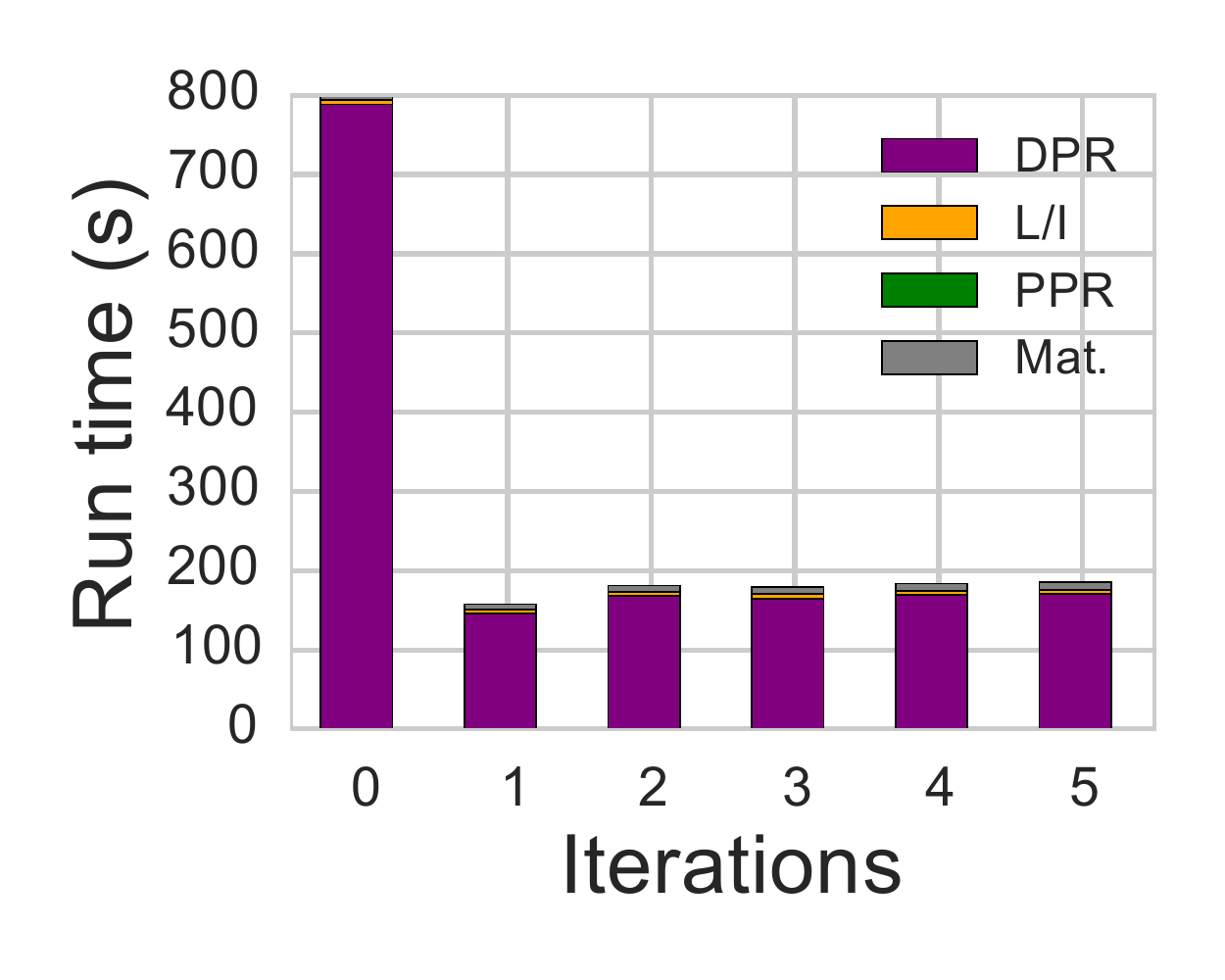}
        \label{fig:nlpBreak}
    \end{subfigure}
    \hfill
    \begin{subfigure}[t]{0.23\textwidth}
        \caption{MNIST}
        \includegraphics[width=\textwidth]{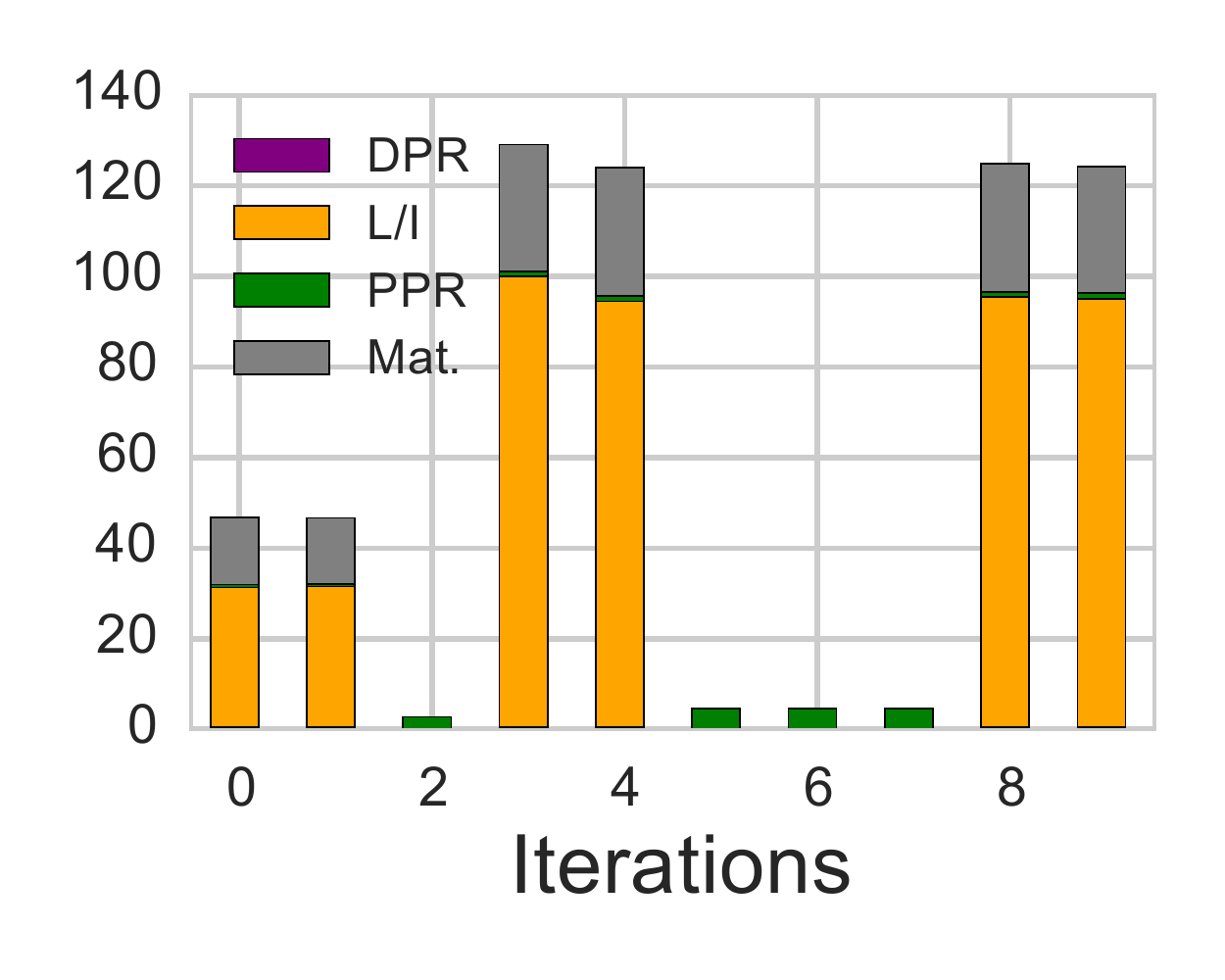}
        \label{fig:mnistBreak}
    \end{subfigure}
\caption{Run time breakdown by workflow component and materialization time per iteration for \name.}
\label{fig:breakDown}
\end{figure}

Figure~\ref{fig:sysBaselines}
shows the cumulative run time 
for all four workflows. 
The x-axis shows the iteration number, while
the y-axis shows the cumulative run time in log scale at the $i$th iteration.
Each point represents the cumulative run time of the first $i$ iterations. 
The color under the curve indicates
the workflow component modified in each iteration 
(purple = DPR, orange = L/I, green = PPR).
For example, the DPR component was modified 
in the first iteration of Census. 
\revision{Figure~\ref{fig:breakDown} shows the breakdown by workflow components and materialization for the individual iteration run times in \name,
with the same color scheme as in Figure~\ref{fig:sysBaselines} for the workflow components and gray for materialization time.}

\topic{Census}
As shown in Figure~\ref{fig:censusSys}, 
the census workflow has the largest cumulative run time gap 
between \opt and the competitor systems---{\em \opt
is $\mathbf{19\times}$ faster than KeystoneML 
as measured by cumulative run time over 10 iterations.}
By materializing and reusing intermediate results \opt is able
to substantially reduce cumulative run-time relative to other systems. 
\revision{Figure~\ref{fig:censusBreak} shows that 
1) on PPR iterations \name recomputes only the PPR;
2) the materialization of L/I outputs, which allows the pruning of DPR and L/I in PPR iterations, 
takes considerably less time than the compute time for DPR and L/I;
3) \opt reruns DPR in iteration 5 (L/I) because \opt avoided materializing the large DPR output in a previous iteration.}
For the first three iterations,
which are DPR (the only type of iterations DeepDive supports),
the $2\times$ reduction between \opt and DeepDive 
is due to the fact that DeepDive does \dataprep with Python and shell scripts,
while \opt uses Spark.
While both KeystoneML and \opt use Spark, 
\revision{KeystoneML takes longer on DPR and L/I iterations than\opt 
due to a longer L/I time 
incurred by its caching optimizer's failing to cache the training data for learning.
}
\techreport{The dominant number of PPR iterations for this workflow 
reflects the fact that users in the social sciences 
conduct extensive fine-grained analysis of results, 
per our literature survey~\cite{xin2018developers}. }

\topic{Genomics}
In Figure~\ref{fig:knowengSys}, \opt shows a $\mathbf{3\times}$ speedup
over KeystoneML for the genomics workflow.
The materialize-nothing strategy in KeystoneML 
clearly leads to no run time reduction in subsequent iterations.
\opt, on the other hand, shows a per-iteration run time 
that is proportional to the number of operators 
affected by the change in that iteration.
\revision{Figure~\ref{fig:knowengBreak} shows that 
1) in PPR iterations \opt has near-zero run time, enabled by a small materialization time in the prior iteration;
2) one of the ML models takes considerably more time, 
and \opt is able to prune it in iteration 4 since it is not changed.
}

\topic{NLP}
Figure~\ref{fig:nlpSys} shows that the cumulative run time for both DeepDive and \opt
increases linearly with iteration for the NLP workflow,
but at a much higher rate for DeepDive than \opt.
This is due to the lack of automatic reuse in DeepDive.
The first operator in this workflow is a time-consuming NLP parsing operator,
whose results are reusable for all subsequent iterations.
While both DeepDive and \opt materialize this operator in the first iteration,
DeepDive does not automatically 
reuse the results.
\revision{\opt, on the other hand, consistently prunes this NLP operation in all subsequent iterations,
as shown in Figure~\ref{fig:nlpBreak}, leading to large run time reductions 
in iterations 1-5 and thus a large cumulative run time reduction.
}

\topic{MNIST}
Figure~\ref{fig:mnistSys} shows the cumulative run times for the MNIST workflow.
As mentioned above, the MNIST workflow has nondeterministic \dataprep,
which means any changes to the DPR and L/I components 
prevents safe reuse of any intermediate result.
However, iterations containing only PPR changes 
can reuse intermediates for DPR and L/I
had they been materialized previously. 
Furthermore, we found that the DPR run time is short
but cumulative size of all DPR intermediates is large.
Thus, materializing all these DPR intermediates would incur a large run time overhead.
KeystoneML, which does not materialize any intermediate results,
shows a linear increase in cumulative run time due to no reuse.
\opt, on the other hand, 
only shows slight increase in runtime over KeystoneML for DPR and L/I iterations
because it is only materializing the L/I results on these iterations,
not the nonreusable, large DPR intermediates.
\revision{In Figure~\ref{fig:mnistBreak}, we see 
1) DPR operations take negligible time, and \opt avoids wasteful materialization of their outputs;
2) }
the materialization time taken in the DPR and L/I iterations 
pays off for \opt in PPR iterations, 
which take negligible run time due to reuse.

\begin{figure}[t]
\centering

	\begin{subfigure}[t]{0.22\textwidth}
    \caption{}
        \includegraphics[width=\textwidth]{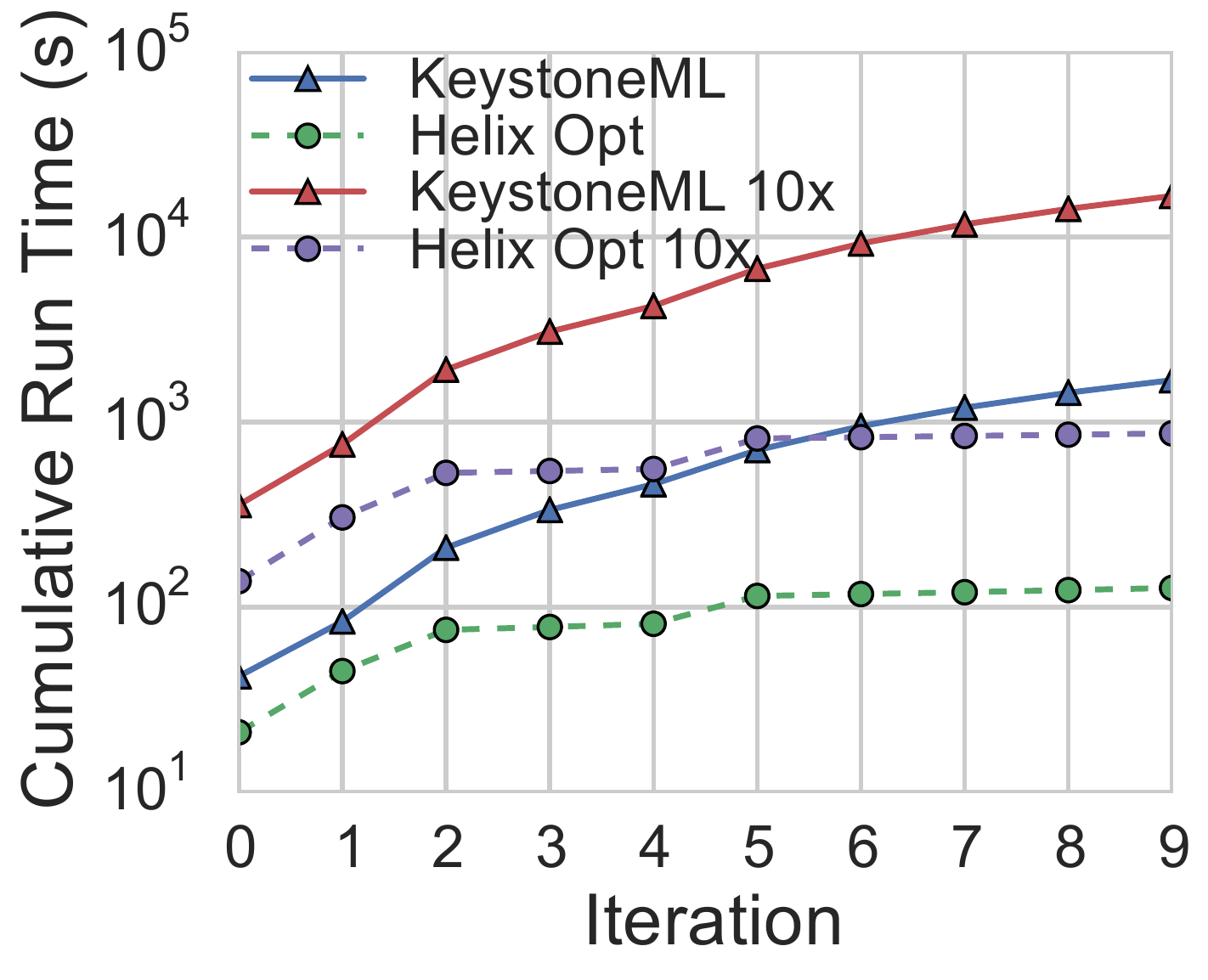}
        \label{fig:censusKSBreak}
    \end{subfigure}
    \hfill
    \begin{subfigure}[t]{0.22\textwidth}
    \caption{}
        \includegraphics[width=\textwidth]{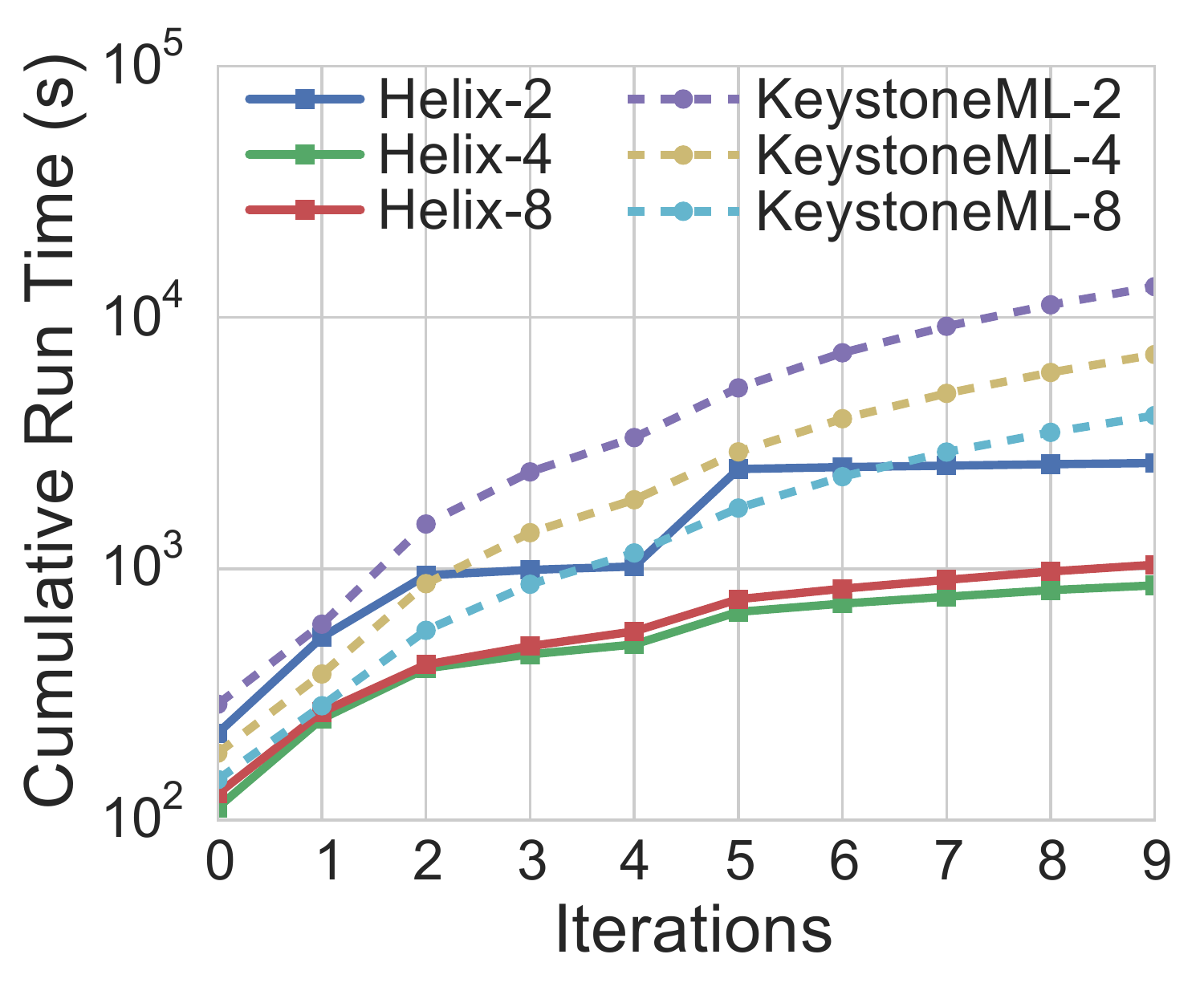}
        \label{fig:twoX}
    \end{subfigure}
\caption{
\revision{a) Census and Census 10x cumulative run time for \name and KeystoneML on a single node; b) Census 10x cumulative run time for \name and KeystoneML on different size clusters.}
}
\label{fig:scale}
\end{figure}

\revision{
\subsubsection{Scalability vs. KeystoneML}
\label{sec:scale}
}

\begin{figure}[h]
\centering
\includegraphics[width=0.45\textwidth]{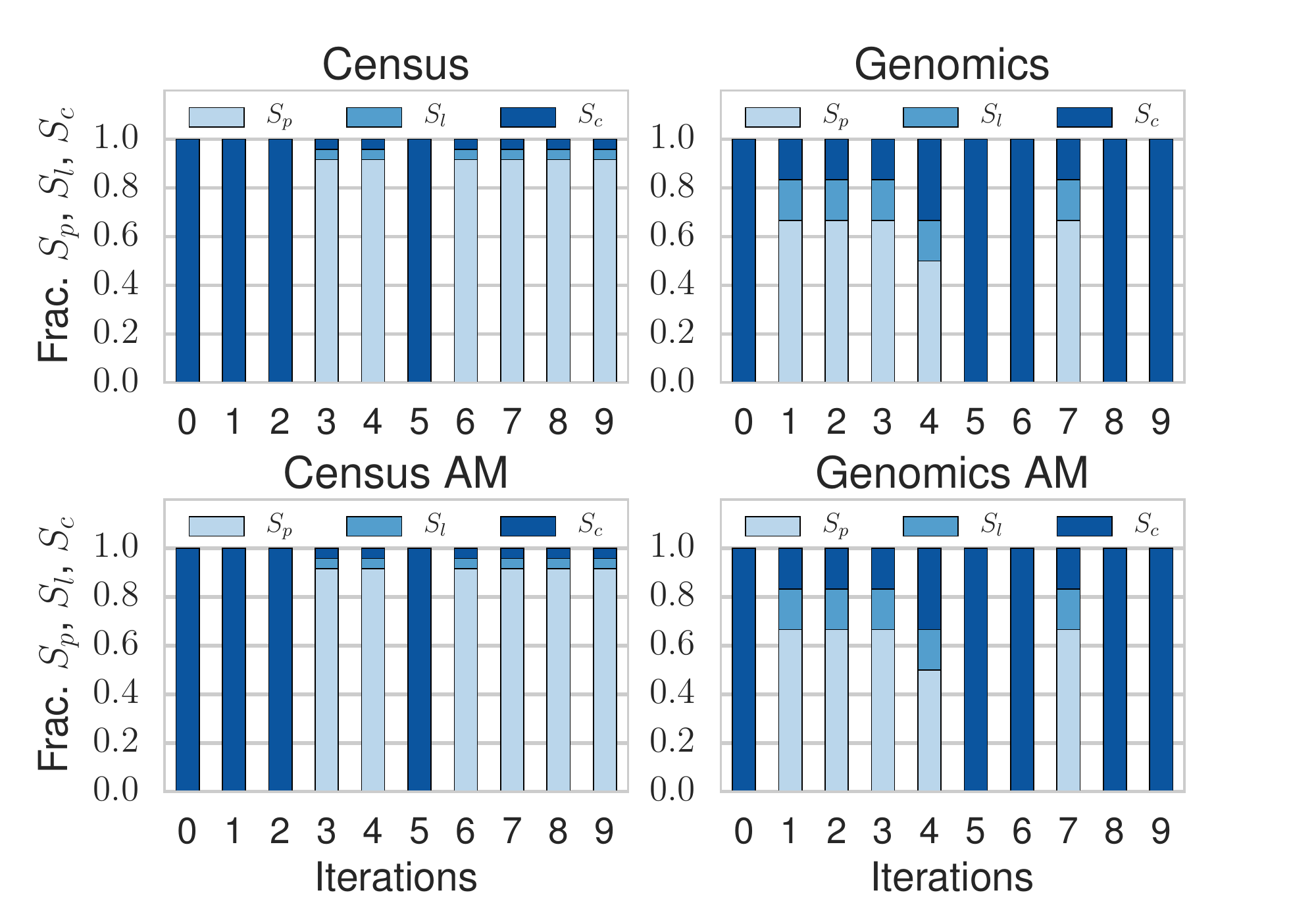}
\caption{\revision{Fraction of states in $S_p, S_l, S_c$ as determined by Algorithm~\ref{algo:oep} for the Census and Genomics workflows for \opt and \am.}}
\label{fig:frac}
\end{figure}

\topic{\revision{Dataset Size}}
\revision{We test scalability of \name and KeystoneML with respect to dataset size 
by running the ten iterations in Figure~\ref{fig:censusSys} of the Census Workflow
on two different sizes of the input. 
Census 10x is obtained by replicating Census ten times in order to preserve the learning objective.
Figure~\ref{fig:censusKSBreak} shows run time performance of \name and KeystoneML on the two datasets on a single node.
Both yield 10x speedup over the smaller dataset, scaling linearly with input data size, but \name continues to dominate KeystoneML.}

\topic{\revision{Cluster}}
\revision{
We test scalability of \name and KeystoneML with respect to cluster size by running the same ten iterations in Figure~\ref{fig:censusSys} on Census 10x described above. 
Using a uniform set of machines, we create clusters with 2, 4, and 8 workers
and run \name and KeystoneML on each of these clusters to collect cumulative run time.}

\revision{Figure~\ref{fig:twoX} shows that 
1) \name has lower cumulative run time than KeystoneML on the same cluster size, 
consistent with the single node results;
2) KeystoneML achieves $\approx 45\%$ run time reduction when the number of workers is doubled, scaling roughly linearly with the number of workers;
3) From 2 to 4 workers, \name achieves up to $75\%$ run time reduction
4) From 4 to 8 workers, \name sees a slight increase in run time.
Recall from Section~\ref{sec:interface} that the semantic unit data structure in \lang 
allows multiple transformer operations 
(e.g., indexing, computing discretization boundaries) to be learned using a single pass over the data via loop fusion.
This reduces the communication overhead in the cluster setting, 
hence the super linear speedup in 3). 
On the other hand, the communication overhead for PPR operations outweighs the benefits of distributed computing, hence the slight increase in 4).
}

\subsection{Evaluation vs. Simpler {\large \name} Versions}
\label{sec:hBaselines}

\begin{figure}[t]
\centering
    
    \begin{subfigure}[t]{0.22\textwidth}
    \centering
        \caption{Census}
        \includegraphics[width=\textwidth]{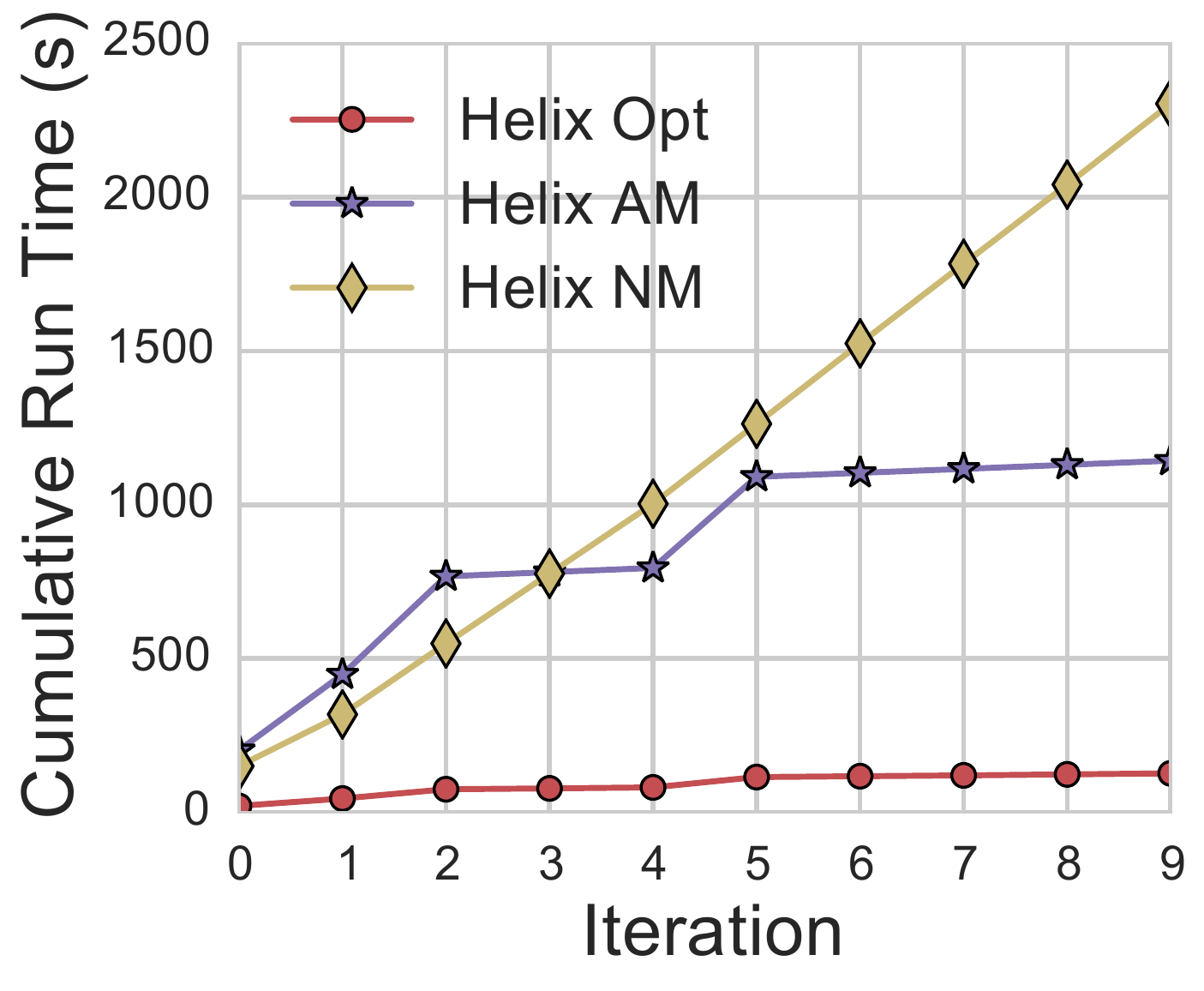}
        \label{fig:censusH}
    \end{subfigure}
    \hfill
    \begin{subfigure}[t]{0.22\textwidth}
        \centering
    \caption{Genomics}
        \includegraphics[width=\textwidth]{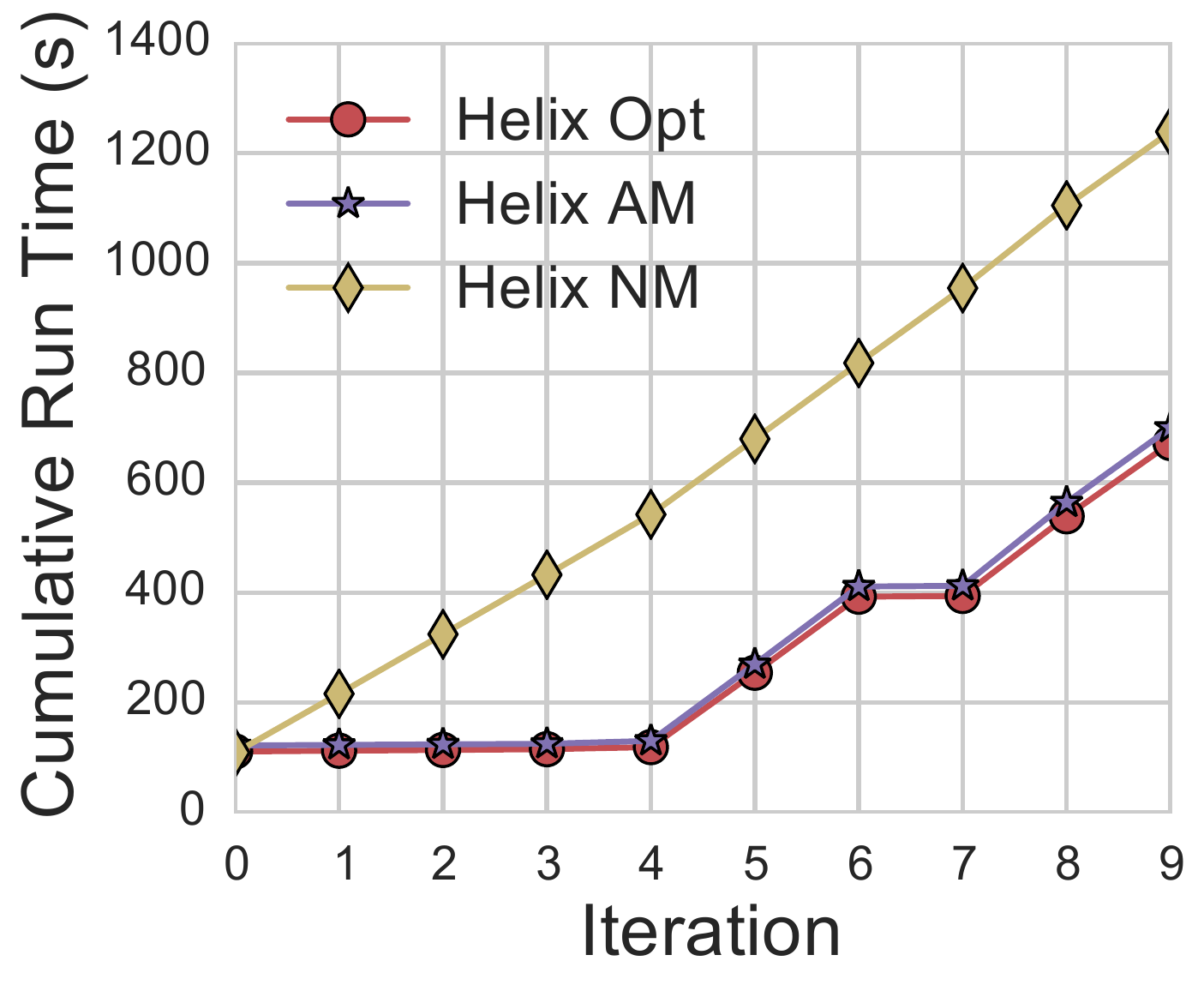}
        \label{fig:knowengH}
    \end{subfigure}
    
    \begin{subfigure}[t]{0.22\textwidth}
        \centering
        \caption{Census Storage}
        \includegraphics[width=\textwidth]{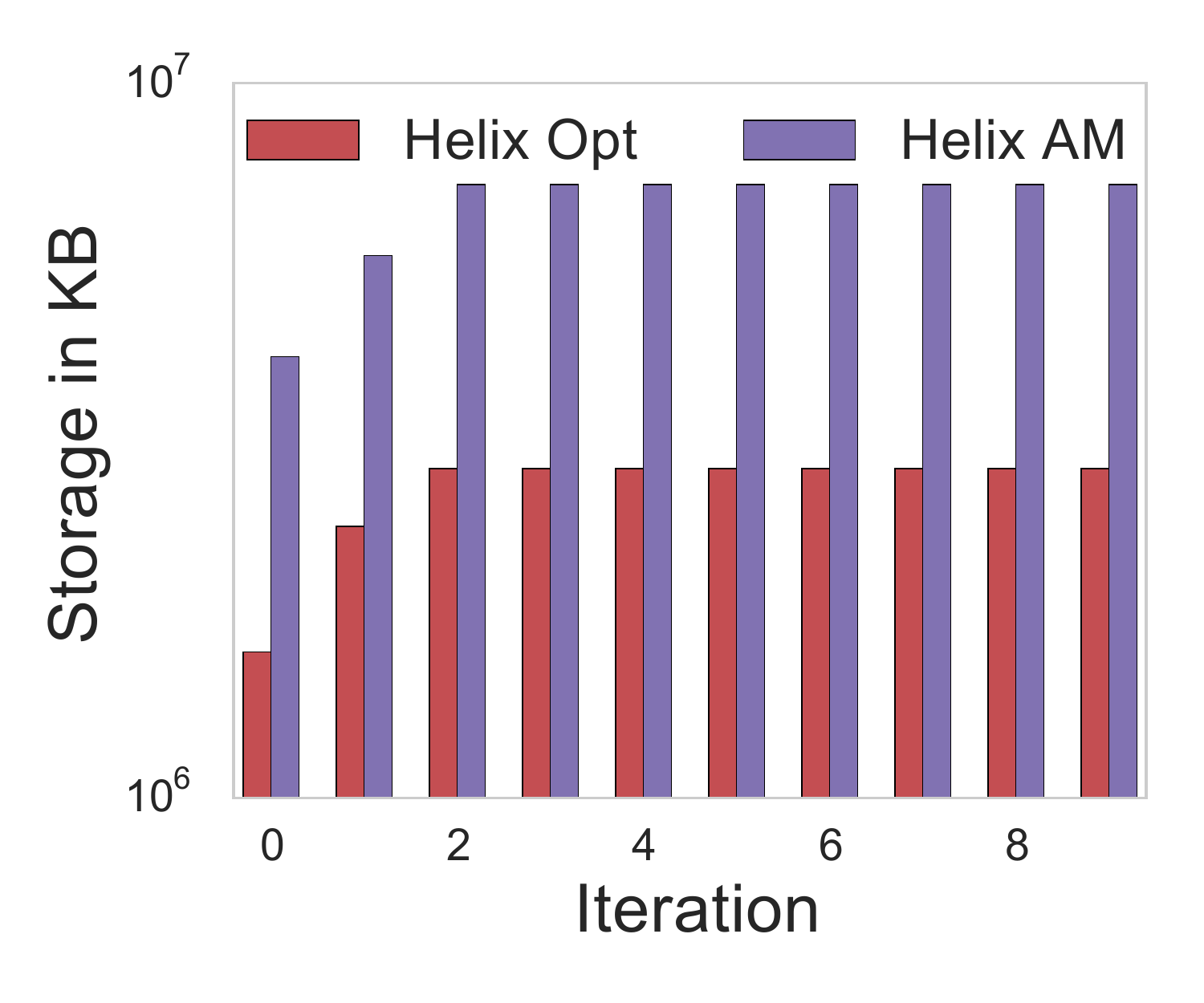}
        \label{fig:censusStoreH}
    \end{subfigure}
    \hfill
    \begin{subfigure}[t]{0.22\textwidth}
        \centering
    \caption{Genomics Storage}
        \includegraphics[width=\textwidth]{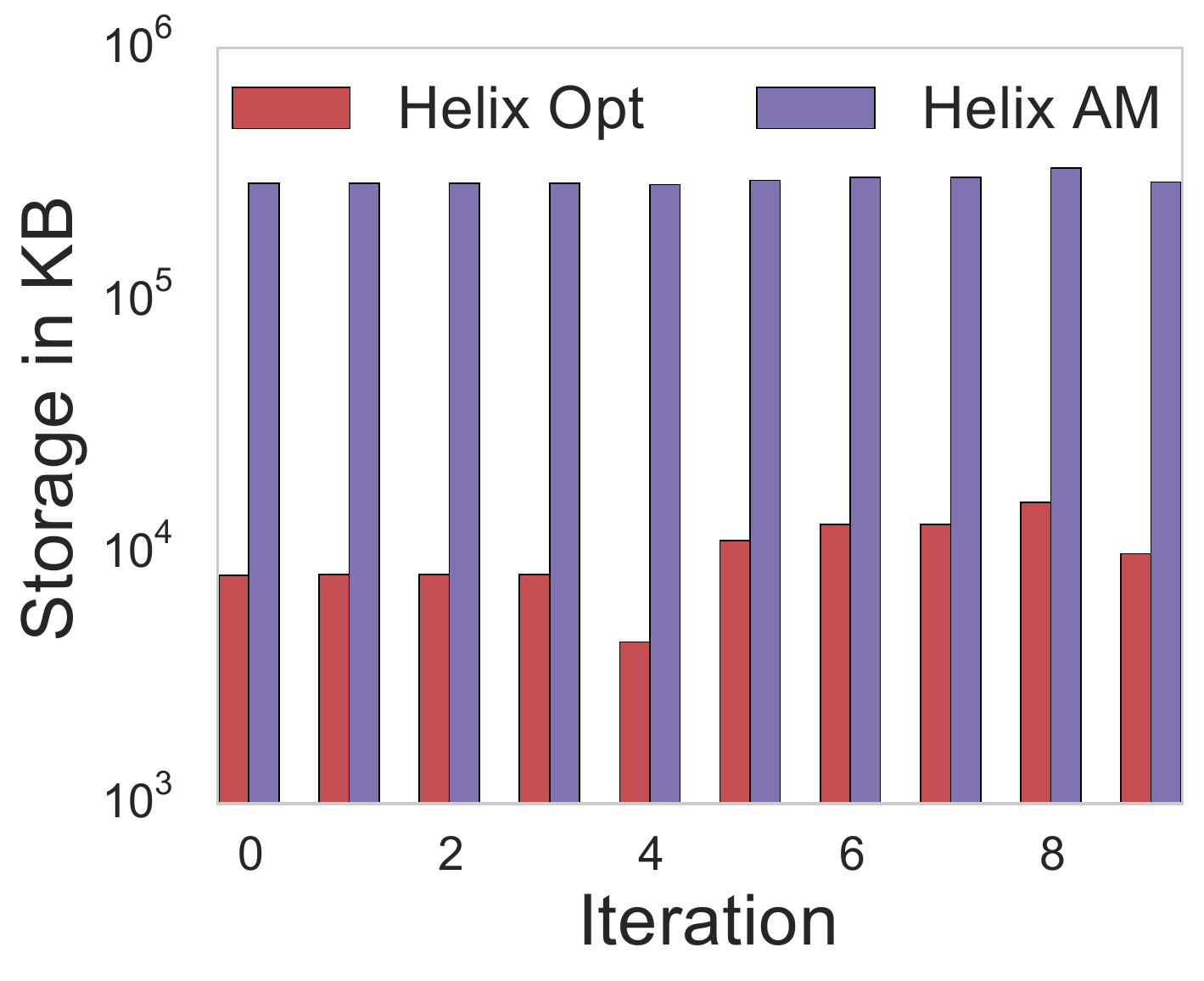}
        \label{fig:knowengStoreH}
    \end{subfigure}

    \begin{subfigure}[t]{0.22\textwidth}
    \caption{NLP}
        \includegraphics[width=\textwidth]{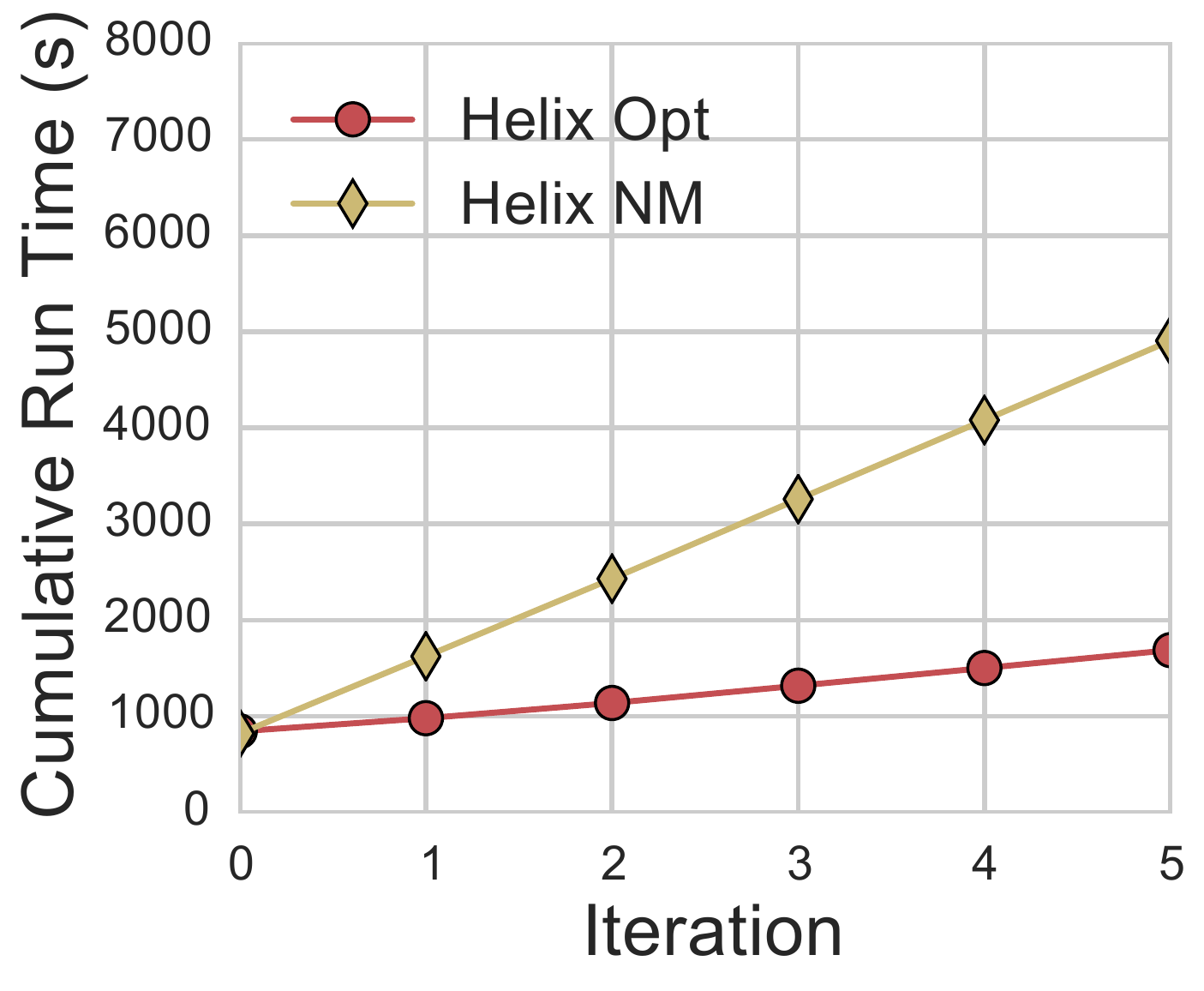}
        \label{fig:nlpH}
    \end{subfigure}
    \hfill
    \begin{subfigure}[t]{0.22\textwidth}
        \caption{MNIST}
        \includegraphics[width=\textwidth]{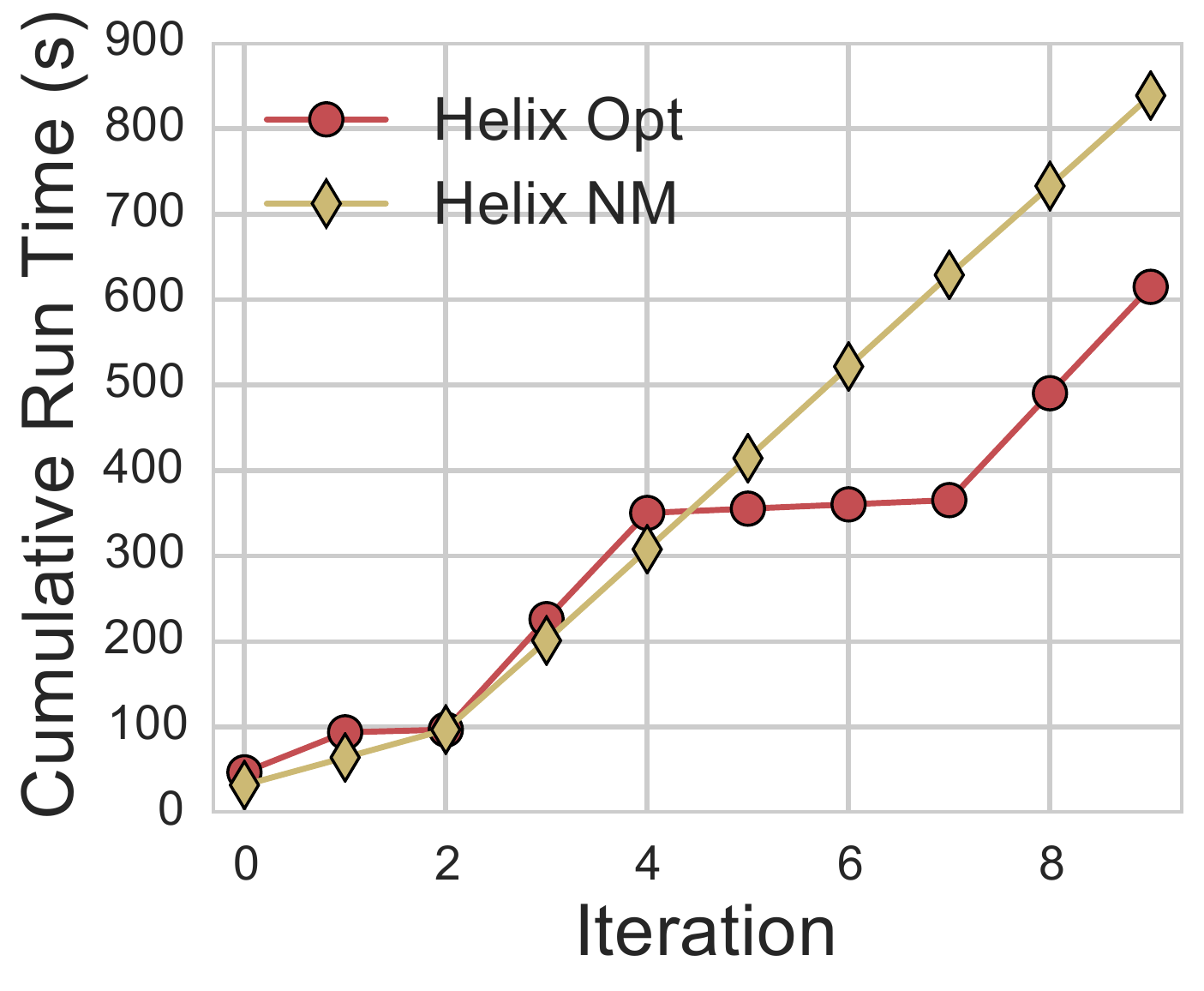}
        \label{fig:mnistH}
    \end{subfigure}
\caption{Cumulative run time and storage use against materialization heuristics
on the same four workflows as in Figure~\ref{fig:sysBaselines}.}
\label{fig:hBaselines}
\end{figure}

\frameme{
\opt achieves the lowest cumulative run time on all workflows compared to 
simpler versions of \name. 
\am often uses more than 30$\times$ the storage of \opt 
when able to complete in a reasonable time, 
while not being able to complete within 50$\times$ of the time
taken by \opt elsewhere.
\nm takes up to 4$\times$ the time taken by \opt. 
}

Next, we evaluate the effectiveness of Algorithm~\ref{algo:somp}
at approximating the solution to the NP-hard \omp problem.
We compare \opt that runs Algorithm~\ref{algo:somp} 
against:
\am that replaces Algorithm~\ref{algo:somp} 
with the policy to always materialize every operator,
and \nm that never materializes any operator.
The two baseline heuristics present two performance extremes:
\am maximizes storage usage, 
time for materialization, 
and the likelihood of being able to reuse unchanged results,
whereas \nm minimizes all three quantities.
\am provides the most flexible choices for reuse.
On the other hand, \nm has no materialization time overhead but also offers no reuse.

Figures~\ref{fig:censusH},~\subref{fig:knowengH},~\subref{fig:nlpH}, and~\subref{fig:mnistH}
show the cumulative run time 
on the same four workflows as in Figure~\ref{fig:sysBaselines} for the three variants.

{\em \am is absent from Figures~\ref{fig:nlpH} and~\subref{fig:mnistH} because
it did not complete within $50\times$ the time it took for other variants.}
The fact that \am failed to complete for the MNIST and NLP workflows
demonstrate that indiscriminately materializing all intermediates can cripple performance.
Figures~\ref{fig:nlpH} and \subref{fig:mnistH} show 
that \opt achieves substantial run time reduction over 
\nm using very little materialization time overhead 
(where the red line is above the yellow line).

For the census and genomics workflows where the materialization time is not prohibitive,
Figures~\ref{fig:censusH} and~\subref{fig:knowengH} show that
in terms of cumulative run time, 
\opt outperforms \am, which attains the best reuse as explained above.
We also compare the storage usage by \am and \nm for these two workflows.
Figures~\ref{fig:censusStoreH} and~\subref{fig:knowengStoreH} show
the storage size snapshot at the end of each iteration.
The x-axis is the iteration numbers,
and the y-axis is the amount of storage (in KB) in log scale.
The storage use for \nm is omitted from these plots because it is always zero.

We find that \opt outperforms \am while using less 
than half the storage used by \am for the census workflow in Figure~\ref{fig:censusStoreH} 
and $\frac{1}{30}$ the storage of \am for the genomics workflow in Figure~\ref{fig:knowengStoreH}.
Storage is not monotonic 
because \name purges any previous materialization of original operators prior to execution,
and these operators may not be chosen for materialization after execution,
thus resulting in a decrease in storage.

\revision{
Furthermore, to study the optimality of Algorithm~\ref{algo:somp}, 
we compare the distribution of nodes in the prune, reload, and compute states $S_p, S_l, S_c$
between \opt and \am for workflows with \am completed in reasonable times. 
Since everything is materialized in \am, it achieves maximum reuse in the next iteration.
Figure~\ref{fig:frac} shows that \opt enables the exact same reuse as \am, 
demonstrating its effectiveness on real workflows.
}

{\em Overall, neither \am nor \nm is the dominant strategy in all scenarios,
and both can be suboptimal in some cases.}

\techreport{
\revision{
\subsection{Memory Usage by  {\large \name}}
\label{sec:memory}
}

\begin{figure}
\includegraphics[width=0.5\textwidth]{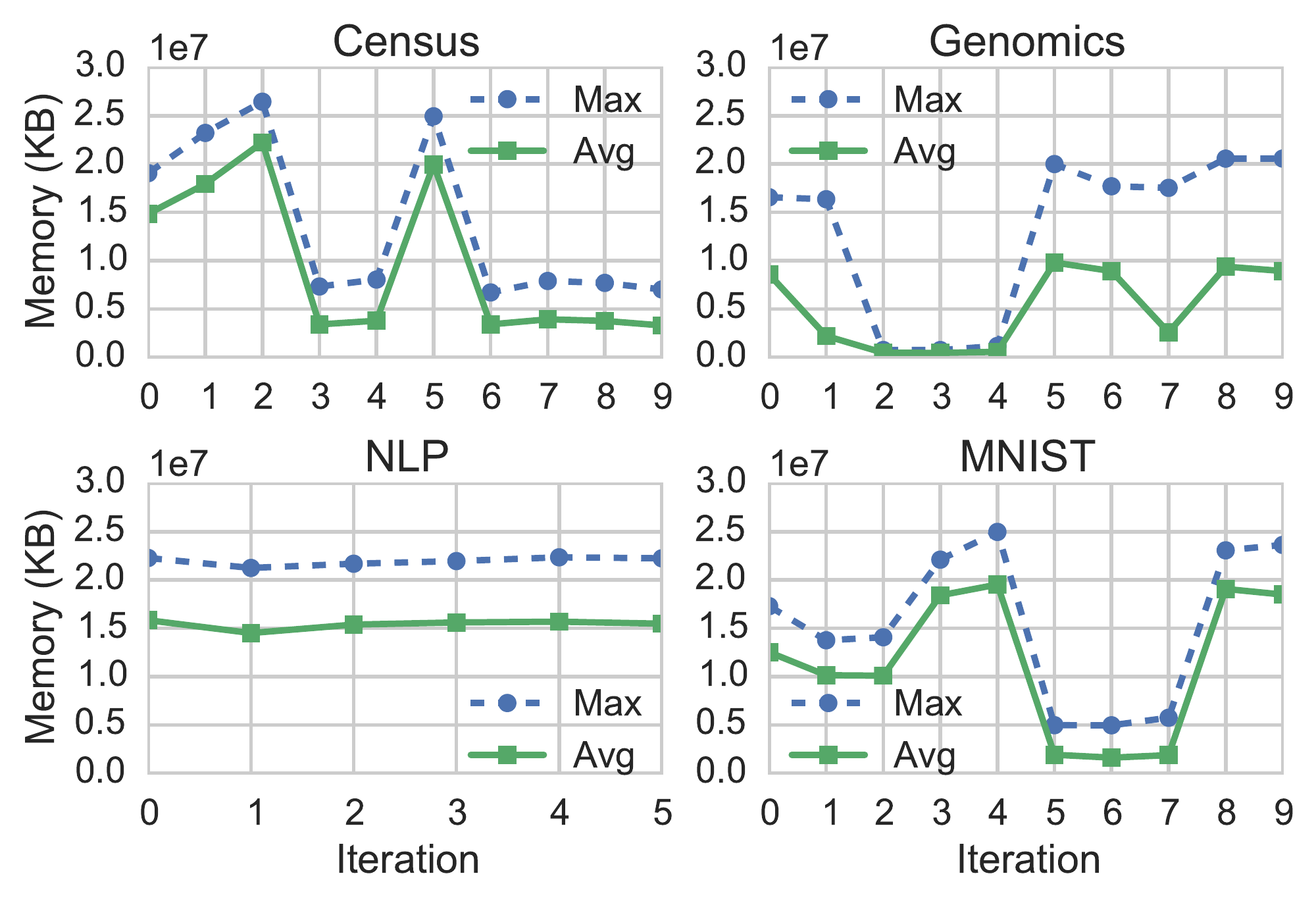}
\caption{\revision{Peak and average memory for \name.}}
\label{fig:mem}
\end{figure}

\revision{
We evaluate memory usage by \name to ensure 
that its materialization and reuse benefits 
do not come at the expense of large memory overhead.
We measure memory usage at one-second intervals during \name workflow execution.
Figure~\ref{fig:mem} shows the peak and average memory used by \name in each iteration
for all four workflows.
We allocate 30GB memory (25\% of total available memory) in the experiments. 
We observe that \name runs within the memory constraints on all workflows.
Furthermore, on iterations where \name reuses intermediate results to
achieve a high reduction in run time compared to other systems,
memory usage is also significantly reduced.
This indicates that \name reuses small intermediates 
that enable the pruning of a large portion of the subgraph to reduce run time, 
instead of overloading memory.}
}

%% file: related2.tex

Many systems have been developed in recent years 
to better support ML workflows.
\revision{We begin by describing ML systems and other general
workflow management tools}, followed
by systems that target the reuse of intermediate results. 

\stitle{Machine Learning Systems.} 
We describe machine learning systems 
that support declarative programming,
followed by other general-purpose systems
that optimize across frameworks.

\emtitle{Declarative Systems.}
Due to the challenges
in developing ML workflows, there has been recent efforts 
to make it easier to do so declaratively.
Boehm et al.\ categorize
declarative ML systems into three groups based on the usage:
{\em declarative ML algorithms}, {\em ML libraries},
and {\em declarative ML tasks}~\cite{boehm2016declarative}.
Systems that support {\em declarative ML algorithms}, 
such as TensorFlow~\cite{abadi2016tensorflow}, 
SystemML~\cite{ghoting2011systemml}, OptiML~\cite{sujeeth2011optiml}, ScalOps~\cite{weimer2011machine}, 
and SciDB~\cite{stonebraker2011architecture},
allow ML experts to program new ML algorithms,
by declaratively specifying linear algebra and statistical operations
at higher levels of abstraction. Although it also
builds a computation graph like \name, TensorFlow has no intermediate reuse\techreport{ and
always performs a full computation e.g. any in-graph data preparation}.
TensorFlow's lower level linear algebra operations are not conducive to data preprocessing. \name handles reuse at a higher level than TensorFlow ops.
{\em ML libraries}, 
such as Mahout~\cite{owen2012mahout}, Weka~\cite{hall2009weka}, 
GraphLab~\cite{low2012distributed}, Vowpal Wabbit~\cite{langford2007vowpal}, 
MLlib~\cite{meng2016mllib} and Scikit-learn~\cite{pedregosa2011scikit},
provide \revision{simple interfaces to} optimized implementations 
of popular ML algorithms. TensorFlow has also
recently started providing TFLearn~\cite{damien2016tflearn},
a high level ML library targeted at deep learning.
Systems that support {\em declarative ML tasks}
allow application developers with limited ML knowledge
to develop models using higher-level primitives than in declarative ML algorithms.
\name falls into this last group of systems,
along with DeepDive~\cite{zhang2015deepdive,deepdive2016} and
KeystoneML~\cite{sparks2016end}. 
\techreport{These systems
perform workflow-level optimizations to reduce end-to-end execution time.
\revision{Finally, at the extreme end of this spectrum are systems for in-RDBMS
analytics~\cite{hellerstein2012madlib,feng2012bismarck,wang2011hybrid}
that extend databases to support
ML, at great cost to flexibility.}}

Declarative ML task systems, like 
\name, can seamlessly make use of improvements in ML library implementations, 
such as MLlib~\cite{meng2016mllib}, CoreNLP~\cite{manning2014stanford} 
and DeepLearning4j~\cite{team2016deeplearning4j}, within UDF calls.
Unlike declarative ML algorithm systems, 
that are targeted at ML experts and researchers,
these systems focus on end-users of ML.

\emtitle{Systems that Optimize Across Frameworks.}
These systems target a broad range of use-cases, including
ML. Weld~\cite{palkar2017weld} and Tupleware~\cite{crotty2015architecture} 
optimize UDFs written in different frameworks
by compiling them down to a common intermediate representation.
Declarative ML task systems like \name can take advantage of 
the optimized UDF implementations; unlike \name,
these systems do not benefit from seamless specification,
execution, and end-to-end optimizations across workflow components
that come from a unified programming model.

\emtitle{\revision{Systems for Optimizing Data Preprocessing.}}
\revision{The database community has identified various opportunities for
optimizing DPR.
Several approaches identify as a key bottleneck in DPR and optimize it~\cite{kumar2015learning,chen2017towards,olteanu2016f,kumar2016join}.
\techreport{
Kumar et al.~\cite{kumar2015learning}
optimizes generalized linear models directly over
factorized / normalized representations of relational data,
avoiding key-foreign key joins.
Morpheus~\cite{chen2017towards} and F~\cite{olteanu2016f}
extend this factorized approach to general linear algebra operations
and linear regression models, respectively (the latter over arbitrary joins).
Some work~\cite{kumar2016join} even attempts to characterize
when joins can be eschewed altogether, without sacrificing performance.
All of these optimizations are orthogonal to those used by \name.}
\revision{Another direction aims at reducing the manual effort
involved in data cleaning and feature
engineering~\cite{ratner2017snorkel,zhang2016materialization,krishnan2016activeclean,anderson2013brainwash,anderson2016input}. 
All of these optimizations are orthogonal to those used by \name, which targets
end-to-end iterative optimizations.}
\techreport{Snorkel~\cite{ratner2017snorkel} supports training data engineering using rules.
\textsc{Columbus}~\cite{zhang2016materialization} 
optimizes feature selection
specifically for regression models.
ActiveClean~\cite{krishnan2016activeclean}
integrates data cleaning with learning convex models,
using gradient-biased samples to identify dirty data.
Brainwash~\cite{anderson2013brainwash}
proposes to expedite feature engineering
by recommending feature
transformations.
Zombie~\cite{anderson2016input}
speeds up data preparation by learning over smaller,
actively-learned informative subsets of data during feature engineering.
These approaches are bespoke for the \dataprep portion of ML workflows
and do not target end-to-end optimizations, although there is no reason they
could not be integrated within \name.}}

\stitle{\revision{ML and non-ML Workflow Management Tools.}}
\revision{Here we discuss ML workflow systems,
production platforms for ML,
\techreport{industry} batch processing workflow systems, and systems for
scientific workflows.}

\emtitle{\revision{ML Workflow Management.}}
\revision{Prior tools for managing ML workflows
focus primarily on making their pipelines easier to debug.
For example, Gestalt~\cite{patel2010gestalt} and Mistique~\cite{vartak2018mistique} 
both tackle the problem of model diagnostics
by allowing users to inspect intermediate results.
The improved workflow components 
in these systems could be easily incorporated within \name.}

\emtitle{\revision{ML Platforms-as-Services.}}
\revision{A number of industry frameworks~\cite{baylor2017tfx,fblearner,barnes2015azure,michelangelo,sagemaker,mlflow}, attempt to automate
typical steps in deploying machine learning by providing
a Platform-as-a-Service (PaaS) capturing common use cases.}
\papertext{\revision{Instead,
\name is not designed to reduce manual
effort of model deployment, but rather model {\em development}.}}
\techreport{\revision{These systems
vary in generality --- frameworks like SageMaker, Azure Studio, and MLFlow are built around
services provided by Amazon, Microsoft, and Databricks, respectively, and provide general
solutions for production deployment of ML models for companies that in-house
infrastructure. On the other hand, TFX, FBLearner
Flow, and Michelangelo are optimized for internal use at Google, Facebook, and Uber,
respectively. For example, TFX is optimized for use with TensorFlow, and Michelangelo
is optimized for Uber's real-time requirements, allowing production models to use features
extracted from streams of live data.}}

\techreport{\revision{The underlying ``workflow'' these
frameworks manage is not always given an explicit representation,
but the common unifying thread is the automation of production deployment,
monitoring, and continuous retraining steps,
thereby alleviating engineers from the labor of ad-hoc solutions.
\name is not designed to reduce manual
effort of model deployment, but rather model {\em development}.
The workflow \name manages sits at a lower level than those of
industry PaaS systems, and therefore the techniques it leverages are quite different.}}

\emtitle{\revision{General Batch Processing Workflow Systems.}}
\revision{A number of companies have implemented workflow management
systems for batch processing~\cite{airflow,sumbaly2013big}. These systems are not
concerned with runtime optimizations, and rather provide features
useful for managing large-scale workflow complexity.}

\emtitle{\revision{Scientific Workflow Systems.}}
\revision{Some systems address the significant mental and computational
overhead associated with scientific workflows. VisTrails~\cite{callahan2006vistrails}
and Kepler~\cite{ludascher2006scientific}
add provenance and other metadata tracking\techreport{ to visualization-producing workflows,
allowing for reproducibility, easier visualization comparison, and faster iteration}.
Other systems attempt to map scientific workflows to cluster resources~\cite{yu2005taxonomy}.
One such system, Pegasus~\cite{deelman2004pegasus},
also identifies reuse opportunities when executing workflows.
The optimization techniques employed by all systems discussed leverage
reuse in a simpler manner than does \name\techreport{,
since the workflows are coarser-grained and computation-heavy, so that the cost
of loading cached intermediate results can be considered negligible}.}

\stitle{Intermediate Results Reuse.}
The OEP/OMP problems within \name are reminiscent of classical work
on view materialization in database systems~\cite{chirkova2012materialized},
but operates at a more coarse-grained level on black box operators.
\techreport{However, the reuse of intermediate results within ML workflows
differs from traditional database view materialization
in that it is less concerned with fine-grained updates,
and instead treats operator outputs as immutable black-box units
due to the complexity of the data analytics operator. }
\textsc{Columbus}~\cite{zhang2016materialization}
focuses on caching feature columns 
for feature selection exploration within a single workflow.
ReStore~\cite{elghandour2012restore} manages reuse of intermediates 
across dataflow programs written in Pig~\cite{olston2008pig}, 
while Nectar~\cite{gunda2010nectar} does so 
across DryadLINQ~\cite{yu2008dryadlinq} workflows.
Jindal et al.~\cite{jindal2018selecting} study SQL subexpression 
materialization within a single workflow with many subqueries. 
\techreport{Perez et al.~\cite{perez2014history} also study SQL subexpression materialization, 
but in an inter-query fashion that uses historical data 
to determine utility of materialization for future reuse.}
In the same vein, Mistique~\cite{vartak2018mistique} \revision{and its spiritual
predecessor Sherlock~\cite{vartak2015supporting} use} historical usage 
as part of \revision{their cost models} for adaptive materialization.
\techreport{
\name shares some similarities with the systems above 
but also differs in significant ways.
Mistique~\cite{vartak2018mistique},
Nectar~\cite{gunda2010nectar}, and ReStore~\cite{elghandour2012restore} share 
the goal of efficiently reusing ML workflow intermediates with \name.}
However, the cost models and algorithms proposed in these systems for deciding what to reuse
do not consider the operator/subquery dependencies in the DAG 
and make decisions for each operator independently 
based on availability, operator type, size, and compute time.
We have shown in Figure~\ref{fig:reduction} 
that decisions can have cascading effects on the rest of the workflow.
\techreport{The reuse problems studied in \textsc{Columbus}~\cite{zhang2016materialization} and Jindal et al.~\cite{jindal2018selecting}
differ from ours in that they are concerned with 
decomposing a set of queries $Q$ into subqueries
and picking the minimum cost set of subqueries to cover $Q$.
The queries and subqueries can be viewed as a bipartite graph,
and the optimization problem can be cast as a {\sc Set Cover}.
They do not handle iteration but rather efficient execution of parallel queries.}
\techreport{Furthermore, the algorithms for choosing what to materialize 
in Mistique~\cite{vartak2018mistique} and Perez et al.~\cite{perez2014history} 
use historical data as signals for likelihood of reuse in the future,
whereas our algorithm directly uses projected savings for the next iteration 
based on the reuse plan algorithm.
Their approaches are reactive, while ours is proactive.}

%% file: conclusion.tex

We presented \name, a declarative system aimed at
accelerating iterative ML application development.
In addition to its user friendly, flexible, and succinct programming interface,
\name tackles two major optimization problems, namely  \oep and \omp,
that together enable cross-iteration optimizations 
resulting in significant run time reduction for future iterations.
We devised a PTIME algorithm to solve \oep by using a reduction to {\sc Max-Flow}.
We showed that \omp is {\sc NP-Hard} and proposed
a light-weight, effective heuristic for this purpose.
We evaluated \name against DeepDive and KeystoneML on workflows from 
social sciences, NLP, computer vision, and natural sciences that vary greatly in characteristics 
to test the versatility of our system.
We found that \name supports a variety of diverse machine learning applications with ease
and \papertext{provides up to a 19$\times$ cumulative run time speedup relative to baselines. }\techreport{and provides {\em 40-60\% cumulative run time reduction} on complex learning tasks
and nearly {\em an order of magnitude reduction} on simpler ML tasks
compared to both DeepDive and KeystoneML.} 
While \name is implemented in a specific way, 
the techniques and abstractions presented in this work 
are general-purpose; other systems can enjoy the benefits of \name's 
optimization modules through simple wrappers and connectors.

In future work, we aim to further accelerate iterative workflow development 
via introspection and querying across workflow versions over time, 
automating trimming of redundant workflow nodes,
as well as auto-suggestion of workflow components to aid workflow development 
by novices. 
\techreport{
\revision{Specifically, \name is capable of tracing specific features in the ML model 
to the operators in the DAG.
This allows information about feature importance learned in the ML model 
to be used directly to prune the DAG.
In addition, the materialization and reuse techniques we proposed 
can be extended to optimize parallel executions of similar workflows. }}

%% file: appendix.tex

\section{\lang Specifications}
\label{sec:grammar}

\input{semantics}

\input{grammar.tex}

\section{Proof for Theorem~\ref{thm:reduction}}
\label{sec:t1proof}

For clarity, we first formulate \oep as an integer linear program 
before presenting the proof itself.

\input{ilp}

\topic{Main proof}
The proof for Theorem~\ref{thm:reduction} follows directly from the two lemmas
proven below.
Recall that given an optimal solution $A$ to PSP, 
we obtain the optimal state assignments for OEP using the following mapping:
\begin{equation}\label{eq:solMap}
s(n_i) = \begin{cases}
S_c & \text{if $a_i \in A$ and $b_i \in A$} \\
S_l & \text{if $a_i \in A$ and $b_i \not\in A$} \\
S_p& \text{if $a_i \not\in A$ and $b_i \not\in A$}
\end{cases}
\end{equation}
 
\begin{lemma}
A feasible solution to PSP under $\varphi$ also produces a feasible solution to OEP.
\end{lemma}

\begin{proof}
We first show that satisfying the prerequisite constraint in PSP 
leads to satisfying Constraint~\ref{constraint:exstate} in \oep. 
Suppose for contradiction that a feasible solution to PSP under $\varphi$
does not produce a feasible solution to OEP.
This implies that for some node $n_i \in N$ s. t. $s(n_i)=S_c$,
at least one parent $n_j$ has $s(n_j) = S_p$.
By the inverse of Eq (\ref{eq:solMap}), 
$s(n_i)=S_c$ implies that $b_i$ was selected, 
while $s(n_j) = S_p$ implies that neither $a_j$ nor $b_j$ was selected.
By construction, there exists an edge $a_j \rightarrow b_i$.
The project selection entailed by the operator states 
leads to a violation of the prerequisite constraint.
Thus, a feasible solution to PSP must produce a feasible solution to OEP under $\varphi$.
\end{proof}

\begin{lemma}
An optimal solution to PSP is also an optimal solution to OEP under $\varphi$.
\end{lemma}

\begin{proof}
Let $Y_{a_i}$ be the indicator for whether project $a_i$ is selected,
$Y_{b_i}$ for the indicator for $b_i$,
and $p(x_i)$ be the profit for project $x_i$.
The optimization object for PSP can then be written as 
\begin{equation}\label{eqn:psp}
\max\limits_{Y_{a_i}, Y_{b_i}} \sum\limits_{i=1}^{|N|} Y_{a_i} p(a_i) + Y_{b_i} p(b_i)
\end{equation}
Substituting our choice for $p(a_i)$ and $P(b_i)$, Eq (\ref{eqn:psp}) becomes
\begin{align}
& \max\limits_{Y_{a_i}, Y_{b_i}} \sum\limits_{i=1}^{|N|} -Y_{a_i} l_i + Y_{b_i} (l_i - c_i) \\
=& \max\limits_{Y_{a_i}, Y_{b_i}} - \sum\limits_{i=1}^{|N|} (Y_{a_i} - Y_{b_i}) l_i + Y_{b_i} c_i
\label{eqn:pspRed}
\end{align}
The mapping established by Eq (\ref{eq:solMap}) 
is equivalent to setting 
$X_{a_i} = Y_{a_i}$
and $X_{b_i} = Y_{b_i}$.
Thus the maximization problem in Eq (\ref{eqn:pspRed}) is equivalent to the minimization problem  in Eq (\ref{eqn:obj}),
and we obtain an optimal solution to OEP from the optimal solution to PSP.
\end{proof}

\section{Proof for Theorem~\ref{thm:nphard}}
\label{sec:npProof}

We show that OMP is NP-hard 
under restrictive assumptions about the structure of $W_{t+1}$ relative to $W_t$,
which implies the general version of OMP is also NP-hard.

In our proof we make the simplifying assumption that 
all nodes in the \wf DAG are reusable in the next iteration:
\begin{equation}\label{eqn:same}
n_i^{t} \equiv n_i^{t+1} \ \forall n_i^t \in N_t, n_i^{t+1} \in N_{t+1}
\end{equation}
Under this assumption, 
we achieve maximum reusability of materialized intermediate results
since all operators that persist across iterations $t$ and $t+1$ are equivalent.
We use this assumption to sidestep the problem of predicting iterative modifications,
which is a major open problem by itself.

Our proof for the NP-hardness of OMP subject to Eq(~\ref{eqn:same}) 
uses a reduction from the known NP-hard Knapsack problem.

\begin{problem} (Knapsack)
Given a knapsack capacity $B$ and a set $N$ of $n$ items, with each $i \in N$ having a size $s_i$ and  a profit $p_i$, find $S^* =$
\begin{equation}
\argmax\limits_{S \subseteq N} \sum\limits_{i \in S} p_i
\end{equation}
such that $\sum_{i \in S^*} s_i \leq B$.
\end{problem} 

\begin{figure}
\centering
\begin{tikzpicture}[
            > = stealth, 
            shorten > = 1pt, 
            auto,
            node distance = 1.5cm, 
            semithick 
        ]

        \tikzstyle{p}=[
            shape=circle,
            draw = black,
            thick,
            fill = white,
            minimum size = 7mm
        ]
        
         \tikzstyle{f}=[
            shape=circle,
            minimum size = 6mm
        ]

        \node[p] (0) [label=right:$\qquad l_0 \leftarrow \epsilon \ll \min_i s_i$] {$0$};
        \node[p] (1) [below left of=0]{$1$};
        \node[p] (2) [right of=1]{$2$};
        \node[f] (3) [right of=2]{$\ldots$};
        \node[p] (n) [right of=3]{$N$};
        
        \path[->] (0) edge (1);
        \path[->] (0) edge (2);
        \path[->] (0) edge (n);
        
        \node[align=left] at (-2,-0.2) {$l_i \leftarrow s_i$};
        \node[align=left] at (-2,-0.5) {$c_i \leftarrow p_i + 2 s_i$};

    \end{tikzpicture}
    \vspace{10pt}
\caption{OMP DAG for Knapsack reduction.}
\label{fig:knapRed}
\end{figure}
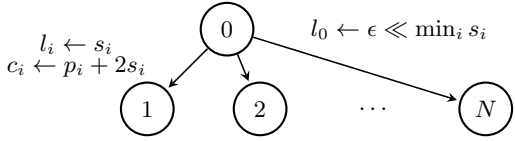

For an instance of Knapsack, we construct a simple \wf DAG $W$ as shown in Figure~\ref{fig:knapRed}.
For each item $i$ in Knapsack, 
we construct an output node $n_i$ with $l_i = s_i$ and $c_i = p_i + 2 s_i$.
We add an input node $n_0$ with $l_0 = \epsilon < \min s_i$ that all output nodes depend on.
Let $Y_i \in \{0, 1\}$ indicate whether a node $n_i \in M$ 
in the optimal solution to OMP in Eq (\ref{eqn:omp})
and $X_i \in \{0, 1\}$ indicate whether an item is picked in the Knapsack problem.
We use $B$ as the storage budget, i.e., $\sum_{i \in \in \{0, 1\}} Y_i l_i \leq B$.
\begin{theorem}
We obtain an optimal solution to the Knapsack problem for $X_i = Y_i \ \ \forall i \in \{1, 2, \ldots, n\}$.
\end{theorem}
\begin{proof} 
First, we observe that for each $n_i$, $T^*(W)$ will pick $\min(l_i, c_i)$ 
given the flat structure of the DAG. 
By construction, $\min(l_i, c_i) = l_i$ in our reduction.
Second, materializing $n_i$ helps in the first iteration 
only when it is loaded in the second iteration.
Thus, we can rewrite Eq (\ref{eqn:omp}) as 
\begin{equation} \label{eqn:ipMat}
\argmin\limits_{\mathbf{Y} \in \{0, 1\}^N} \sum\limits_{i = 1}^N Y_i l_i + \left( \sum\limits_{i = 1}^N Y_i l_i + (1 - Y_i)c_i \right)
\end{equation}
where $\mathbf{Y} = (Y_1, Y_2, \ldots, Y_N)$.
Substituting in our choices of $l_i$ and $c_i$ in terms of $p_i$ and $s_i$ in (\ref{eqn:ipMat}), 
we obtain $\argmin_{\mathbf{Y} \in \{0, 1\}^N} \sum_{i = 1}^N - Y_i p_i$.
Clearly, satisfying the storage constraint 
also satisfies the budget constraint in Knapsack by construction.
Thus, the optimal solution to OMP as constructed gives the optimal solution to Knapsack.
\end{proof}

%% file: semantics.tex
\begin{table*}
\centering
\small
\begin{tabular}{|C{1.8cm}|C{4cm}|C{5cm}|C{5cm}|}
\hline
\textbf{Phrase} & \textbf{Usage} & \textbf{Operation} & \textbf{Example} \\
\hline
\code{refers\_to} & \textit{string} \code{refers\_to} \textit{\name object} & 
Register a \textit{\name object} to a \textit{string} name & 
\code{``ext1'' refers\_to Extractor(...)} \\
\hline
\code{is\_read\_into ... using} & $DC_i[SU]$ \code{is\_read\_into} $DC_j[SU]$ \code{using} \textit{scanner} & 
Apply \textit{scanner} on $DC_i$ to obtain $DC_j$ & 
\makecell{ \code{``sentence'' is\_read\_into ``word''}\\
\code{using whitespaceTokenizer}} \\
\hline
\code{has\_extractors} & $DC[SU]$ \code{has\_extractors} \textit{extractor+} & Apply extractors to $DC$ & \code{``word'' has\_extractors (``ext1'', ``ext2'')}\\
\hline
\code{on} & \textit{synthesizer/learner/reducer} \code{on} $DC[*]+$ & Apply \textit{synthesizer/learner}  on input DC(s) to produce an output DC[E] &  
\makecell{\code{``match'' on  (``person\_candidate'' ,} \\
\code{``known\_persons'')}} \\
\hline
\multirow{2}{*}{\code{results\_from}} & $DC_i[E]$ \code{results\_from} $DC_j[*]$ [\code{with\_label} \textit{extractor}]& Wrap each element in $DC_i$ in an Example and optionally labels the Examples with the output of \textit{extractor}. & \makecell{\code{``income'' results\_from ``rows''}\\
\code{with\_label ``target''}} \\
\cline{2-4}
 & $DC[E]$/Scalar \code{results\_from} \textit{clause} &  Specify the name for \textit{clause}'s output DC[E]. & 
 \code{``learned'' results\_from ``L'' on ``income'' } \\
\hline
\code{uses} & \textit{synthesizer/learner/reducer} \code{uses} \textit{extractors+} & 
Specify  \textit{synthesizer/learner}'s dependency on the output of \textit{extractors+} 
to prevent pruning or uncaching of intermediate results due to optimization. & \code{``match'' uses (``ext1'', ``ext2'')} \\
\hline
\code{is\_output} & $DC[*]$\textit{/result} \code{is\_output} & Requires $DC$\textit{/result} to be materialized.   
& \code{``learned'' is\_output} \\
\hline
\end{tabular}
\vspace{0.2cm}
\caption{
Usage and functions of key phrases in \lang. 
$DC[A]$ denotes a DC with name $DC$ and elements of type $A \in \{SU, E\}$, 
with $A = *$ indicating both types are legal.
$x$+ indicates that $x$ appears one or more times.
When appearing in the same statement, 
\code{on} takes precedence over \code{results\_from}.
}
\label{tab:semantics}
\end{table*}

%% file: grammar.tex

\setlength{\grammarparsep}{1pt}
\setlength{\grammarindent}{7em}
\begin{figure}[ht]
\begin{grammar}
<var> ::= <string>

<scanner> ::= <var> | <scanner-obj>

<extractor> ::= <var> | <extractor-obj>

<typed-ext> ::= `(' <var> `,' <extractor> `)'

<extractors> ::= `('  <extractor>  \{ `,'  <extractor> \}  `)' 

<typed-exts> ::= `(' <typed-ext> \{`,' <typed-ext>\} `)'

<obj> ::= <data-source> | <scanner-obj> | <extractor-obj> | <learner-obj> | <synthesizer-obj> | <reducer-obj>

<assign> ::= <var> `refers\_to' <obj>

<expr1> ::= <var> `is\_read\_into' <var> `using' <scanner>

<expr2> ::=  <var> `has\_extractors' <extractors>

<list> ::= <var> | `(' <var> `,'  <var> \{ `,' <var> \} `)'

<apply> ::= <var> `on' <list>

<expr3> ::=  <apply> `as\_examples' <var>

<expr4> ::=  <apply> `as\_results' <var>

<expr5> ::= <var> `as\_examples' <var> \\
`with\_labels' <extractor> 

<expr6> ::= <var> `uses' <typed-exts>

<expr7> ::= <var> `is\_output()'

<statement> ::= <assign> | <expr1> | <expr2> | <expr3> | <expr4> | <expr5> | <expr6> | <expr7> | <Scala expr>

<program> ::= `object' <string> `extends Workflow \{' \\
\{ <statement> <line-break> \} \\
`\}'
\end{grammar}
\caption{\name syntax in Extended Backus-Naur Form. 
\textit{<string>} denotes a legal String object in Scala;
\textit{<*-obj>} denotes the correct syntax for instantiating object of type ``*'';
\textit{<Scala expr>} denotes any legal Scala expression.
A \name Workflow can be comprised of any combination of \name and Scala expressions,
a direct benefit of being an embedded DSL.}
\label{fig:syntax}
\end{figure}

%% file: ilp.tex
\topic{Integer Linear Programming Formulation}
\Cref{prob:oep} can be formulated as an integer linear program (ILP)
as follows. First, for each node $n_i \in G$,
introduce binary indicator variables $\jobAi$ and $\jobBi$ defined as follows:
\begin{align*}
\jobAi &= \indic{s(n_i) \ne S_p} \\
\jobBi &= \indic{s(n_i) = S_c}
\end{align*}
That is, $\jobAi = 1$ if node $n_i$ is not pruned, 
and $\jobBi = 1$ if node $n_i$ is computed. 
Note that it is not possible to have $\jobAi = 0$ and $\jobBi = 1$. 
Also note that these variables uniquely determine node $n_i$'s state $s(n_i)$.

With the $\{\jobAi\}$ and $\{\jobBi\}$ thus defined, our ILP is as follows:
\begin{mini!}
    {\jobAi,\jobBi}{\sum_{i=1}^{|N|} \jobAi l_i + \jobBi(c_i-l_i)}
    {}{} \label{eqn:obj}
    \addConstraint{\jobAi-\jobBi}{\geq 0,\quad}
    {1\leq i\leq|N|}
    \label{cons:statefunc}
    \addConstraint{\sum_{n_j \in \text{Pa}(n_i)}\jobAj-\jobBi}{\geq 0,\quad}
    {1\leq i\leq|N|}
    \label{cons:pselect}
    \addConstraint{\jobAi, \jobBi\in\ }{\{0, 1\},\ }
    {1\leq i\leq|N|}
    \label{cons:binary}
\end{mini!}
\Cref{cons:statefunc} prevents the assignment 
$\jobAi = 0$ ($n_i$ is pruned) and $\jobBi = 1$ ($n_i$ is computed),
since a pruned node cannot also be computed by definition.
\Cref{cons:pselect} is equivalent to
\Cref{constraint:exstate} --- if $\jobBi = 1$ ($n_i$ is computed),
any parent $n_j$ of $n_i$ must not be pruned, i.e., $\jobAj = 1$,
in order for the sum to be nonnegative.
\Cref{cons:binary} requires the solution to be integers.

This formulation does not specify a constraint 
corresponding to Constraint~\ref{constraint:rerun}.
Instead, we enforce Constraint~\ref{constraint:rerun} 
by setting the load and compute costs of nodes that need to be recomputed to specific values,
as inputs to Problem~\ref{prob:oep}.
Specifically, we set the load cost to $\infty$
and the compute cost to $-\epsilon$ for a small $\epsilon > 0$.
With these values, the cost of a node in $S_l, S_p, S_c$ are $\infty, 0, -\epsilon$ respectively,
which makes $S_c$ a clear choice for minimizing Eq(\ref{eqn:obj}).

Although ILPs are, in general, NP-Hard, the astute reader may
notice that the constraint matrix associated with the above
optimization problem is {\em totally unimodular} (TU), which means
that an optimal solution for the LP-relaxation (which removes constraint~\ref{cons:binary}
in the problem above) assigns integral values to $\{\jobAi\}$ and $\{\jobBi\}$,
indicating that it is both optimal and feasible for the problem above as well~\cite{schrijver1998theory}.
In fact, it turns out that the above problem is the dual of a flow problem; specifically, it
is a minimum cut problem~\cite{yannakakis1985class,hochbaum2000performance}.
This motivates the reduction introduced in Section~\ref{sec:oep}.

%% file: p390-xin.bbl
\begin{thebibliography}{10}

\bibitem{ddRepo}
Deepdive census example.
\newblock
  \url{https://github.com/HazyResearch/deepdive/tree/master/examples/census}.

\bibitem{sklearnUG}
Scikit-learn user guide.
\newblock \url{http://scikit-learn.org/stable/user_guide.html}.

\bibitem{michelangelo}
Meet michelangelo: Uber’s machine learning platform, 2017 (accessed October
  8, 2018).
\newblock https://eng.uber.com/michelangelo/.

\bibitem{sagemaker}
{\em Amazon SageMaker Developer Guide}, 2018 (accessed October 8, 2018).
\newblock https://docs.aws.amazon.com/sagemaker/latest/dg/sagemaker-dg.pdf.

\bibitem{abadi2016tensorflow}
M.~Abadi, A.~Agarwal, P.~Barham, E.~Brevdo, Z.~Chen, C.~Citro, G.~S. Corrado,
  A.~Davis, J.~Dean, M.~Devin, et~al.
\newblock Tensorflow: Large-scale machine learning on heterogeneous distributed
  systems.
\newblock {\em arXiv preprint arXiv:1603.04467}, 2016.

\bibitem{anderson2013brainwash}
M.~R. Anderson, D.~Antenucci, V.~Bittorf, M.~Burgess, M.~J. Cafarella,
  A.~Kumar, F.~Niu, Y.~Park, C.~R{\'e}, and C.~Zhang.
\newblock Brainwash: A data system for feature engineering.
\newblock In {\em CIDR}, 2013.

\bibitem{anderson2016input}
M.~R. Anderson and M.~Cafarella.
\newblock Input selection for fast feature engineering.
\newblock In {\em Data Engineering (ICDE), 2016 IEEE 32nd International
  Conference on}, pages 577--588. IEEE, 2016.

\bibitem{armbrust2015sparksql}
M.~Armbrust, R.~S. Xin, C.~Lian, Y.~Huai, D.~Liu, J.~K. Bradley, X.~Meng,
  T.~Kaftan, M.~J. Franklin, A.~Ghodsi, et~al.
\newblock Spark sql: Relational data processing in spark.
\newblock In {\em Proceedings of the 2015 ACM SIGMOD International Conference
  on Management of Data}, pages 1383--1394. ACM, 2015.

\bibitem{barnes2015azure}
J.~Barnes.
\newblock Azure machine learning microsoft azure essentials, 2015.

\bibitem{baylor2017tfx}
D.~Baylor, E.~Breck, H.-T. Cheng, N.~Fiedel, C.~Y. Foo, Z.~Haque, S.~Haykal,
  M.~Ispir, V.~Jain, L.~Koc, et~al.
\newblock Tfx: A tensorflow-based production-scale machine learning platform.
\newblock In {\em Proceedings of the 23rd ACM SIGKDD International Conference
  on Knowledge Discovery and Data Mining}, pages 1387--1395. ACM, 2017.

\bibitem{airflow}
M.~Beauchemin.
\newblock Airflow: a workflow management platform, 2015 (accessed October 8,
  2018).
\newblock
  https://medium.com/airbnb-engineering/airflow-a-workflow-management-platform-46318b977fd8.

\bibitem{boehm2016declarative}
M.~Boehm, A.~V. Evfimievski, N.~Pansare, and B.~Reinwald.
\newblock Declarative machine learning-a classification of basic properties and
  types.
\newblock {\em arXiv preprint arXiv:1605.05826}, 2016.

\bibitem{sklearnApi}
L.~Buitinck, G.~Louppe, M.~Blondel, F.~Pedregosa, A.~Mueller, O.~Grisel,
  V.~Niculae, P.~Prettenhofer, A.~Gramfort, J.~Grobler, R.~Layton,
  J.~VanderPlas, A.~Joly, B.~Holt, and G.~Varoquaux.
\newblock {API} design for machine learning software: experiences from the
  scikit-learn project.
\newblock In {\em ECML PKDD Workshop: Languages for Data Mining and Machine
  Learning}, pages 108--122, 2013.

\bibitem{callahan2006vistrails}
S.~P. Callahan, J.~Freire, E.~Santos, C.~E. Scheidegger, C.~T. Silva, and H.~T.
  Vo.
\newblock Vistrails: visualization meets data management.
\newblock In {\em Proceedings of the 2006 ACM SIGMOD international conference
  on Management of data}, pages 745--747. ACM, 2006.

\bibitem{chen2017towards}
L.~Chen, A.~Kumar, J.~Naughton, and J.~M. Patel.
\newblock Towards linear algebra over normalized data.
\newblock {\em Proceedings of the VLDB Endowment}, 10(11):1214--1225, 2017.

\bibitem{chirkova2012materialized}
R.~Chirkova, J.~Yang, et~al.
\newblock Materialized views.
\newblock {\em Foundations and Trends{\textregistered} in Databases},
  4(4):295--405, 2012.

\bibitem{crotty2015architecture}
A.~Crotty, A.~Galakatos, K.~Dursun, T.~Kraska, C.~Binnig, U.~Cetintemel, and
  S.~Zdonik.
\newblock An architecture for compiling udf-centric workflows.
\newblock {\em Proceedings of the VLDB Endowment}, 8(12):1466--1477, 2015.

\bibitem{damien2016tflearn}
A.~Damien et~al.
\newblock Tflearn, 2016.

\bibitem{deepdive2016}
C.~De~Sa, A.~Ratner, C.~R{\'e}, J.~Shin, F.~Wang, S.~Wu, and C.~Zhang.
\newblock Deepdive: Declarative knowledge base construction.
\newblock {\em SIGMOD Rec.}, 45(1):60--67, June 2016.

\bibitem{deelman2004pegasus}
E.~Deelman, J.~Blythe, Y.~Gil, C.~Kesselman, G.~Mehta, S.~Patil, M.-H. Su,
  K.~Vahi, and M.~Livny.
\newblock Pegasus: Mapping scientific workflows onto the grid.
\newblock In {\em Grid Computing}, pages 11--20. Springer, 2004.

\bibitem{Dua:2017}
D.~Dheeru and E.~Karra~Taniskidou.
\newblock {UCI} machine learning repository, 2017.

\bibitem{fblearner}
J.~Dunn.
\newblock Introducing fblearner flow: Facebook’s ai backbone, 2018 (accessed
  October 8, 2018).
\newblock
  https://code.fb.com/core-data/introducing-fblearner-flow-facebook-s-ai-backbone/.

\bibitem{edmonds1972theoretical}
J.~Edmonds and R.~M. Karp.
\newblock Theoretical improvements in algorithmic efficiency for network flow
  problems.
\newblock {\em Journal of the ACM (JACM)}, 19(2):248--264, 1972.

\bibitem{elghandour2012restore}
I.~Elghandour and A.~Aboulnaga.
\newblock Restore: reusing results of mapreduce jobs.
\newblock {\em Proceedings of the VLDB Endowment}, 5(6):586--597, 2012.

\bibitem{feng2012bismarck}
X.~Feng, A.~Kumar, B.~Recht, and C.~R{\'e}.
\newblock Towards a unified architecture for in-rdbms analytics.
\newblock In {\em Proceedings of the 2012 ACM SIGMOD International Conference
  on Management of Data}, pages 325--336. ACM, 2012.

\bibitem{ghoting2011systemml}
A.~Ghoting, R.~Krishnamurthy, E.~Pednault, B.~Reinwald, V.~Sindhwani,
  S.~Tatikonda, Y.~Tian, and S.~Vaithyanathan.
\newblock Systemml: Declarative machine learning on mapreduce.
\newblock In {\em 2011 IEEE 27th International Conference on Data Engineering},
  pages 231--242. IEEE, 2011.

\bibitem{gordon1995tutorial}
A.~D. Gordon.
\newblock A tutorial on co-induction and functional programming.
\newblock In {\em Functional Programming, Glasgow 1994}, pages 78--95.
  Springer, 1995.

\bibitem{gunda2010nectar}
P.~K. Gunda, L.~Ravindranath, C.~A. Thekkath, Y.~Yu, and L.~Zhuang.
\newblock Nectar: Automatic management of data and computation in datacenters.
\newblock In {\em OSDI}, volume~10, pages 1--8, 2010.

\bibitem{hall2009weka}
M.~Hall, E.~Frank, G.~Holmes, B.~Pfahringer, P.~Reutemann, and I.~H. Witten.
\newblock The weka data mining software: an update.
\newblock {\em ACM SIGKDD explorations newsletter}, 11(1):10--18, 2009.

\bibitem{hellerstein2012madlib}
J.~M. Hellerstein, C.~R{\'e}, F.~Schoppmann, D.~Z. Wang, E.~Fratkin,
  A.~Gorajek, K.~S. Ng, C.~Welton, X.~Feng, K.~Li, et~al.
\newblock The madlib analytics library: or mad skills, the sql.
\newblock {\em Proceedings of the VLDB Endowment}, 5(12):1700--1711, 2012.

\bibitem{hochbaum2000performance}
D.~S. Hochbaum and A.~Chen.
\newblock Performance analysis and best implementations of old and new
  algorithms for the open-pit mining problem.
\newblock {\em Operations Research}, 48(6):894--914, 2000.

\bibitem{jindal2018selecting}
A.~Jindal, K.~Karanasos, S.~Rao, and H.~Patel.
\newblock Selecting subexpressions to materialize at datacenter scale.
\newblock {\em Proceedings of the VLDB Endowment}, 11(7):800--812, 2018.

\bibitem{karpathy2015deep}
A.~Karpathy and L.~Fei-Fei.
\newblock Deep visual-semantic alignments for generating image descriptions.
\newblock In {\em Proceedings of the IEEE Conference on Computer Vision and
  Pattern Recognition}, pages 3128--3137, 2015.

\bibitem{kleinberg2006algorithm}
J.~Kleinberg and E.~Tardos.
\newblock {\em Algorithm design}.
\newblock Pearson Education, 2006.

\bibitem{kohavi1996scaling}
R.~Kohavi.
\newblock Scaling up the accuracy of naive-bayes classifiers: a decision-tree
  hybrid.
\newblock In {\em Proceedings of the Second International Conference on
  Knowledge Discovery and Data Mining}, pages 202--207. AAAI Press, 1996.

\bibitem{krishnan2016activeclean}
S.~Krishnan, J.~Wang, E.~Wu, M.~J. Franklin, and K.~Goldberg.
\newblock Activeclean: Interactive data cleaning while learning convex loss
  models.
\newblock {\em arXiv preprint arXiv:1601.03797}, 2016.

\bibitem{kumar2015learning}
A.~Kumar, J.~Naughton, and J.~M. Patel.
\newblock Learning generalized linear models over normalized data.
\newblock In {\em Proceedings of the 2015 ACM SIGMOD International Conference
  on Management of Data}, pages 1969--1984. ACM, 2015.

\bibitem{kumar2016join}
A.~Kumar, J.~Naughton, J.~M. Patel, and X.~Zhu.
\newblock To join or not to join?: Thinking twice about joins before feature
  selection.
\newblock In {\em Proceedings of the 2016 International Conference on
  Management of Data}, pages 19--34. ACM, 2016.

\bibitem{langford2007vowpal}
J.~Langford, L.~Li, and A.~Strehl.
\newblock Vowpal wabbit online learning project, 2007.

\bibitem{lecun1998gradient}
Y.~LeCun, L.~Bottou, Y.~Bengio, and P.~Haffner.
\newblock Gradient-based learning applied to document recognition.
\newblock {\em Proceedings of the IEEE}, 86(11):2278--2324, 1998.

\bibitem{low2012distributed}
Y.~Low, D.~Bickson, J.~Gonzalez, C.~Guestrin, A.~Kyrola, and J.~M. Hellerstein.
\newblock Distributed graphlab: a framework for machine learning and data
  mining in the cloud.
\newblock {\em Proceedings of the VLDB Endowment}, 5(8):716--727, 2012.

\bibitem{ludascher2006scientific}
B.~Lud{\"a}scher, I.~Altintas, C.~Berkley, D.~Higgins, E.~Jaeger, M.~Jones,
  E.~A. Lee, J.~Tao, and Y.~Zhao.
\newblock Scientific workflow management and the kepler system.
\newblock {\em Concurrency and Computation: Practice and Experience},
  18(10):1039--1065, 2006.

\bibitem{manning2014stanford}
C.~D. Manning, M.~Surdeanu, J.~Bauer, J.~R. Finkel, S.~Bethard, and
  D.~McClosky.
\newblock The stanford corenlp natural language processing toolkit.
\newblock In {\em ACL (System Demonstrations)}, pages 55--60, 2014.

\bibitem{linq}
E.~Meijer, B.~Beckman, and G.~Bierman.
\newblock Linq: Reconciling object, relations and xml in the .net framework.
\newblock In {\em Proceedings of the 2006 ACM SIGMOD International Conference
  on Management of Data}, SIGMOD '06, pages 706--706, New York, NY, USA, 2006.
  ACM.

\bibitem{meng2016mllib}
X.~Meng, J.~Bradley, E.~Sparks, S.~Venkataraman, D.~Liu, J.~Freeman, D.~Tsai,
  M.~Amde, S.~Owen, et~al.
\newblock Mllib: Machine learning in apache spark.
\newblock 2016.

\bibitem{mikolov2013distributed}
T.~Mikolov, I.~Sutskever, K.~Chen, G.~S. Corrado, and J.~Dean.
\newblock Distributed representations of words and phrases and their
  compositionality.
\newblock In {\em Advances in neural information processing systems}, pages
  3111--3119, 2013.

\bibitem{munson2012study}
M.~A. Munson.
\newblock A study on the importance of and time spent on different modeling
  steps.
\newblock {\em ACM SIGKDD Explorations Newsletter}, 13(2):65--71, 2012.

\bibitem{olston2008pig}
C.~Olston, B.~Reed, U.~Srivastava, R.~Kumar, and A.~Tomkins.
\newblock Pig latin: a not-so-foreign language for data processing.
\newblock In {\em Proceedings of the 2008 ACM SIGMOD international conference
  on Management of data}, pages 1099--1110. ACM, 2008.

\bibitem{olteanu2016f}
D.~Olteanu and M.~Schleich.
\newblock F: regression models over factorized views.
\newblock {\em Proceedings of the VLDB Endowment}, 9(13):1573--1576, 2016.

\bibitem{owen2012mahout}
S.~Owen, R.~Anil, T.~Dunning, and E.~Friedman.
\newblock Mahout in action.
\newblock 2012.

\bibitem{palkar2017weld}
S.~Palkar, J.~J. Thomas, A.~Shanbhag, D.~Narayanan, H.~Pirk, M.~Schwarzkopf,
  S.~Amarasinghe, M.~Zaharia, and S.~InfoLab.
\newblock Weld: A common runtime for high performance data analytics.
\newblock In {\em Conference on Innovative Data Systems Research (CIDR)}, 2017.

\bibitem{paszke2017pytorch}
A.~Paszke, S.~Gross, S.~Chintala, and G.~Chanan.
\newblock Pytorch: Tensors and dynamic neural networks in python with strong
  gpu acceleration, 2017.

\bibitem{patel2010gestalt}
K.~Patel, N.~Bancroft, S.~M. Drucker, J.~Fogarty, A.~J. Ko, and J.~Landay.
\newblock Gestalt: integrated support for implementation and analysis in
  machine learning.
\newblock In {\em Proceedings of the 23nd annual ACM symposium on User
  interface software and technology}, pages 37--46. ACM, 2010.

\bibitem{pedregosa2011scikit}
F.~Pedregosa, G.~Varoquaux, A.~Gramfort, V.~Michel, B.~Thirion, O.~Grisel,
  M.~Blondel, P.~Prettenhofer, R.~Weiss, V.~Dubourg, et~al.
\newblock Scikit-learn: Machine learning in python.
\newblock {\em Journal of Machine Learning Research}, 12(Oct):2825--2830, 2011.

\bibitem{perez2014history}
L.~L. Perez and C.~M. Jermaine.
\newblock History-aware query optimization with materialized intermediate
  views.
\newblock In {\em Data Engineering (ICDE), 2014 IEEE 30th International
  Conference on}, pages 520--531. IEEE, 2014.

\bibitem{pitts1997operationally}
A.~M. Pitts.
\newblock Operationally-based theories of program equivalence.
\newblock {\em Semantics and Logics of Computation}, 14:241, 1997.

\bibitem{rasband2012imagej}
W.~Rasband.
\newblock Imagej: Image processing and analysis in java.
\newblock {\em Astrophysics Source Code Library}, 2012.

\bibitem{ratner2017snorkel}
A.~Ratner, S.~H. Bach, H.~Ehrenberg, J.~Fries, S.~Wu, and C.~R{\'e}.
\newblock Snorkel: Rapid training data creation with weak supervision.
\newblock {\em arXiv preprint arXiv:1711.10160}, 2017.

\bibitem{recht2011hogwild}
B.~Recht, C.~Re, S.~Wright, and F.~Niu.
\newblock Hogwild: A lock-free approach to parallelizing stochastic gradient
  descent.
\newblock In {\em Advances in neural information processing systems}, pages
  693--701, 2011.

\bibitem{ren2017life}
X.~Ren, J.~Shen, M.~Qu, X.~Wang, Z.~Wu, Q.~Zhu, M.~Jiang, F.~Tao, S.~Sinha,
  D.~Liem, et~al.
\newblock Life-inet: A structured network-based knowledge exploration and
  analytics system for life sciences.
\newblock {\em Proceedings of ACL 2017, System Demonstrations}, pages 55--60,
  2017.

\bibitem{rice1953classes}
H.~G. Rice.
\newblock Classes of recursively enumerable sets and their decision problems.
\newblock {\em Transactions of the American Mathematical Society},
  74(2):358--366, 1953.

\bibitem{rosenpyspark}
J.~Rosen.
\newblock Pyspark internals.

\bibitem{schrijver1998theory}
A.~Schrijver.
\newblock {\em Theory of linear and integer programming}.
\newblock John Wiley \& Sons, 1998.

\bibitem{sparks2016end}
E.~R. Sparks, S.~Venkataraman, T.~Kaftan, M.~J. Franklin, and B.~Recht.
\newblock Keystoneml: Optimizing pipelines for large-scale advanced analytics.
\newblock In {\em Data Engineering (ICDE), 2017 IEEE 33rd International
  Conference on}, pages 535--546. IEEE, 2017.

\bibitem{stonebraker2011architecture}
M.~Stonebraker, P.~Brown, A.~Poliakov, and S.~Raman.
\newblock The architecture of scidb.
\newblock In {\em International Conference on Scientific and Statistical
  Database Management}, pages 1--16. Springer, 2011.

\bibitem{sujeeth2011optiml}
A.~Sujeeth, H.~Lee, K.~Brown, T.~Rompf, H.~Chafi, M.~Wu, A.~Atreya, M.~Odersky,
  and K.~Olukotun.
\newblock Optiml: an implicitly parallel domain-specific language for machine
  learning.
\newblock In {\em Proceedings of the 28th International Conference on Machine
  Learning (ICML-11)}, pages 609--616, 2011.

\bibitem{sumbaly2013big}
R.~Sumbaly, J.~Kreps, and S.~Shah.
\newblock The big data ecosystem at linkedin.
\newblock In {\em Proceedings of the 2013 ACM SIGMOD International Conference
  on Management of Data}, pages 1125--1134. ACM, 2013.

\bibitem{tang2015line}
J.~Tang, M.~Qu, M.~Wang, M.~Zhang, J.~Yan, and Q.~Mei.
\newblock Line: Large-scale information network embedding.
\newblock In {\em Proceedings of the 24th International Conference on World
  Wide Web}, pages 1067--1077. International World Wide Web Conferences
  Steering Committee, 2015.

\bibitem{team2016deeplearning4j}
D.~Team.
\newblock Deeplearning4j: Open-source distributed deep learning for the jvm.
\newblock {\em Apache Software Foundation License}, 2, 2016.

\bibitem{dd2deeplearning4j}
D.~Team et~al.
\newblock Deeplearning4j: Open-source distributed deep learning for the jvm.
\newblock {\em Apache Software Foundation License}, 2.

\bibitem{vartak2018mistique}
M.~Vartak, J.~M. da~Trindade, S.~Madden, and M.~Zaharia.
\newblock Mistique: A system to store and query model intermediates for model
  diagnosis.
\newblock In {\em Proceedings of the 2018 ACM International Conference on
  Management of Data}, 2018.

\bibitem{vartak2015supporting}
M.~Vartak, P.~Ortiz, K.~Siegel, H.~Subramanyam, S.~Madden, and M.~Zaharia.
\newblock Supporting fast iteration in model building.
\newblock In {\em NIPS Workshop LearningSys}, 2015.

\bibitem{wang2011hybrid}
D.~Z. Wang, M.~J. Franklin, M.~Garofalakis, J.~M. Hellerstein, and M.~L. Wick.
\newblock Hybrid in-database inference for declarative information extraction.
\newblock In {\em Proceedings of the 2011 ACM SIGMOD International Conference
  on Management of data}, pages 517--528. ACM, 2011.

\bibitem{weimer2011machine}
M.~Weimer, T.~Condie, R.~Ramakrishnan, et~al.
\newblock Machine learning in scalops, a higher order cloud computing language.
\newblock In {\em NIPS 2011 Workshop on parallel and large-scale machine
  learning (BigLearn)}, volume~9, pages 389--396, 2011.

\bibitem{woodcock2009formal}
J.~Woodcock, P.~G. Larsen, J.~Bicarregui, and J.~Fitzgerald.
\newblock Formal methods: Practice and experience.
\newblock {\em ACM computing surveys (CSUR)}, 41(4):19, 2009.

\bibitem{deem}
D.~Xin, L.~Ma, J.~Liu, S.~Macke, S.~Song, and A.~Parameswaran.
\newblock Accelerating human-in-the-loop machine learning: Challenges and
  opportunities (vision paper).
\newblock In {\em Proceedings of the Second Workshop on Data Management for
  End-To-End Machine Learning}, DEEM'18. ACM, 2018.

\bibitem{demo}
D.~Xin, L.~Ma, J.~Liu, S.~Macke, S.~Song, and A.~Parameswaran.
\newblock Helix: Accelerating human-in-the-loop machine learning (demo paper).
\newblock {\em Proceedings of the VLDB Endowment}, 2018.

\bibitem{xin2018developers}
D.~Xin, L.~Ma, S.~Song, and A.~Parameswaran.
\newblock How developers iterate on machine learning workflows--a survey of the
  applied machine learning literature.
\newblock {\em KDD IDEA Workshop}, 2018.

\bibitem{dorx2017}
D.~Xin, S.~Macke, L.~Ma, R.~Ma, S.~Song, and A.~Parameswaran.
\newblock Helix: Holistic optimization for accelerating iterative machine
  learning.
\newblock {\em Technical Report
  http://data-people.cs.illinois.edu/helix-tr.pdf}, 2018.

\bibitem{yannakakis1985class}
M.~Yannakakis.
\newblock On a class of totally unimodular matrices.
\newblock {\em Mathematics of Operations Research}, 10(2):280--304, 1985.

\bibitem{yu2005taxonomy}
J.~Yu and R.~Buyya.
\newblock A taxonomy of workflow management systems for grid computing.
\newblock {\em Journal of Grid Computing}, 3(3-4):171--200, 2005.

\bibitem{yu2008dryadlinq}
Y.~Yu, M.~Isard, D.~Fetterly, M.~Budiu, {\'U}.~Erlingsson, P.~K. Gunda, and
  J.~Currey.
\newblock Dryadlinq: A system for general-purpose distributed data-parallel
  computing using a high-level language.
\newblock In {\em OSDI}, volume~8, pages 1--14, 2008.

\bibitem{mlflow}
M.~Zaharia.
\newblock Introducing mlflow: an open source machine learning platform, 2018
  (accessed October 8, 2018).
\newblock
  https://databricks.com/blog/2018/06/05/introducing-mlflow-an-open-source-machine-learning-platform.html.

\bibitem{zaharia2012resilient}
M.~Zaharia, M.~Chowdhury, T.~Das, A.~Dave, J.~Ma, M.~McCauley, M.~J. Franklin,
  S.~Shenker, and I.~Stoica.
\newblock Resilient distributed datasets: A fault-tolerant abstraction for
  in-memory cluster computing.
\newblock In {\em Proceedings of the 9th USENIX conference on Networked Systems
  Design and Implementation}. USENIX Association, 2012.

\bibitem{zhang2015deepdive}
C.~Zhang.
\newblock {\em DeepDive: A data management system for automatic knowledge base
  construction}.
\newblock PhD thesis, Citeseer, 2015.

\bibitem{zhang2016materialization}
C.~Zhang, A.~Kumar, and C.~R{\'e}.
\newblock Materialization optimizations for feature selection workloads.
\newblock {\em ACM Transactions on Database Systems (TODS)}, 41(1):2, 2016.

\bibitem{zinkevich2010parallelized}
M.~Zinkevich, M.~Weimer, L.~Li, and A.~J. Smola.
\newblock Parallelized stochastic gradient descent.
\newblock In {\em Advances in neural information processing systems}, pages
  2595--2603, 2010.

\end{thebibliography}
